\documentclass[aps,prl,twocolumn,superscriptaddress,floatfix]{revtex4-2}

\usepackage{float}
\usepackage{amsthm}
\usepackage{pdfpages}   % 必须保留
\usepackage{pgffor}     % 用于循环，后面插入补充材料时建议使用，但不是必选项

\makeatletter
\AtBeginDocument{\let\LS@rot\@undefined}
\makeatother
\usepackage[table]{xcolor}
\usepackage{caption}
\usepackage{amsmath}
\usepackage{ragged2e}
\usepackage{booktabs} 
\usepackage[justification=justified,singlelinecheck=false]{caption}
\usepackage{mathtools}
\usepackage{amsmath}
\usepackage{amssymb}
\usepackage{graphicx}
\usepackage{float}
\usepackage{xcolor}
\usepackage{braket}
\usepackage{pdfpages}
\usepackage{tikz}
\usetikzlibrary{shapes, arrows.meta, positioning, decorations.pathmorphing, backgrounds, fit, calc, decorations.pathreplacing}
\usepackage{tcolorbox}

\newtheorem{thm}{Theorem}
\newtheorem{lem}{Lemma}

\begin{document}

\title{On-Chip Quantum Randomness Amplification}

\author{Lang Li}
\email{langli@hku.hk}
\affiliation{QICI Quantum Information and Computation Initiative, School of Computing and Data Science, The University of Hong Kong, Pokfulam Road, Hong Kong, China}% 编号 1
\affiliation{HKU-Oxford Joint Laboratory for Quantum Information and Computation} % 编号 3

\author{Yutian Wu}
\affiliation{QICI Quantum Information and Computation Initiative, School of Computing and Data Science, The University of Hong Kong, Pokfulam Road, Hong Kong, China}
\affiliation{HKU-Oxford Joint Laboratory for Quantum Information and Computation}

\author{Giulio Chiribella}
\email{giulio@cs.hku.hk}
\affiliation{QICI Quantum Information and Computation Initiative, School of Computing and Data Science, The University of Hong Kong, Pokfulam Road, Hong Kong, China}
\affiliation{HKU-Oxford Joint Laboratory for Quantum Information and Computation}
\affiliation{Department of Computer Science, University of Oxford, Wolfson Building, Parks Road, Oxford, UK} % 编号 4
\affiliation{Perimeter Institute for Theoretical Physics, 31 Caroline Street North, Waterloo, Ontario, Canada} % 编号 5 

\author{Ravishankar Ramanathan}
\email{ravi@cs.hku.hk}
\affiliation{QICI Quantum Information and Computation Initiative, School of Computing and Data Science, The University of Hong Kong, Pokfulam Road, Hong Kong, China}
\affiliation{HKU-Oxford Joint Laboratory for Quantum Information and Computation}

\begin{abstract}
Randomness amplification, the task of extracting uniform  private bits from biased seeds that may be partly known by a malicious third party, is of central importance in cryptography.  The highest  security in this task is provided by a class of quantum protocols known as  device-independent,  which however are challenging to integrate into scalable devices. Semi-device-independent (SDI) protocols are a promising alternative that guarantees security under few natural assumptions,  such as bounds on the amount of energy used by the devices.  Here, we provide the first demonstration of SDI randomness amplification on an integrated silicon photonic chip, achieving a throughput rate of $20$ Mbps suitable for practical applications. This rate is achieved through  a novel technique  for  SDI entropy certification, which delivers strictly tighter von Neumann entropy bounds compared to existing methods and remains valid even if  the preparation and measurement devices share quantum correlations.  Overall, the methods developed in this work enable the integration of SDI technology into portable telecom devices, opening up a new generation of quantum cryptographic hardware.
 \end{abstract}
\flushbottom
\maketitle
 
\thispagestyle{empty}

\section*{Introduction}\label{Intr}

\indent Randomness amplification— the task of extracting private and uniform bits from biased and partly compromised seeds — has applications in  cryptographic primitives, as well as in randomized algorithms, scientific computation, and online gaming \cite{Colbeck2012,Pironio2010,Acin2016,Bierhorst2018},  
While randomness cannot be amplified in the world of classical physics, quantum mechanics provides an unprecedented opportunity to generate high-quality, cryptographically secure randomness.   
The highest level of security is achieved by  device-independent (DI) protocols, which   guarantee unpredictability against general attackers without making any assumption on the internal functioning of the devices~\cite{Ramanathan2016,Brandao2016}. Recently,  proof-of-principle demonstrations of  DI randomness amplification have been provided \cite{Kulikov2024,Foreman2023,Ramanathan2020}.  A practical limitation  of DI  protocols, however, is that they have  demanding requirements,  including high-visibility entanglement, loophole-free detection,  stringent noise suppression, and  spacelike separation between two or more measurement stations, which typically results into bulky  setups. 
  These requirements  limit scalability, and tend to result in low generation rates that  significantly restrict  potential applications.  

A promising approach is provided by semi-device-independent (SDI) protocols \cite{Roch2025,Van2017}, which maintain most of the features of DI security while making minimal assumptions on the devices,  such as bounds on the amount of energy used by them \cite{Van2017, Senno2021}. Leveraging this limited knowledge,  SDI protocols can  greatly  enhance the speed of  randomness amplification and remove the need of spacelike separation,  thereby enabling fast and compact setups.    
%enabling a promising route toward scalable integrated implementations 
%Integrated photonics emerges as a transformative platform for quantum technologies, offering compact, stable, and mass-producible systems that integrate state preparation, manipulation, and detection on a single chip~\cite{Wang2020,Pelucchi2022}. 
However,  experimental demonstrations of such fast and compact setups  remain absent. 
%{\color{blue} [why do we mention DI procotols here? We already conclude the part on DI. Now we need to provide the state of the art of SDI implementations.]}. 
Chip-scale realizations are notably lacking, as existing SDI protocols have not yet been adapted to the constraints of integrated photonics. In particular,  miniaturization and high-speed operation demand optimized entropy bounds under realistic noise and loss profiles~\cite{Pelucchi2022, Bartolucci2024,Bogaerts2024}, which however have been lacking so far. \\
%\indent SDI protocols are particularly amenable to this architecture, as their partial trust assumptions align with the controlled environments of photonic waveguides, enabling energy-bounded schemes that leverage photon-number constraints for security~\cite{Roch2025,VanHimbeeck2017}. Such energy-limited frameworks facilitate high-speed randomness amplification without the need for entanglement distribution over long distances, making them ideal for on-chip implementations.  However, realizing these advantages requires operation at an ultralow mean photon number while simultaneously sustaining shot-noise-limited heterodyne detection, ultra-low propagation loss, suppressed free-carrier absorption, and high phase-modulation efficiency, placing stringent constraints on silicon photonic integration\cite{Bartolucci2024,Bogaerts2024}. Moreover, in energy-constrained SDI settings, entropy estimation is inherently suboptimal due to reliance on fixed-node Gauss--Radau spectral methods \cite{Brown2024} for bounding nonlinear functionals such as the von Neumann entropy~\cite{Kessler2020,Shirokov2018}. This rigidity leads to conservative eigenvalue approximations and loose bounds in the restricted state spaces of practical devices, ultimately suppressing extractable randomness and hindering scalability in chip-scale systems.\\
\indent In this paper,  we report the first implementation of SDI randomness amplification on a fully integrated silicon photonic chip.  The chip has  dimensions of $3~\mathrm{mm}\times11~\mathrm{mm}\times11~\mathrm{mm}$, incorporates quantum state preparation and measurement, and achieves amplification for sources with a Santha-Vazirani (SV) structure \cite{Santha1986} in which each bit can deviate from the uniform distribution by an amount up to $\varepsilon \leq 0.3$. As a typical example, for an initial bit string  of $10^{11}$ bits with a practically relevant  bias of $\varepsilon=0.12$, our setup  produces an output string of length $1.8\times10^{9}$ that is $\varepsilon_{\text{sec}}\le 10^{-12}$-close to uniform, at an overall rate of $20$ Mbps, suitable for applications such as  VPNs, digital signatures,  and gaming. 

The  high rate of our setup is enabled not only by advances in the photonic hardware, but also by a novel technique for certifying the amount of randomness present in the output bits. This technique, called the variational Gauss-Radau method, offers  tighter bounds on the single-round von Neumann entropy of the output bits,  and works even in the presence of prior correlations between the state preparation and the measurement devices.  Overall, our results establish a scalable integrated photonic platform for quantum randomness amplification, opening up a route towards secure SDI  randomness generation  in portable devices such as laptops and mobile phones.
%, and, more generally, tasks that require protecting sensitive data-in-transit. 

\begin{figure*}[t]
\centering
\begin{tikzpicture}[
    adversary/.style={rectangle, fill=orange!80!black, rounded corners=4pt, minimum width=1.6cm, minimum height=0.8cm, text=white, font=\bfseries\Large},
    svband/.style={rectangle, fill=cyan!10, draw=cyan!30, thick, rounded corners=4pt, minimum width=1.4cm, minimum height=8.0cm, align=center, font=\ttfamily\small, text=black!80},
    % Extractor box made taller (8.0cm) but kept same width (1.8cm)
    extractor/.style={rectangle, fill=orange!15, draw=orange!40, thick, rounded corners=4pt, minimum width=1.8cm, minimum height=8.0cm, align=center, font=\bfseries\sffamily},
    % Vertically-stacked outband style (Removed background color)
    outband_v/.style={rectangle, draw=none, thick, rounded corners=4pt, minimum width=1.4cm, minimum height=8.0cm, align=center, font=\ttfamily, text=black!80},
    arrow/.style={-{Stealth[scale=1.2]}, thick, draw=black!80},
    qchannel/.style={-{Stealth[scale=1.5]}, line width=2pt, draw=blue!70!black}, 
    entanglement/.style={decorate, decoration={snake, amplitude=1mm, segment length=4mm}, line width=1.5pt, orange!90},
    sideinfo/.style={dashed, thick, orange!90, -{Stealth[scale=1.2]}},
    constraint/.style={rectangle, fill=white, draw=blue!70!black, thick, rounded corners=3pt, inner sep=4pt, font=\bfseries\color{blue!80!black}}
]

    % ==========================================
    % 3D BOX MACRO (Recalibrated for Wider Boxes)
    % ==========================================
    \newcommand{\drawthreedbox}[4]{
        % #1: x, #2: y, #3: node name, #4: label text
        % Top face (Widened to match 2.4cm front face)
        \filldraw[fill=black!60, draw=black!90, thick, line join=round] (#1-1.2, #2+0.6) -- (#1-0.8, #2+1.0) -- (#1+1.6, #2+1.0) -- (#1+1.2, #2+0.6) -- cycle;
        % Right Side face
        \filldraw[fill=black!85, draw=black!90, thick, line join=round] (#1+1.2, #2+0.6) -- (#1+1.6, #2+1.0) -- (#1+1.6, #2-0.2) -- (#1+1.2, #2-0.6) -- cycle;
        % Front face (Widened to 2.4cm so "Meas (M)" fits perfectly)
        \node[minimum width=2.4cm, minimum height=1.2cm, fill=black!75, draw=black!90, thick, text=white, font=\bfseries\Large, inner sep=0pt] (#3) at (#1,#2) {#4};
        % Anchor coordinates for easy routing
        \coordinate (#3_top) at (#1, #2+1.0);
        \coordinate (#3_bottom) at (#1, #2-0.6);
        \coordinate (#3_right) at (#1+1.6, #2+0.2);
        \coordinate (#3_left) at (#1-1.2, #2); 
    }

    % ==========================================
    % 1. LEFT COLUMN: SV SOURCE
    % ==========================================
    \node[svband] (sv) at (0, -0.5) {\vdots \\ 0\\1\\1\\1\\0\\1\\0\\[0.2cm] \vdots \\ \vdots \\[0.2cm]0\\1\\1\\1 \\ \vdots \\ \vdots};
   \node[
    above=0.2cm of sv,
    font=\sffamily\normalsize,
    align=center,
    text=black!90
] (svlabel)
{$\varepsilon$-biased\\SV source};

    % ==========================================
\node[
    adversary,
    minimum width=2.0cm,
    minimum height=1.3cm,
    font=\bfseries\LARGE,
    align=center
] (eve) at (-2.8,4.2)
{
Eve\\[-2mm]
{\Large\color{white}$\lambda$}
};

% Bracket 1: Eve's part (top)
\draw[decorate, decoration={brace, amplitude=5pt}, thick, draw=black!80]
    (-0.8, 2.2) -- (-0.8, 3.4)
    coordinate[midway, xshift=-6pt] (eve_brace);

% lambda entangled with the SV source
\draw[entanglement, dashed]
    (eve.south) to[bend right=25] (eve_brace);
    % Bracket 2: Input X (middle)
    \draw[decorate, decoration={brace, amplitude=6pt}, thick, draw=black!80] 
        (-0.8, -0.8) -- (-0.8, 2.0) node[midway, left=0.2cm, font=\bfseries] {Input $\mathbf{X}$};

    % Bracket 3: Seed Z (bottom)
    \draw[decorate, decoration={brace, amplitude=6pt}, thick, draw=black!80] 
        (-0.8, -4.4) -- (-0.8, -1.0) node[midway, left=0.2cm, font=\bfseries] {Seed $\mathbf{Z}$};

    % ==========================================
    % 3. UNTRUSTED DEVICES (3D P and M)
    % ==========================================
    \drawthreedbox{4.4}{2.0}{prep}{Prep (\textbf{P})}
    \drawthreedbox{4.4}{-1.5}{meas}{Meas (\textbf{M})}

    % Potential Pre-shared Entanglement (\sigma_PM)
    % Routed cleanly using the macro's right-side anchor
    \draw[entanglement, dashed] (prep_right) to[bend left=45] node[right=0.1cm, font=\bfseries, text=orange!90!black] {$\sigma_{P M}$} (meas_right);

    % Background Box for Untrusted Lab
    \begin{scope}[on background layer]
        % Expanded bounding box coordinates to wrap the wider 3D shapes perfectly
        \coordinate (box_top) at (4.0, 3.6);
        \coordinate (box_right) at (6.8, 0); 
        \coordinate (box_bottom) at (4.0, -2.8);
        \coordinate (box_left) at (2.4, 0);
        \node[fill=black!5, draw=black!20, thick, rounded corners=8pt, inner sep=0.8cm, fit=(box_top) (box_bottom) (box_left) (box_right)] (untrusted) {};
      \node[
    above,
    font=\sffamily\large\bfseries,
    color=black!60
] at (untrusted.south)
{Untrusted Device};
    \end{scope}

    % ==========================================
    % 4. EXTRACTOR AND OUTPUT 
    % ==========================================
   % ==========================================
% 4. EXTRACTOR AND OUTPUT 
% ==========================================
\node[extractor] (ext) at (9.5, -0.5) {Extractor};

\node[
    rectangle,
    fill=cyan!10,
    draw=cyan!40,
    thick,
    rounded corners=4pt,
    minimum width=1.8cm,
    minimum height=4.2cm,
    align=center,
    font=\ttfamily\small,
    text=black!85
] (out) at (12.5, -0.5)
{
0\\
1\\
0\\
1\\
0\\
1\\
0\\
1\\[0.2cm]
\vdots\\
0\\
1\\
0
};

% Top label: keep this
\node[
    above=0.2cm of out,
    font=\sffamily\large,
    align=center,
    text=black!90
]
{Nearly-perfect \\
random bits};

% Output brace: move below the blue box
\draw[
    decorate,
    decoration={brace, amplitude=6pt, mirror},
    thick,
    draw=black!80
]
($(out.south west)+(0.15,-0.25)$)
--
($(out.south east)+(-0.15,-0.25)$)
node[midway, below=0.25cm, font=\bfseries]
{Output $\mathbf{O}$};

    % ==========================================
    % 5. ROUTING AND CONNECTIONS
    % ==========================================
    
    % X input from SV to Prep
    \draw[arrow] (sv.east |- prep_left) -- node[above, font=\bfseries] {$\mathbf{X}$} (prep_left);
    
    % Quantum Channel (P down to M)
    \draw[qchannel] (prep_bottom) -- (meas_top);
    \node[constraint] at (4.4, 0.55) {$\text{Energy} \leq \omega$};
    
\node[
    align=center,
    text=blue!70!black,
    anchor=west
] at (1.6,0.5)
{
{\footnotesize\bfseries Quantum}\\
{\footnotesize\bfseries system}\\[1.5mm]
{\normalsize\bfseries $S$}
};
    
    % B output from Meas to Extractor
    \draw[arrow] (meas_right) -- node[above, font=\bfseries] {$\mathbf{B}$} (ext.west |- meas_right);
    
    % Z seed from SV to Extractor
    \draw[arrow] (sv.east |- 0,-4.1) -- node[above, font=\bfseries] {$\mathbf{Z}$} (ext.west |- 0,-4.1);

    % O output from Extractor to Final String
    \draw[arrow] (ext.east) -- node[above, font=\bfseries] {} (ext.east -| out.west);
%$\mathbf{O}$
\end{tikzpicture}
\caption{\justifying
    \textbf{Schematic of our energy-constrained SDI randomness amplification protocol.}  
    A sequence of imperfect random bits is provided by an $\varepsilon$-biased Santha-Vazirani source (left). 
    The sequence is partitioned into a subsequence $\mathbf{X}$, used as input for an untrusted device, 
    and another subsequence $\mathbf{Z}$, used as a seed for a randomness extractor. 
    The untrusted device consists of a preparation module ($\mathbf{P}$), which prepares states  with energy upper bounded by a given value $\omega$,  and a measurement module ($\mathbf{M}$), which produces a string of classical output bits $\mathbf{B}$. 
    %{\color{blue}[in the picture, instead of the formula with the trace, we can write "Energy $\le \omega$"]}   
    The preparation and measurement module, instead, are allowed to be correlated, by sharing an entangled state $\sigma_{PM}$ (dashed line %{\color{blue} [the line is a bit too thin, can we make it thicker and easier to see?]}
    ), which may be correlated with the adversary.  In addition, the adversary may possess  
    classical side information $\lambda$ about the bits in the source. In the end, the extractor compresses $\mathbf{B}$ 
    and $\mathbf{Z}$ to generate a certified, nearly-uniform output string $\mathbf{O}$ (right).  }
\label{fig:sdi_protocol_schematic}
\end{figure*}
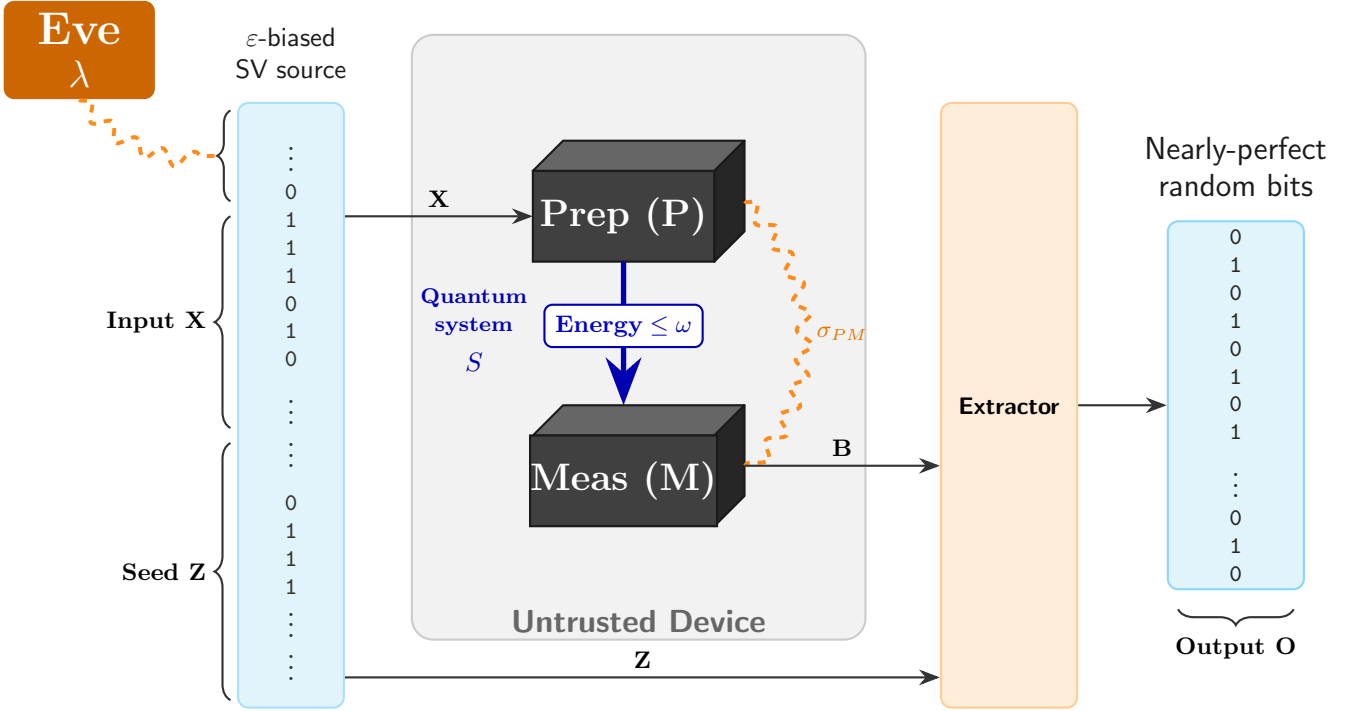

\section*{Results}\label{Results}

{\bf Randomness amplification protocol.}   In our protocol which builds upon the energy-constrained SDI framework introduced in \cite{Van2017,Senno2021}, the initial randomness is provided by a single source  of the Santha-Vazirani (SV) type, namely a source that generates sequence of input bits $x_1,x_2,\ldots,x_n$ such that, at every round $i \in \{1,\dots, n\}$,  the probability  of obtaining  the bit value $x_i$, conditional on  the  previous bit values $x_1,\dots ,  x_{i-1}$,  deviates from the   uniform distribution by a bounded amount $\varepsilon$.  Explicitly, the SV condition is  $\frac{1}{2}-\varepsilon \leq \mu_{X_i|X_1,\ldots,X_{i-1}}  (x_i|x_1,\ldots,x_{i-1}) \leq \frac{1}{2}+\varepsilon$, where   $\mu_{X_i|X_1,\ldots,X_{i-1}}$ is the conditional probability distribution of the random variable $X_i$ corresponding to the $i$-th  bit conditionally on the random variables $X_1\cdots X_{i-1}$ corresponding to the first $i-1$  bits. 

The high-level structure of the protocol is illustrated in Fig. \ref{fig:sdi_protocol_schematic}.  Initially, some  of the bits generated by the SV source are used as input for a preparation device, which prepares states of a quantum system $S$.  The specific nature of system $S$ is  not assumed to be known: $S$ could be a photon, an ion, or any other quantum system. Upon receiving the $i$-th input bit, the device prepares the system in a quantum state depending on the bit value. In the honest implementation of the protocol, the state    depends only on  $x_i$, and can be denoted as  $\rho^S_{x_i}$.  In general, the state may depend also on other parameters, including the attacker's side information, as discussed later in  this section. 
%\rho^S_{x_i|x_1\dots x_{i-1}}$ 
 After the state preparation, a measurement device performs a quantum measurement on system $S$, generating an output bit $b_i$.  In the following, we will use the notations $\mathbf{X} := X_1,X_2,\ldots,X_n$ and   $\mathbf{B} := B_1, B_2, \ldots, B_n$ for the input and output bit strings generated in $n$ rounds, respectively.   

Our protocol leverages the following   key assumptions.  (1) Low-energy preparation: at every round, the quantum state produced by the preparation device has energy upper bounded by a known value $\omega < E_1$, where $E_1$ is the energy of the first excited state.  Here  we do not assume any detailed knowledge of the system's Hamiltonian: we only assume that the ground state is unique, so that the condition that the energy is upper bounded by $\omega$ implies that the states produced by the device are in a neighborhood of  the ground state.
%conventionally set its energy to zero, and choose energy units such that the energy of the first excited state is  $\Delta E = 1$.  
%there exists a Hermitian operator $\hat{H}$, characterized by a lowest nondegenerate eigenvalue of $0$ and a unit spectral gap such that for any round $i \leq n$, given the input sequence $\mathbf{x} \in \{0, 1\}^i$ and prior measurement outcomes $\mathbf{b} \in \{0, 1\}^{i-1}$, the state satisfies:
    %$$
   % \text{Tr}[\hat{H} \rho^S_{x_i|x_1\dots x_{i-1},\mathbf{b},\lambda}] \leq \omega_{x_i},
%    $$
 %   where $\rho^S_{x_i|x_1,\dots, x_{i-1},\mathbf{b}}$ denotes the quantum state emitted by the preparation device in round $i$ conditioned on $\mathbf{x} =x_1,\dots, x_i$ and $\mathbf{b} = b_1,\dots, b_{i-1}$. {\color{blue}[ in the figure, we just have a single value $\omega$. Do we really need the dependence on $x$?  Or maybe we can just set $\omega: =\max\{\omega_0, \omega_1\}$ in the main text and leave the more general case of $\omega_0\not=\omega_1$ to the Supplemental Material?  ]}). 
 (2) Source-device independence: the quantum devices do not affect the behaviour of the  SV source in subsequent rounds. (3) Classical side information about the SV source: the adversary (Eve) holds only classical side information  about the SV source, and, conditionally on this information, the source remains of the SV type with bias $\varepsilon$. Explicitly, the condition is $\frac{1}{2}-\varepsilon \leq \mu_{X_i|\Lambda ,X_1,\ldots,X_{i-1}}  (x_i|\lambda , x_1,\ldots,x_{i-1}) \leq \frac{1}{2}+\varepsilon$, where $\Lambda$ is the random variable corresponding to Eve's side information, and $\lambda$ are its possible values.

To amplify the initial randomness, we exploit the  non-classicality of the measurement statistics,  witnessed by a quantity known as the  Measurement-Dependent Local (MDL) score  \cite{Pütz2014, Pütz2016, Senno2021}. The MDL score is computed using the distribution $p_{B_i,X_i}$ of the output bit $B_i$ given the input bit  $X_i$ at the $i$-th   round of the protocol.  Explicitly,  it is given by
\begin{equation}
    I_{\varepsilon}^{\omega}(p_{B,X}) = v_{\varepsilon}^{\omega} \Big( p(0,0) + p(1,1) \Big) - \frac{1}{v_{\varepsilon}^{\omega}} \Big( p(1,0) + p(0,1) \Big) \, ,
\end{equation} 
where we omitted  the subscript  $i$ for the round, and we defined $v_{\varepsilon}^{\omega}:= (1/4 - \varepsilon^2)\omega^2$  (see Supplementary Note 1 for more detail).  
 
The minimum value of the MDL score achievable by classical strategies provides a benchmark for quantum strategies:   any  value of the MDL score below the classical minimum witnesses a non-classical behavior \cite{Pütz2014, Pütz2016}.   In our setting, a classical strategy can be seen as the special case of quantum strategy,  in which the states generated by the preparation device are diagonal in the eigenbasis of system $S$' Hamiltonian. In \cite{Van2017,Senno2021}, it was shown that the minimum value of the MDL score achieved by classical strategies is $
\mathcal{B}_{\rm c} := [v^{\omega}_{\varepsilon} + \frac{1}{v^{\omega}_{\varepsilon}} ]\left(1/2-\varepsilon\right)(1-2\omega ) - \frac{1}{v^{\omega}_{\varepsilon}}$. Furthermore every violation of the classical bound $    I_{\varepsilon}^{\omega}(p_{B,X}) \ge  \mathcal{B}_{\rm c}$ guarantees that the measurement statistics  contains some intrinsic randomness,  which can   be used to amplify the randomness in the original SV source.

        In our protocol,   the value of the MDL score is estimated from the statistics of the output bit string $\bf B$ conditional on the value of the input bit string $\bf X$.    We fix an acceptance threshold $I_{\text{th}}  \le \mathcal{B}_{\rm c}$, and abort if the estimated value of the MDL score is above this threshold. If the estimated value passes this test, we then proceed to the final step: the extraction of randomness from the measurement data and from additional bits generated by the SV source.  In this step,  the bit string $\bf B$ is combined with another bit string $\bf Z$, generated by the SV source, and the two strings together are fed into an extractor, which produces the final output of the protocol: a bit string $\bf O$ that is nearly-uniform and nearly-unpredictable by the eavesdropper. For this purpose, one can use any extractor $\bf Ext$ that is secure against any adversary holding a quantum system that is  correlated with the inputs of the extractor. In our case, we use a Trevisan strong extractor  \cite{Foreman2025,De2012}, which, as proven in  Methods, yields a high-rate extraction even in our energy-constrained setting.

\textbf{Variational Gauss-Radau method for SDI randomness certification.} 
%\textcolor{red}{ In the many crytography tasks, the single-round randomness has a clear meaning: It is the amount of private unpredictability generated in one use of the devices against the strongest Eve allowed by the model. The fundamental single-round security quantity is the conditional min-entropy. It provides the worst-case entropy estimate. Eve can use all of her side information and perform the best possible measurements.
%In this sense, the min-entropy is too conservative. Because a small set of predicted events can substantially reduce the certified entropy. By contrast, the conditional von Neumann entropy provides a more balanced characterization of Eve's uncertainty. It capture eve's average information and more physically transparent.} 
 We now provide a method for  estimating the amount of randomness produced in each individual round of our protocol, and, more generally, in SDI protocols with an energy constraint.  Improving the single-round entropy estimates is important because these estimates feed into the estimate of the number of nearly perfect random bits produced after multiple rounds  of the protocol. Since the number of rounds is large (in our implementation $10^{11}$), even modest improvements in the single-round estimate result in significant enlargements in the length of the output bit string.

% Technically, here we use the Generalized Entropy Accumulation Theorem (GEAT) \cite{Metger2024}, which provides a guarantee on the amount of entropy present in the measurement data generated over $n$ rounds.    Compared to its non-generalized version \cite{Dupuis2020},  the GEAT allows us to treat Eve's side information and Eve's quantum system as dynamical variables, which are updated at every round of the protocol by incorporating the information obtained in all past rounds. These settings fits the multiround structure of our protocol, where we  allow sequential measurements and general quantum side information for $n$ uses of the devices in subsequent rounds. 
 
 \begin{figure*}[t!]
	
	\centering
	\includegraphics[width=1.0\textwidth]{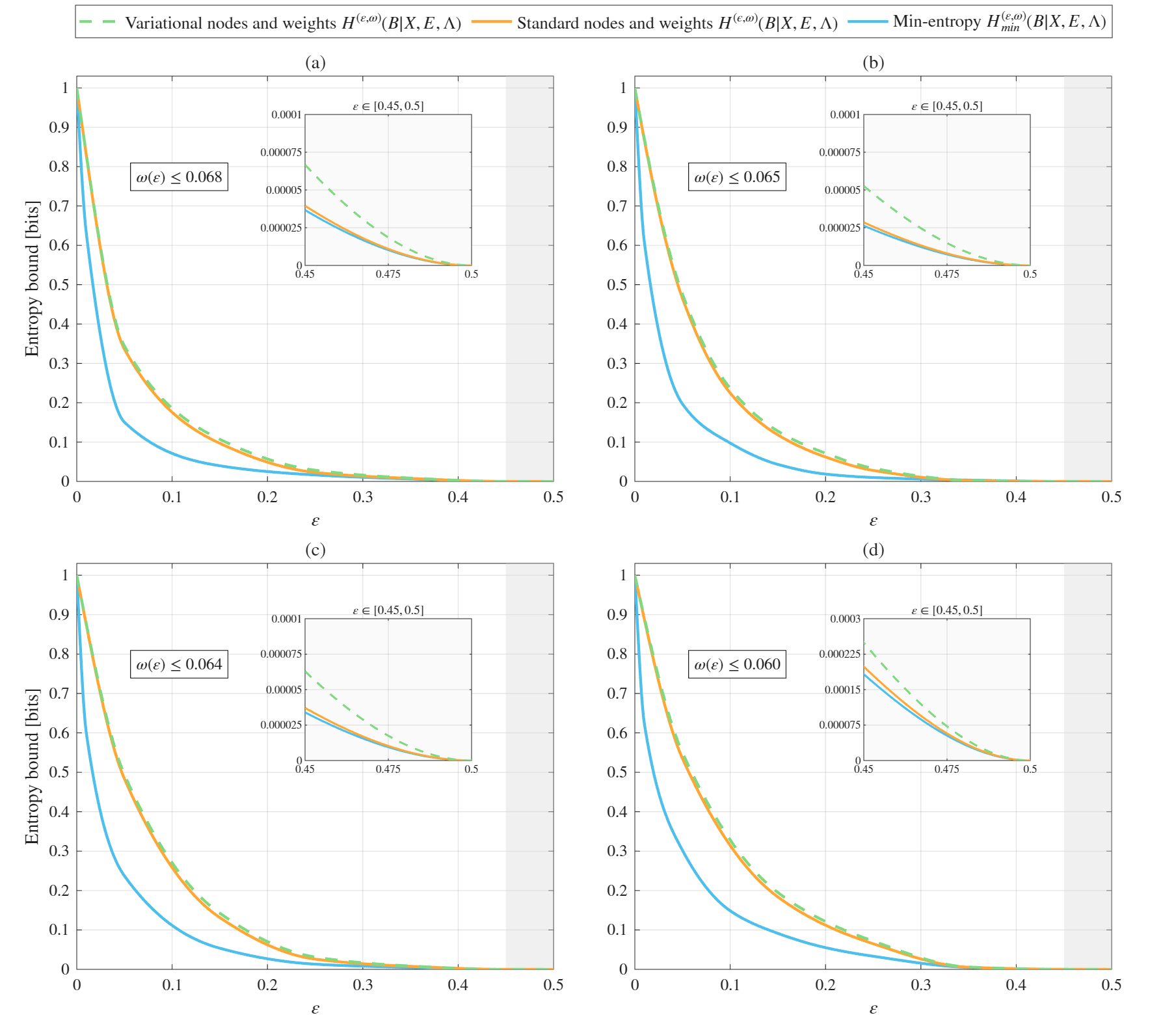}
	
    \caption{\justifying
    \textbf{Certified randomness in the presence of classical correlations between preparation and measurement devices.}       The plots display certified single-round entropy bounds as a function of the Santha-Vazirani source bias $\varepsilon \in [0, 0.5]$. 
    The blue curve represents the standard conditional min-entropy $H^{(\varepsilon,\omega)}_{\min}(B|X,E,\Lambda)$. 
    The orange curve depicts the conditional von Neumann entropy $H^{(\varepsilon,\omega)}(B|X,E,\Lambda)$ computed via standard Gauss-Radau method, 
    while the green dashed curve shows the improved bound using our variational method. 
    Subplots {(a)}--{(d)}, $\text{with suitable choices of parameters in range: } (a)\, \omega(\varepsilon) \le 0.068,\ (b)\, \omega(\varepsilon) \le 0.065,\ (c)\, \omega(\varepsilon) \le 0.064,\ (d)\, \omega(\varepsilon) \leq 0.060$, 
    correspond to varying energy constraints that is related to the experiment data. 
    Insets highlight the high-bias regime ($\varepsilon \to 0.5$), demonstrating that the variational framework yields tighter certification than both the standard node approximation and the min-entropy limit.}
	\label{fig:fig1}
\end{figure*}

To quantify the single-round entropy, we use the  conditional von Neumann entropy $H^{(\varepsilon,\omega)}(B|X,E,\Lambda)$, of the output bit $B$ conditioned on the input bit $X$, on the quantum system $E$ in Eve's control, and on  Eve's side information $\Lambda$ about the SV source.  Informally, the idea is that   Eve tries to  predict the output bit by exploiting her pre-shared entanglement with the devices and her classical side information about the source. In this task,  she is constrained by  conditions 1-3, and by the fact that her operations must  yield the value of the MDL score observed from the data.   

Mathematically,  Eve's optimized strategy arises from a minimization of the conditional von Neumann entropy   over the set of all possible quantum strategies compatible with the physical constraints 1-3 and with the observed statistics.  The challenge in this  constrained minimization, however, is that the entropy includes  a logarithmic   term that is hard to evaluate directly.  The evaluation is typically performed by replacing the logarithmic term by a finite sum of linear fractional functions, using the so-called  Gauss-Radau method \cite{Brown2024}.  This method builds on the integral representation of the log function, and turns it into a discrete sum,  evaluating the integrand at a fixed set of   sample points (called nodes) determined by the roots of a Legendre polynomial.  The values of the integrand at the nodes are then multiplied by suitable scaling factors (called weights).
More details are  provided in  Supplementary Note 2. 
 
 The  problem with the standard Gauss-Radau method is that it adopts a fixed choice of nodes, whose locations are uniquely determined  once and for all as soon as the total number of nodes is fixed.   This approach  works well in the fully  DI setting, where there is no assumption on the internal functioning of the devices.    However, it tends to provide loose bounds  in the SDI  scenario, where the constraint over the energy has major consequences on the evaluation of the integral.   In our scenario, the condition that the energy of the prepared states is less than the energy of  the first excited state, implies that the prepared states  are close to the ground state.    As a consequence, these states are nearly pure, and can have small, near-zero eigenvalues.  Since the log function has a singularity at zero,  the overall estimation of the entropy becomes very sensitive to  these small eigenvalues.    However, the standard Gauss-Radau method fails to place  sufficiently many  integration points near the origin to efficiently resolve the logarithmic singularity, and therefore results in unstable extrapolations, yielding relatively loose bounds on the entropy.  In the standard method,    the only way to improve the accuracy of the  estimation is to increase the total number of nodes, which makes  numerical evaluation computationally  challenging.

 To address the above problems, we introduce an upgraded version of the Gauss-Radau method. Our framework is applicable to the evaluation of the von Neumann entropy in the energy-constrained  scenario, and, more generally,  to other nonlinear  functions in the presence of a constraint on the expectation value of a given observable. The idea is to adjust the distribution of the nodes, making them more dense   in the regions of the spectrum that match the physical constraints in the SDI scenario, and  assigning higher weights to the nodes that  contribute more significantly to the estimation of the function of interest (in our specific problem, the conditional von Neumann entropy). 

Technically, the flexible choice of integration nodes is achieved by extending the original approach, based on Legendre polynomials,  to a broader family of polynomials, known as shifted Jacobi polynomials \cite{Szegő1939}.  The shifted Jacobi polynomials form a 2-parameter family, specified by  parameters $\alpha$ and $\beta$ in the range $(-1,1]$, and reduces to the Legendre polynomials when $\alpha  = \beta = 0$.   By treating $\alpha$  and $\beta$ as variational parameters,  we can then optimize the position of the integration nodes for increased precision in the near-zero region  where the logarithm function becomes more sensitive.   The details of this approach are provided in Methods and Supplementary Note 2.

%In our extended Gauss-Radau framework, the original Legendre polynomials are included as the special case in which both parameters $\alpha$ and $\beta$ are set to zero.  This fact suggests that, by optimizing the choice of parameters, one can obtain better bounds on the entropy. 
In the Methods section, we show that in the low-energy regime ($\omega \to 0$),  our variational method yields strictly tighter entropy bounds than the standard Gauss-Radau approach \cite{Brown2024}.  Moreover, we show that  increasing  the number of nodes and adjusting the variational weights in further iterations monotonically improves the quality of the approximation. Furthermore, in Supplementary Note 3,  we develop  a semidefinite programming (SDP) method for  minimizing the  conditional entropy over the prepare-and-measure realizations compatible with the observed data and the given energy constraint.

The comparison between our variational method and the standard Gauss-Radau method  is presented in Fig.~\ref{fig:fig1}. 
%The figure presents three lower bounds on the single-round conditional von Neumann entropy  as a function of  the SV source bias $\varepsilon$, for fixed values of the energy bound $\omega$. 
%The numerical results illustrate the advantage of our variational method over  existing techniques.  
  For ease of comparison, here we make the standard assumption that the preparation and measurement devices can only share classical correlations  \cite{Senno2021}. This assumption is lifted in Supplementary Note 4, where we show how to extend our variational Gauss-Radau method to the scenario where arbitrary entangled states can serve as inputs of the preparation and measurement devices.
  
  The  orange curve and green dashed curve in Fig.~\ref{fig:fig1}  are the bounds obtained from the standard Gauss-Radau method and from our variational method, respectively. Comparing them, one can observe improvements of the order of $10^{-5}$ to $10^{-2}$, which lead to improvements of the order of $10^6$ to $10^9$ when accumulated over $10^{11}$ rounds.  It is also important to observe that our method yields significant improvements in the high-bias regime ($\varepsilon \to  1/2$), when the bits produced by the SV source are almost completely predictable for Eve.  The improvements are  illustrated in the insets in the figure; for example, they  lead to a  $25\%$ increase (on average) of the estimated entropy for $\varepsilon  =  0.475$.  Finally, we note that the advantage of our method becomes increasingly prominent in the low-energy regime ($\omega\to 0$), as the eigenvalues of the quantum states satisfying the energy bound cluster near zero, thereby highlighting the benefit of  our variational choice of nodes.  

In addition to the comparison between our method and the standard Gauss-Radau method, Fig.~\ref{fig:fig1} also shows a standard bound on the conditional min-entropy (solid blue curve). Since the min-entropy  is a lower bound to the  von Neumann entropy,  this bound provides a benchmark for the bounds computed with Gauss-Radau methods.

\smallskip

\begin{figure*}[t!]
	
	\centering
	\includegraphics[width=1.0\textwidth]{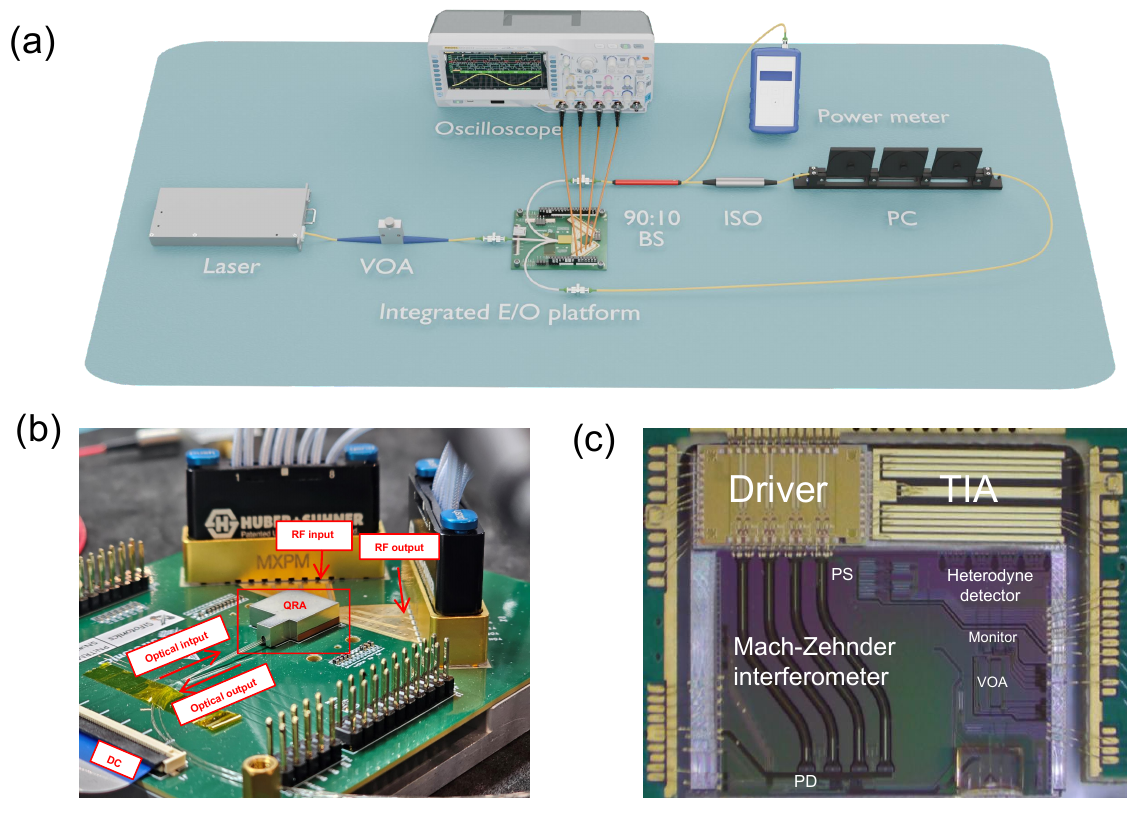}
\caption{\justifying
\textbf{Experimental SDI randomness amplification.}
{\em (a) Schematics of the setup.} The light emitted by 1550 nm laser is attenuated by a variable optical attenuator (VOA).  The attenuated light then enters into an integrated electronic/optical (E/O) platform,  which produces the coherent states used by our protocol. The coherent states exit the E/O platform  through the optical output port (upper link in the picture) and undergo a test to verify that they satify the required energy bound. They go through a 90:10 beam splitter (BS), and the  smaller fraction of the light is directed to a power meter (top right), while the larger fraction eventually undergoes a heterodyne measurement. Before measurement, the light goes through an optical isolator (ISO), which prevents backscattered photons from interfering with the energy monitoring, and  through a polarization controller (PC), which adjusts the polarization  to maximize the interference.  The output of the PC is then sent back inside the E/O platform (lower link in the picture) for heterodyne detection. Finally, the output photocurrent exiting the E/O platform is sent to an oscilloscope for readout of the classical outcome. 
{\em (b) Photograph of the E/O platform.} The platform consists of  electronic circuits and a photonic chip, where the main quantum operations take place. Both the electronics and optics are powered by a DC current (bottom left).  The photonic chip  receives its quantum input from (transmits its quantum output through) optical fibers, labelled as ``optical input''  (``optical output'') in the figure. Along with the quantum input, the photonic chip receives a classical input from a radiofrequency signal (RF input) controlling the phase of the quantum state.  The chip also produces a classical output (RF output) from heterodyne detection.  
{\em (c) Optical microscope image of the  photonic chip.}  The chip performs two main tasks: state preparation (left part of the figure) and measurement (right part). The state preparation module is controlled by  an external RF signal, which amplified to a high-voltage by a driver (top left).   An initial phase is fixed by a thermo-optical phase shifter (PS).  Depending on the RF signal, a $0/\pi$ phase is then added by  a Mach-Zehnder interferometer, made of four waveguides, which are also used to further attenuate the amplitude of the coherent states. Finally, the energy of the quantum state is monitored by  photodetector (PD) receiving the lower power output of a 1:99 beamsplitter (positioned under the PD and not visible in the picture) and the measured value is used to stabilize the working point of the Mach-Zehnder interferometer.     The state measurement  module starts  with a  variable optical attenuator (VOA), which adjusts the power of the incoming beam, and a power monitor, which ensures that the power is within the working space of heterodyne detection. Then, the quantum state undergoes heterodyne detection and the resulting photocurrent is amplified by a transimpedance amplifier (TIA) for readout of the classical outcome.  }

	\label{fig:setup}
\end{figure*}
\textbf{On-chip implementation.}  We now report the experimental  implementation of our protocol on a photonic setup,  which works at room temperature in the standard 1550 nm telecom band~\cite{Asakawa2020,AghaeeRad2025} and is therefore  suitable for large-scale deployment in  existing telecom architectures~\cite{Clark2026}.      The setup includes  a photonic chip, where the main quantum operations (state preparation and measurement) take place. The chip  measures $11\,\mathrm{mm}\times 11\,\mathrm{mm}\times 3\,\mathrm{mm}$, and is  sealed inside a protective casing for vibration-resistant, interference-resistant, and stable long-term operation.

   The schematic of the overall setup is illustrated   in Fig.~\ref{fig:setup}(a).  
A  narrow-linewidth (2 kHz)  1550 nm laser   emits a continuous-wave, which then passes through a variable optical attenuator (VOA) to set the optical power required for preparing energy-constrained coherent states. The attenuated light then enters into an integrated electronic/optical (E/O) platform.  The E/O platform further reduces the power of the light, and, depending on the value of the input bit $x\in\{0,1\}$ from the SV source, modulates its phase, producing one of two coherent states  $|\psi_x\rangle=\left|(-1)^x\alpha\right\rangle$, where $\alpha  >0$ is a  coherent-state amplitude satisfying the energy bound $\alpha^2\le \omega$ with $\omega = 0.0185$.
  Afterwards, the modulated coherent states  exit the E/O platform through an optical output port (upper link in the picture) and their energy is monitored to verify that they satisfy the required energy bound. The energy monitoring is performed with  a 90:10 beam splitter (BS). The  smaller fraction of the light is directed to a power meter (top right of the figure) for real-time energy measurement, while the larger fraction is sent back into the platform, where it undergoes heterodyne detection to determine the value of the position quadrature.   The classical readout takes place at an oscilloscope outside the E/O platform. Finally, the sign of the position quadrature determines the  value of the output bit $b$. A photo of the E/O platform  is shown in Fig.~\ref{fig:setup}(b). The platform includes a photonic chip for quantum state preparation and measurement, shown in  in Fig.~\ref{fig:setup}(c).      
%The heterodyne setup has eight output ports, which provide a differential photocurrent, sent to an external oscilloscope to  readout the value of the position quadrature. 
 More details on the  components of our setup are provided in  Supplementary Material Note 7.

\begin{figure}[t!]

	\makebox[\columnwidth][c]{%
		\includegraphics[width=1.25\columnwidth]{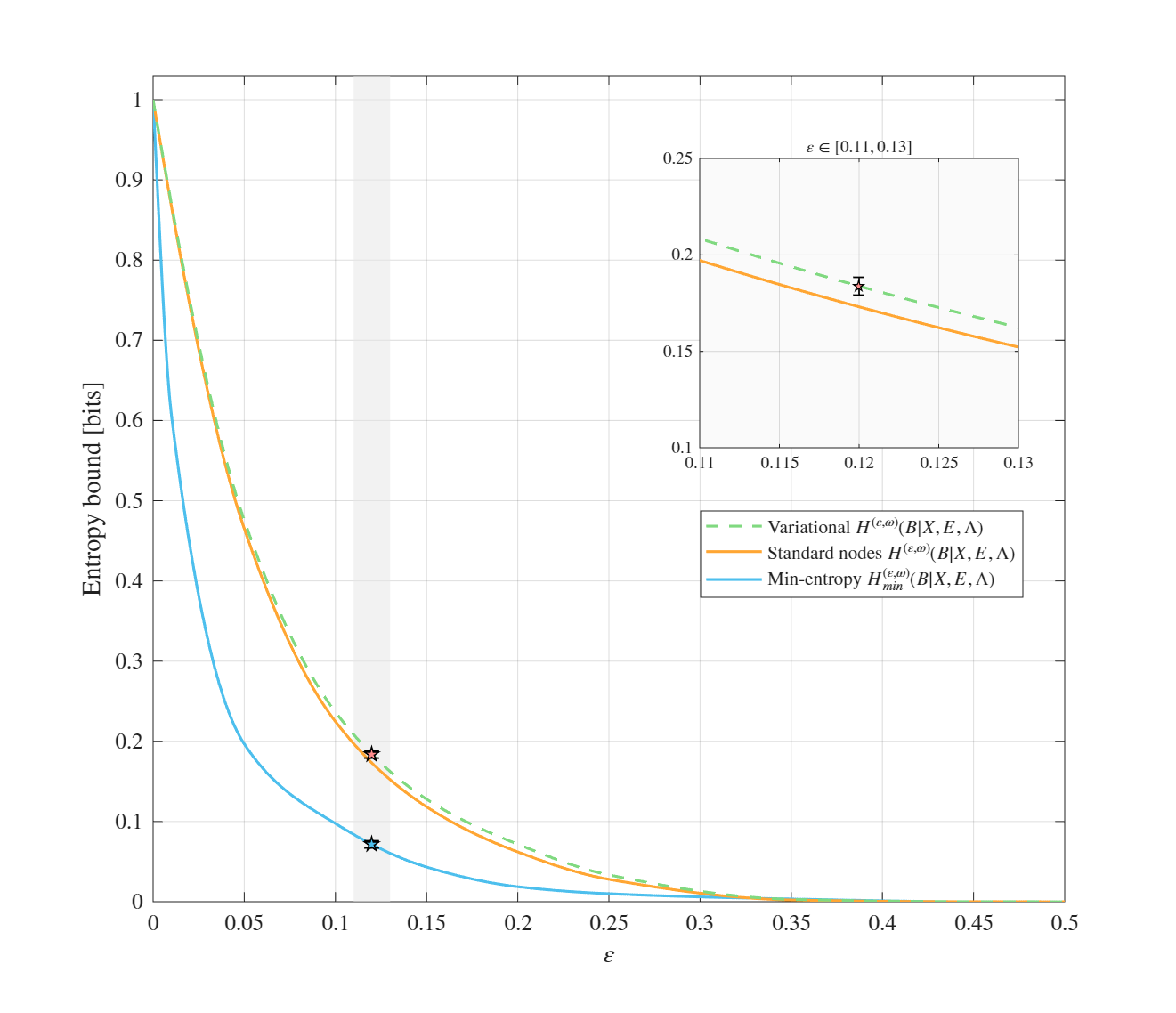}
	}
    \caption{\justifying \textbf{Experimentally certified entropies.}  
    The  three curves show  three lower bounds on the von Neumann entropy as a function of the SV source bias under  the experimental parameters of our setup. Specifically, the bounds are obtained  from our variational Gauss-Radau method (green dashed line), from the standard node Gauss-Radau method (orange solid line), and from the min-entropy baseline (blue solid line). In both panels, the insets provide a magnified view of the low-entropy region $\varepsilon \in [0.11, 0.13]$. The star markers show the certified values for the representative value $\varepsilon = 0.12$.  }
    %add error bar
	\label{fig:fig21}
\end{figure}

%Physically, unpredictability arises from the indistinguishability of low-energy coherent states approaching ground state(vacuum state). The SV source selects these indistinguishable quantum states, enabling true randomness generation and amplification.}\\

\indent  In Supplementary Note 8, we  apply our variational Gauss-Radau method to perform entropy certification under the experimental conditions characterizing our setup, namely $\omega  \approx  0.0185$,    $p(0,0) \approx 0.163$,    $p(0,1) \approx 0.342$,  $p(1,0) \approx 0.326$, and $p(1,1) \approx 0.169 $.   The results of the entropy certification are illustrated in Fig.~\ref{fig:fig21}, which shows the curves corresponding to three different lower bounds on the conditional von Neumann entropy as a function of the SV source bias $\varepsilon$. The figure shows that a non-zero  entropy can be certified for biases up to  $\varepsilon \leq 0.3$.   The star markers in the figure highlight the values for $\varepsilon =0.12$, which is approximately in the middle of the certifiable range.   Compared with the current state of the art for DI protocols, which tolerate biases up to   $\varepsilon \leq 0.0075$  \cite{Kulikov2024}, the  $\varepsilon = 0.12$ already represents a substantial improvement offered by the SDI approach, demonstrating its potential for  high-rate entropy generation within a  photonic chip.

Based on the experimentally certified entropy at $\varepsilon=0.12$, the on-chip implementation operating at a $1\,\mathrm{GHz}$ modulation rate yields a peak real-time certified randomness generation rate of approximately $1.8\times10^{8}\,\mathrm{bits/s}$.   Overall, this rate results in an estimated randomness amplification throughput of about $20~\mathrm{Mbps}$ (see Supplementary Note 8 for more details).

\section*{Discussion}
This work establishes a  route toward on-chip high-rate  SDI randomness amplification.  The observed enhancements in the amount of certified randomness  do not arise from isolated device improvements, but rather from a combination of device engineering  and  entropy certification techniques,  provided by the variational Gauss-Radau framework. Our demonstration paves the way to the realization of chip-size devices for SDI randomness amplification, which can be achieved with standard fabrication techniques by integrating the external elements of our setup into the chip design. In particular,  the electronic circuitry of the E/O platform can be fabricated on a monolithic chip with the  photonic elements.  Moreover, the core functions of external instruments such as the RF signal generator 
%, based on a digital analog converter (DAC), 
and the oscilloscope for outcome readout
%based on an analog digital converter (ADC),
can also be integrated into a single chip,  eliminating the need for  external equipment and enabling deployment in mobile devices like smartphones. 

%The key challenge in integrating both high-speed DAC and ADC lies in synchronizing the two, thereby eliminating the  uncertainty caused by cable delays in discrete setups. This challenge can  be addressed by including a shared clock for accurate timing. 

\indent Looking ahead to future applications, it  is useful to  point out some potential directions for improvement.  On the hardware side, fabrication variability and thermal drift
%, and finite data-size effects  
pose a limit to the performance of randomness amplification over short acquisition times. These issues can be addressed by enhanced fabrication procedures and by incorporating temperature stabilization into the chip design. 
%A first direction of improvement is that entropy certification relies on accurate spectral information, while the outcomes of the heterodyne measurement  may deteriorate under strong noise. 
%or high-dimensional spectra. %Second, side information about the quantum measurement should be further considered, as  it can impact the certified randomness. 
On the theory side, our analysis could be further extended to  block min-entropy sources or arbitrary min-entropy sources beyond the SV assumption.  All together, these improvements are expected to enable a viable technology of high-security, high-rate randomness generation for a broad range of applications in existing telecom infrastructures.   

%Overall, by combining variational entropy linearization with chip-scale photonic integration, this work demonstrates that SDI quantum randomness amplification is no longer constrained by loose entropy bounds or by physical assumptions that were previously unjustified or ad hoc. Instead, it establishes a unified framework in which tight certification and integrated implementation mutually reinforce each other, opening a realistic pathway toward high-speed, scalable quantum randomness sources for quantum-secure infrastructures.

\section*{Methods}

\textbf{Variational surrogates for spectral functions.} 
We introduce the following general framework to replace nonlinear spectral functions \(\Phi(A)\) by finite-dimensional surrogates, where self-adjoint operator \(A=\sum_i x_i \pi_i\) and \(\Phi(A)=\sum_i \Phi(x_i)\pi_i\) through spectral decomposition. The purpose is to approximate the target scalar function \(\Phi(x)\) on the eigenvalue domain and also to preserve the operator inequality required by the security proof.

Let \(\Phi:\mathcal D\to\mathbb R\) be a spectral function with a Stieltjes-type representation \cite{Davis1984}
\begin{equation}
    \Phi(x) = \int_{\Gamma} K(x,\tau)\,d\mu(\tau),
\end{equation}
where \(\mathcal D\) is the relevant spectral interval containing the eigenvalues \(x_i\) of feasible operators \(A\), \(K(x,\tau)\) being a positive kernel, and \(d\mu(\tau)\) is a non-negative measure. To construct computable approximation, we consider positive quadrature surrogates for each order \(m\)
\begin{equation}
    \mathcal S_m(K) = \left\{ S_m(x) =  \sum_{j=1}^m \gamma_j K(x,\tau_j)
    \;\Bigg|\; \gamma_j\ge0,\ \tau_j\in\Gamma \right\},
    \label{eq:surrogate_cone}
\end{equation}
where the nodes \(\tau_j\) and weights \(\gamma_j\) are chosen so that $\|S_m-\Phi\|_{L^\infty(\mathcal D)} \rightarrow 0, \; m\to\infty$. Under the standard assumptions for generalized Gaussian quadrature, one obtains a sequence \(S_m\in\mathcal S_m(K)\) that converges uniformly to \(\Phi\) on the relevant spectral interval and satisfies the one-sided domination \(S_m(A)\preceq \Phi(A)\). Increasing the quadrature order or adapting the nodes can tighten the approximation.
% for every feasible self-adjoint operator \(A\).

The next step is to choose the nodes and weights so that the surrogate with finite-order $m$ is tight on the relevant spectral region. We implement this by selecting the quadrature rule that reduces the accumulated approximation error, written as \(\inf_{u(\cdot)}\mathcal C[y,u]\) (see Supplementary Note 2).  We tune the underlying quadrature measure to change the associated quadrature nodes and weights, where this measure is written as \(d\mu(\tau)=u(\tau)d\tau\), and \(u(\tau)\ge0\) is the measure density. Therefore  the construction is formulated as a measure-selection problem. The accumulated error is written as
\begin{equation}
\mathcal{C}[y,u] =
\int_{\tau_0}^{\tau_f} \mathcal{L}(\tau,y(\tau),u(\tau),y'(\tau))\,d\tau +
\Psi(y(\tau_f)),
\label{eq:method_cost}
\end{equation}
where \(\mathcal L\) measures the local approximation error and \(\Psi\) imposes the endpoint condition. Here \([\tau_0,\tau_f]\) is the interval on which the quadrature rule is built, \(u(\tau)\) plays the role of the tunable measure density, and \(y(\tau)\) records the approximation object generated by this choice of measure \cite{Warga1972}. 

To turn the selected measure into a concrete surrogate defined in Eq. \eqref{eq:surrogate_cone}, we solve the resulting finite-dimensional approximation problem in the weighted polynomial basis. The weighted coefficient solution of \(\gamma\) is denoted by \(q\), which is not fixed uniquely by the linear equations alone, due to the corresponding linearized system having a nontrivial kernel. We remove this freedom related to the kernel by imposing additional linear constraints \(\{\Lambda_i(q)=d_i\}_{i=1}^{\ell}\), where \(\Lambda_i\) denotes the \(i\)-th constraint functional and \(\mathbf d=(d_1,\ldots,d_\ell)\) is the prescribed vector of normalization or compatibility values. If \(q_p\) is one particular solution and \(\{\alpha_j\}_{j=1}^{\ell}\) spans the kernel, then the corrected solution is
\begin{equation}
q^* = q_p + \sum_{j=1}^{\ell}
\left( \mathbf M^{-1}
(\mathbf d-\boldsymbol{\Lambda}(q_p))
\right)_j \alpha_j,
\end{equation}
where matrix \(\mathbf M\) has elements \(\mathbf M_{ij}=\Lambda_i(\alpha_j)\). Thus the imposed constraints select the element satisfying \(\Lambda_i(q^*)=d_i\) for all \(i\). Usually, this is implemented by a Galerkin expansion \cite{Canuto2006} (see Supplementary Note 2). The resulting \(q^*\) is used to construct the final surrogate \(S_m\).

Furthermore, the quadrature surrogate construction can be lifted from a scalar spectral approximation to a more general variational objective. In the quadrature cone in Eq. (3), let \(r_m\in\mathcal S_m(K)\) be the selected order-\(m\) surrogate for the target spectral function \(\Phi(A)\). For a lower bound, we require
\(r_m(A)\preceq \Phi(A)\) for every feasible self-adjoint operator \(A\). Thus \(r_m\) can replace \(\Phi\) inside any order-preserving variational objective without changing the direction of the bound. In particular, if the variational objective \(F_\Phi(\theta) = L(\theta)+\sum_i c_i\Omega_i(\Phi(A_i(\theta)))\),
with \(c_i\ge0\), and each \(\Omega_i\) is order-preserving, then
\begin{equation}
\begin{split}
    F_{r_m}(\theta)
    &= L(\theta) + \sum_i c_i
    \Omega_i\!\left(r_m(A_i(\theta))\right)                         \\
    &\le L(\theta) + \sum_i c_i
    \Omega_i\!\left(\Phi(A_i(\theta))\right) = F_\Phi(\theta).
\end{split}
\end{equation}

\textbf{Shifted-Jacobi-weighted variational surrogates.} 
 We now specialize the variational measure-selection step to a family of two-parameter shifted Jacobi polynomials used in this work. The aim is to flexibly capture the spectral distribution near endpoints of the interval while reducing the complexity of optimizing over general measures. This corresponds to choosing the density \(u(\tau)=W_{\alpha,\beta}(\tau)\) on the spectral interval \([\tau_0,\tau_f]=[0,1]\), where $W_{\alpha,\beta}(\tau) = (1-\tau)^\alpha\tau^\beta, \; \alpha,\beta>-1.$ The shifted Jacobi polynomials \(J_{(\alpha,\beta)}^{\,n}\) are orthogonal with respect to this weight, and tuning \((\alpha,\beta)\) changes the quadrature nodes. To obtain a Radau rule with one fixed endpoint at \(\tau=1\), we define the degree-\(N\) node-generating polynomial
\begin{equation}
\begin{split}
    Q_N(\tau)
    &= \frac{
    J_{(\alpha,\beta)}^{N+1}(\tau) + c_NJ_{(\alpha,\beta)}^N(\tau)}
    {\tau-1}, \;
    c_N = - \frac{J_{(\alpha,\beta)}^{N+1}(1)} {J_{(\alpha,\beta)}^N(1)} .
\end{split}
\label{eq:method_jacobi_radau_nodes}
\end{equation}
where \(\int_0^1 J_{(\alpha,\beta)}^n(\tau)J_{(\alpha,\beta)}^m(\tau)
W_{\alpha,\beta}(\tau)\,d\tau = \gamma_n^{(\alpha,\beta)}\delta_{mn}\). The coefficient \(c_N\) makes the numerator vanish at the pinned endpoint, so \(Q_N\) is a polynomial of degree \(N\). Its zeros define the interior nodes, while \(\tau_f=1\) is kept as the fixed Radau node. For every choice of \((\alpha,\beta)\), the resulting quadrature rule has positive weights and achieves exactness for all polynomial components up to degree \(2N\)(see Supplementary Note 2). If \(\{\tau_j\}_{j=1}^N\) are the zeros of \(Q_N\), then the weights \(\{\gamma_j\}_{j=0}^N\) satisfy \(\int_0^1 \phi(\tau)W_{\alpha,\beta}(\tau)\,d\tau = \sum_{j=0}^N\gamma_j\phi(\tau_j)\) for all
\(\phi\in\mathbb P_{2N}\), with
\begin{equation}
\begin{split}
    \gamma_0
    &=
    \frac{1}{Q_N(1)}
    \int_0^1 Q_N(\tau)W_{\alpha,\beta}(\tau)\,d\tau,                      \\
    \gamma_j
    &=
    \frac{1}{1-\tau_j}
    \frac{k_{N+1}^{(\alpha,\beta)}}{k_N^{(\alpha,\beta)}}
    \frac{\|Q_{N-1}\|^2_{\widetilde W_{\alpha,\beta}}}
    {Q_{N-1}(\tau_j)Q_N'(\tau_j)}.
\end{split}
\label{eq:method_jacobi_radau_weights}
\end{equation}
Here $1\le j\le N$, \(\widetilde W_{\alpha,\beta}=(1-\tau)W_{\alpha,\beta}\), and \(k_N^{(\alpha,\beta)}\) is the leading coefficient of \(J_{(\alpha,\beta)}^{\,N}\). The general control variable \(u(\tau)\) is therefore reduced to the tunable parameters \((\alpha,\beta)\). Varying these parameters moves the quadrature nodes in a controlled way. 
% while preserving positivity, polynomial exactness, and the one-sided surrogate structure needed for the entropy bound.

% \textcolor{blue}{add parts about the variable-node shifted jacobi weighted GR quadrature framework here, to make the connection of the lemma and the upper general framework more tight.}

For the entropy bounds considered in this work, the target spectral function is the logarithm. We use surrogates obtained from shifted Jacobi-weighted Gauss-Radau quadrature, verifying their admissibility and comparing the certified bounds obtained from different quadrature measures (see Supplementary Note 2).

\begin{lem}
\label{lem:jacobi_bound}
Let $d\mu_{\alpha,\beta}(\tau) = (1-\tau)^\alpha \tau^\beta\,d\tau$ be a shifted Jacobi measure on \([0,1]\), with \(\alpha,\beta>-1\). Let \(r^{\,m}_{\,\alpha,\beta}(x)\) be the order-\(m\) rational approximant to \(\ln x\) obtained from Gauss--Radau quadrature with fixed node
\(\tau_m=1\). Then, for every \(x\in(0,1)\), $r^{\,m}_{\,\alpha,\beta}(x)\le \ln x .$
\end{lem}

\begin{proof}
Write $g_x(\tau) = \frac{x-1}{1+\tau(x-1)}, \; \ln x = \int_0^1 g_x(\tau)\,d\tau .$
The Gauss-Radau remainder for the Jacobi-weighted rule has the form
\begin{equation}
\ln(x)-r^{\,m}_{\,\alpha,\beta}(x)
= \frac{g_x^{(2m-1)}(\xi)}{(2m-1)!}
\int_0^1 \phi_{m-1}(\tau)^2(\tau-1)\, d\mu_{\alpha,\beta}(\tau),
\end{equation}
for some \(\xi\in(0,1)\), where \(\phi_{m-1}\) is the nodal polynomial. For \(x\in(0,1)\), one has \(g_x^{(2m-1)}(\xi)<0\). Moreover, \(\phi_{m-1}(\tau)^2\ge0\) and \(\tau-1\le0\) on \([0,1]\), so the integral is non-positive. Hence the remainder is non-negative, and therefore $r^{\,m}_{\,\alpha,\beta}(x)\le \ln x.$
\end{proof}

% With Lemma~\ref{lem:jacobi_bound}, admissible Jacobi-Radau surrogate gives a valid entropy lower bound. We then compare the certified bounds obtained from different quadrature measures. % mix with the previous part, talking about the entropy bound.

\begin{thm}[Variational nodes improvement]
\label{thm:jacobi_improvement}
For an \(\varepsilon\)-biased SV source and parameters \(\alpha,\beta>-1\), let \(H^{\,m,\,(\varepsilon,\omega)}_{\,\alpha,\beta}\) denote the certified lower bound on the conditional von Neumann entropy obtained from an \(m\)-point variational nodes quadrature rule with respect to $d\mu_{\alpha,\beta}(\tau) = (1-\tau)^\alpha\tau^\beta\,d\tau,$ under the energy constraint \(\operatorname{Tr}(\hat H\rho_x^S)\le\omega_x\). If there exists an admissible choice \((\alpha,\beta)\) such that the corresponding logarithmic surrogate is pointwise larger than the Legendre surrogate on the active spectral region, then
\begin{equation}
    H^{m,\,(\varepsilon,\omega)}_{\,\alpha,\beta}(B|X,E,\Lambda) \ge
    H^{m,\,(\varepsilon,\omega)}_{\,0,0}(B|X,E,\Lambda),
\end{equation}
where \((\alpha,\beta)=(0,0)\) is the fixed Legendre case.
\end{thm}

\begin{proof}
The certified entropy bound is monotone in the logarithmic surrogate. A larger admissible lower surrogate for \(\ln x\) gives a larger certified lower bound. For a fixed quadrature order \(m\), the approximation error $\mathcal E^{\,m}_{\alpha,\beta}(x) = |\ln x-r^{\,m}_{\alpha,\beta}(x)|$ depends on how the quadrature nodes are distributed over the spectral domain. In the low-energy regime, the relevant generalized eigenvalues are concentrated near the boundary of the interval. The Legendre rule \((\alpha,\beta)=(0,0)\)
does not adapt its nodes to this concentration. By changing \((\alpha,\beta)\), the Jacobi rule moves the quadrature nodes and can reduce the error on the active spectral region. Whenever this produces an admissible surrogate satisfying $r^{\,m}_{\,0,0}(x) \le r^{\,m}_{\,\alpha,\beta}(x) \le \ln x$ on that region, the order-preserving variational inequality gives $H^{\,m,\ (\varepsilon,\omega)}_{\,\alpha,\beta}(B|X,E,\Lambda) \ge H^{\,m,\,(\varepsilon,\omega)}_{\,0,0}(B|X,E,\Lambda).$
\end{proof}

\textbf{Randomness accumulation and extraction.} % We convert the single-round entropy bound into a finite-key randomness rate under an IID attack model. In each round, the same prepare-and-measure realization is used independently, while Eve may hold a quantum system \(E_i\) correlated with the transmitted system \(S_i\). Thus the post-measurement state factorizes as $\rho_{B^nX^nE^n\Lambda^n} = \rho_{BXE\Lambda}^{\otimes n}.$ The adversarial side information is therefore quantum in each round, but the attack does not introduce correlations between different rounds.
We now convert the single-round entropy bound into a finite-size randomness generation rate, considering two possible scenarios.  The first scenario corresponds to the standard setting of SDI protocols, in which the adversary is only allowed to have classical correlation with the devices. This assumption is  well justified in our chip setup, where it is unlikely that quantum correlations with external devices could be maintained by the adversary (at the same time, we allow the preparation and measurement devices to have internal quantum memories, located inside the chip). In this setting, since the adversary is only classically correlated with the devices, an elementary use of the chain rule  guarantees that the total entropy is equal to the single-round entropy times the number of rounds.  

In the second scenario, we allow the adversary to have arbitrary quantum correlations with the devices in each round, but we require that there is no correlation from one round to the next, corresponding to the so-called IID assumption often used in the literature \cite{Christandl,Lütkenhaus,Shor}. It is worth mentioning  that, in principle, the IID assumption could be lifted using an energy-constrained analogue of the Generalized Entropy Accumulation Theorem (GEAT) \cite{Metger2024}.     However, at present it  is unclear how to incorporate the energy constraints  into the GEAT, and  straightforward extensions can be proven to be insecure  \cite{Acin}. For this reason, here we maintain the more conservative IID assumption, while using  a quantum-proof  extractor in which the energy-constrained analog of the GEAT  can be directly incorporated when it becomes available.

Under the IID assumption, the single-round prepare-and-measure realization is used independently in all \(n\) rounds. Eve can hold a quantum system \(E_i\) correlated with the transmitted system \(S_i\), so the post-measurement state is $\rho_{B^nX^nE^n\Lambda^n} = \rho_{BXE\Lambda}^{\otimes n}.$ The adversarial side information is therefore quantum in each round, but the attack does not introduce correlations between different rounds. The single-round conditional von Neumann entropy is certified by the energy-constrained SDI framework. From the observed empirical MDL-type score \(\hat I_\varepsilon^\omega\), we include a statistical margin \(\mu_I\) such that the true single-round score is contained in the relaxed confidence region except with probability \(\varepsilon_{\rm stat}\). We then evaluate the worst-case single-round entropy rate \(h_{\rm}=H(\hat I_\varepsilon^\omega+\mu_I)\) over the SDI quantum set compatible with the energy and the SV-source constraints. The function \(H(I)\) is computed by the single-round entropy SDP.

Under the IID assumption, the conditional von Neumann entropy adds linearly over the rounds, giving a multi-round entropy contribution at least \(n h_{\rm}\) (see Supplementary Note 5). The fully quantum asymptotic equipartition property (AEP) \cite{Tomamichel2009} converts this IID von Neumann entropy bound into a smooth min-entropy bound $ H_{\min}^{\varepsilon_{\rm sm}} $. For smoothing parameter \(\varepsilon_{\rm sm}\),
\begin{equation}
\begin{split}
    H_{\min}^{\varepsilon_{\rm sm}}
    (B^n|X^n,E^n,\Lambda^n) &\ge
    n\,h_{\rm} - \Delta_{\rm AEP}(n,\varepsilon_{\rm sm},|B|),                       \\
    \Delta_{\rm AEP}(n,\varepsilon_{\rm sm},|B|)
    &= \sqrt n\,\kappa_B \sqrt{\log_2\frac{2}{\varepsilon_{\rm sm}^2}} \;,
\end{split}
\label{eq:method_collective_min_entropy}
\end{equation}
where \(\kappa_B=2\log_2(1+2|B|)\). Thus the raw smooth min-entropy used for
extraction is \(k_B=nH(\hat I_\varepsilon^\omega+\mu_I)
-\Delta_{\rm AEP}(n,\varepsilon_{\rm sm},|B|)\).

Finally, the raw string \(B^n\) is converted into the output \(O=\mathrm{Ext}(B^n,Z^d)\) using the quantum-proof two-source extractor described in Supplementary Note~6. Conditioned on acceptance, the output satisfies
\begin{equation}
\begin{split}
    \frac12
    \Pr[\Omega_{\rm acc}]
    \left\|
    \rho_{O X^n Z E^n\Lambda^n|\Omega_{\rm acc}}
    -\chi_O\otimes
    \rho'
    \right\|_1 \le \varepsilon_{\rm sec},
\end{split}
\label{eq:method_collective_security}
\end{equation}
where $\rho'=\rho_{X^n Z E^n\Lambda^n|\Omega_{\rm acc}}$, \(\varepsilon_{\rm sec}\) includes the statistic error, the AEP smoothing error, and the extractor error. In the numerics, we use the Trevisan quantum-proof two-source extractor \cite{Foreman2025}. If the weak seed has length \(d\) and min-entropy \(k_2\), the certified output length is the largest integer \(m\) satisfying \(m \le \frac{1}{10}\bigl(k_B+4(k_2-d)-4\log_2 m +8\log_2\varepsilon_{\rm ext}+9\log_2(4/3)-6\bigr)\).

\section*{Acknowledgements}
We thank Stefano Pironio and Antonio Ac\'in for helpful discussions. 
This work has been supported by the Hong Kong Research Grant Council (RGC) through the Research Impact Fund (RIF) grant  No. R7035-21F  and through the  General Research Fund (GRF) grant  No.~17211122,  by the Chinese Ministry of Science and Technology (MOST) through grant 2023ZD0300600, by  the Open Research Fund of the State Key Laboratory of Photonics and Communications (Grant No. 2026QZKF023),    and by the State Key Laboratory of Quantum Information Technologies and Materials, Chinese University of Hong Kong. Research at the Perimeter Institute is supported by the Government of Canada through the Department of Innovation, Science and Economic Development Canada and by the Province of Ontario through the Ministry of Research, Innovation and Science. 

\section*{Author contributions}
LL, YW, RR, and GC  conceived the research and  developed the theoretical framework. LL and YW performed numerical simulations. LL designed and experimentally implemented the quantum randomness amplification chip and performed the data analysis. LL, YW, RR, and GC jointly wrote the manuscript. 

\section*{Competing interests} The authors declare no competing interests.\\

\section*{Data availability}
The data that support the plots within this paper are available from the corresponding authors on reasonable request.

\section*{Code availability}
Computer codes  are available from the corresponding authors on reasonable request.

\clearpage
\onecolumngrid


\section*{Supplementary material}

# Supplementary Material: On-Chip Quantum Randomness Amplification

Lang Li,[1,2,*] Yutian Wu,[1,2] Giulio Chiribella,[1,3,4,†] and Ravishankar Ramanathan[1,2,‡]

[1] *QICI Quantum Information and Computation Initiative, School of Computing and Data Science, The University of Hong Kong, Pokfulam Road, Hong Kong, China*

[2] *HKU-Oxford Joint Laboratory for Quantum Information and Computation*

[3] *Department of Computer Science, University of Oxford, Wolfson Building, Parks Road, Oxford, UK*

[4] *Perimeter Institute for Theoretical Physics, 31 Caroline Street North, Waterloo, Ontario, Canada*

* langli@hku.hk

† giulio@cs.hku.hk

‡ ravi@cs.hku.hk

## CONTENTS

| Symbol | Meaning and Description |
|---|---|
| $\rho, \sigma$ | Arbitrary quantum density matrices or quantum states |
| $\rho_x^S$ | Quantum state prepared on the transmitted system $S$ for input $x$ |
| $\rho_x^{S\mathcal{M}}$ | Preparation state on the transmitted system $S$ and the measurement-device memory $\mathcal{M}$, used in the relaxed no-entanglement model |
| $M_b$ | POVM element associated with the classical outcome $b$ |
| $\hat{H}$ | Any physically motivated Hermitian operator whose expectation value limits the set of allowed preparations |
| $\omega$, $\omega_x$ | Physical upper bound on the mean energy of the transmitted states |
| $E, \Lambda$ | Eve's quantum side information $E$ of the system and classical side information $\Lambda$ of the source |
| $\mathcal{Q}^\omega$ | single-round quantum set compatible with the SDI energy constraint and the SV-source input constraints |
| $\mathcal{Q}_\omega^{\text{sep}}$, $\mathcal{Q}_\omega^{\text{ent}}$ | Correlation sets obtained without and with preshared entanglement between the preparation and measurement devices |
| $\Gamma$ | Moment matrix used in the noncommutative polynomial SDP relaxation |
| $Z_{b,i}$ | Auxiliary operator used to linearize the Gauss-Radau surrogate in the entropy SDP |
| $\tau_i, \gamma_i$ | Quadrature nodes and positive quadrature weights |
| $W_{\alpha,\beta}(\tau)$ | Shifted Jacobi weight $W_{\alpha,\beta}(\tau) = (1-\tau)^\alpha \tau^\beta$ on $[0,1]$ |
| $J^n_{(\alpha,\beta)}$ | Shifted Jacobi polynomial of degree $n$, orthogonal with respect to $W_{\alpha,\beta}$ |
| $Q_N$ | Variational node-generating polynomial; zeros give the interior quadrature nodes |
| $I_\varepsilon^\omega$, $\hat{I}_\varepsilon^\omega$ | Energy-constrained MDL-type score and its empirical value estimated from the accepted data |
| $h$ | Worst-case single-round conditional von Neumann entropy rate under IID attack model |
| $\Delta_{\text{AEP}}$ | Finite-size correction from the fully quantum asymptotic equipartition property |
| $k_B$ | Certified raw smooth min-entropy of $B^n$ |
| $B^n, Z^d, O$ | Raw output string, weak seed block, and final extracted output string |
| Ext | Quantum-proof strong two-source randomness extractor |
| $d, m$ | Weak seed length and final extracted output length |
| $k_1, k_2$ | Min-entropy lower bounds for the two extractor sources |
| $h_{\text{SV}}$ | Per-bit min-entropy rate of an $\varepsilon_{\text{SV}}$-biased SV source |
| $\varepsilon_{\text{stat}}, \varepsilon_{\text{sm}}, \varepsilon_{\text{ext}}, \varepsilon_{\text{sec}}$ | Statistical error, smoothing error, extractor error, and security error |

TABLE I: List of key symbols and their mathematical descriptions used throughout the supplementary material.

## SUPPLEMENTARY NOTE 1: SEMI-DEVICE-INDEPENDENT RANDOMNESS AMPLIFICATION FRAMEWORK

This supplementary note collects the notation and definitions used throughout the analysis. We first list the main symbols in Table I. We then introduce the SDI protocol with physical constraints, and the entropy quantities and measures in the following subsections.

### 1.1 Semi-device-independent amplification protocol with physical constraints

In this subsection, we describe the semi-device-independent (SDI) prepare-and-measure framework used in our randomness amplification protocol. The protocol uses a public Santha-Vazirani (SV) source to choose the inputs sent

to the preparation device. In this source for each generated bit, its conditional probability is bounded away from both 0 and 1. In the binary case, we write

$$\frac{1}{2} - \varepsilon \leq \mu(x_i = 0|x_1, \ldots, x_{i-1}, \Lambda) \leq \frac{1}{2} + \varepsilon, \qquad 0 \leq \varepsilon < \frac{1}{2}, \tag{1}$$

where $\Lambda$ denotes classical side information held by the adversary about the weak source. The device consists of a preparation box and a measurement box. In one round, the preparation box receives an input $x \in X$ and emits a quantum system $S$. The emitted state can depend on $x$ and on adversarial side information $\lambda$. We denote it by $\rho^S_{x,\lambda}$. The measurement box receives $S$ and produces a classical output $b \in B$. Since the measurement is uncharacterized, it is modeled by an unknown POVM

$$M_b \succeq 0, \qquad \sum_{b \in B} M_b = \mathbb{1}. \tag{2}$$

The observed single-round behaviour is

$$p(b|x, \lambda) = \mathrm{Tr}\left[M_b \rho^S_{x,\lambda}\right]. \tag{3}$$

When the dependence of $x$ on $\lambda$ is explicitly included, the joint distribution is

$$p(b, x) = \sum_{\lambda} p(\lambda)\mu(x|\lambda)p(b|x, \lambda). \tag{4}$$

Without a restriction on the prepared states, no randomness can be certified in a prepare-and-measure framework. The preparation device could encode the input $x$ into perfectly distinguishable states, and the measurement device can recover $x$ and generate any desired deterministic response. The purpose of the SDI assumption is to retain a physically meaningful restriction on the devices while avoiding a full characterization of the internal states and measurements. In our setting, the physical constraint imposed on the prepared quantum systems is an energy bound. The measurement device remains uncharacterized. The role of the SDI assumption is to rule out such unrestricted classical deterministic strategies.

Let $\hat{H}$ be any physically motivated Hermitian operator whose expectation value limits the set of allowed preparations. In optical implementations, $\hat{H}$ can be the photon-number operator $\hat{n}$. We normalize $\hat{H}$ so that its lowest eigenvalue is 0, and its first nonzero excitation has energy at least 1. The SDI energy assumption can be written as follows

$$\mathrm{Tr}\left[\hat{H} \rho^S_{x,\lambda}\right] \leq \omega_x, \qquad \forall x, \lambda. \tag{5}$$

The vector $\omega = (\omega_x)_{x \in X}$ specifies the energy budget for each input. When the input alphabet is binary, we write $\omega = (\omega_0, \omega_1)$. We denote by $\mathcal{Q}^\omega$ the set of all quantum behaviours and the energy bound $\omega$:

$$\mathcal{Q}^\omega = \left\{ p(b|x) : \begin{array}{l} p(b|x) = \mathrm{Tr}[M_b \rho^S_x], \\ \rho^S_x \succeq 0, \quad \mathrm{Tr}[\rho^S_x] = 1, \\ \mathrm{Tr}[\hat{H} \rho^S_x] \leq \omega_x, \\ M_b \succeq 0, \quad \sum_b M_b = \mathbb{1} \end{array} \right\}. \tag{6}$$

When adversarial side information $\lambda$ is present, the same definition is applied conditionally for each $\lambda$. The observed behaviour is obtained by convex mixing according to $p(\lambda)$ and $\mu(x|\lambda)$.

Because the input source is weak, the classical benchmark should allow correlations between the input $X$ and the adversary's side information $\Lambda$. A classical energy-bounded explanation consists of a hidden variable $\lambda$, a weak input distribution $\mu(x|\lambda)$, and deterministic response functions for the measurement output. The corresponding measurement-dependent classical set is

$$\mathcal{C}^\omega_\varepsilon = \left\{ p(b, x) : \begin{array}{c} p(b, x) = \sum_\lambda p(\lambda)\mu(x|\lambda)D_\lambda(b|x) \mid \exists \, \{p(\lambda), \mu(x|\lambda), D_\lambda(b|x), \omega_\lambda\}_{\lambda \in \Lambda} \\ \text{s.t.} \quad D_\lambda(b|x) \in \{0, 1\}^4 \cap \mathcal{Q}^{\omega^\lambda}, \quad \sum_\lambda p(\lambda)\omega_{x,\lambda} \leq \omega_x, \quad \mu(x|\lambda) \in \mathrm{SV}_\varepsilon \end{array} \right\}, \tag{7}$$

The classical model is allowed to exploit the same weak input source and the same admissible energy bound as the quantum model.

A behaviour lying outside $\mathcal{C}_\varepsilon^\omega$ cannot be explained by a deterministic energy-bounded model correlated with the SV source. This is the role of the MDL-type witness used in the protocol. We write this witness as

$$I_\varepsilon^\omega(p_{B,X}) = \nu_\omega^\varepsilon\big(p(0,0)+p(1,1)\big) - \frac{1}{\nu_\omega^\varepsilon}\big(p(0,1)+p(1,0)\big), \tag{8}$$

where $\nu_\omega^\varepsilon = (\frac{1}{4}-\varepsilon^2)\omega_0\omega_1$. The classical bound is denoted by

$$I_\varepsilon^\omega(p_{B,X}) \geq \mathcal{B}_c, \qquad \text{for all } p \in \mathcal{C}_\varepsilon^\omega. \tag{9}$$

The exact form of the classical bound associated with the witness $I_\varepsilon^\omega$ can be expressed as follows. In the binary case, we write a conditional behaviour as $p_{B|X}(b|x)$. The energy-bounded classical set $\mathcal{C}^\omega$ is a polytope [1]. For $\omega_0+\omega_1 \leq 1$, its vertices are

$$\begin{aligned}\mathcal{V}(\mathcal{C}^\omega) = \Big\{&(1,0,1,0),\ (0,1,0,1),\ (1-\omega_0-\omega_1,\omega_0+\omega_1,1,0),\ (0,1,\omega_0+\omega_1,1-\omega_0-\omega_1),\\ &(1,0,1-\omega_0-\omega_1,\omega_0+\omega_1),\ (\omega_0+\omega_1,1-\omega_0-\omega_1,0,1)\Big\}.\end{aligned} \tag{10}$$

The measurement-dependent set $\mathcal{C}_\varepsilon^\omega$ is obtained by combining behaviours in $\mathcal{C}^\omega$ with $\varepsilon$-SV input distributions. Since the witness $I_\varepsilon^\omega$ is affine in $p_{B,X}$, it is enough to check the extremal behaviours. These are of the form

$$p(b,x) = \mu(x)p(b|x), \qquad p_{B|X} \in \mathcal{V}(\mathcal{C}^\omega), \tag{11}$$

For every admissible input distribution we have $\mu(x) \geq \frac{1}{2}-\varepsilon$. The useful consequence of the vertex description in Eq. (10) is $p(0|0)+p(1|1) \geq 1-\omega_0-\omega_1$, for every vertex of $\mathcal{C}^\omega$. By convexity, the same inequality also holds for every behaviour in $\mathcal{C}^\omega$.

**Lemma 1** (Energy-constrained MDL bound [2]). *For every joint behaviour $p_{BX} \in \mathcal{C}_\varepsilon^\omega$ generated by an energy-bounded deterministic strategy correlated with an $\varepsilon$-SV source, one has*

$$I_\varepsilon^\omega(p_{B,X}) := \nu_\omega^\varepsilon\big[p_{B,X}(0,0)+p_{B,X}(1,1)\big] - \frac{1}{\nu_\omega^\varepsilon}\big[p_{B,X}(0,1)+p_{B,X}(1,0)\big] \geq \mathcal{B}_c, \tag{12}$$

*where*

$$\mathcal{B}_c = \left(\nu_\omega^\varepsilon + \frac{1}{\nu_\omega^\varepsilon}\right)\left(\frac{1}{2}-\varepsilon\right)(1-\omega_0-\omega_1) - \frac{1}{\nu_\omega^\varepsilon}, \qquad \nu_\omega^\varepsilon = \left(\frac{1}{4}-\varepsilon^2\right)\omega_0\omega_1. \tag{13}$$

Hence every behaviour in $\mathcal{C}_\varepsilon^\omega$ satisfies $I_\varepsilon^\omega(p_{B,X}) \geq \mathcal{B}_c$. Therefore, an observed value $I_\varepsilon^\omega(p_{B,X}) < \mathcal{B}_c$ certifies that the observed behaviour is outside $\mathcal{C}_\varepsilon^\omega$, and hence cannot be reproduced by energy-bounded deterministic classical strategies correlated with the weak source.

The protocol uses $n$ rounds of the prepare-and-measure device. In each round $i$, the weak source produces an input $X_i$, the preparation box emits a state satisfying the energy constraint, and the measurement box produces an output $B_i$. The raw output string is $B^n = (B_1,\ldots,B_n)$. The empirical behaviour is

$$\mathrm{freq}_{BX}(b,x) = \frac{|\{i \in \{1,\ldots,n\} : (B_i,X_i)=(b,x)\}|}{n}. \tag{14}$$

The protocol aborts if the empirical MDL-type value exceeds the tolerated honest value:

$$I_\varepsilon^\omega(\mathrm{freq}_{BX}) > I_{\mathrm{exp}} + \mu_I, \tag{15}$$

where $I_{\mathrm{exp}}$ is the expected honest value of the witness and $\mu_I$ is the statistical tolerance.

**Randomness amplification protocol**

**Parameters:**

$\varepsilon$: bias of the Santha-Vazirani (SV) source
$\omega$: upper bound on the energy of the prepared states
$I_{\text{th}}$: acceptance threshold on the MDL score
$n \in \mathbb{N}$: number of rounds, equal to the length of the bit strings $\mathbf{X}$ and $\mathbf{B}$
$d \in \mathbb{N}$: length of the seed $\mathbf{Z}$ used for extraction
$m \in \mathbb{N}$: length of the final bit string $\mathbf{O}$
$\text{Ext} : \{0,1\}^n \times \{0,1\}^d \to \{0,1\}^m$, $\mathbf{B} \times \mathbf{Z} \mapsto \mathbf{O}$: quantum-proof two-source extractor

---

1. **Single-round prepare and measure protocol:**
   At round $i \in \{1, 2, \ldots, n\}$:
   (I) Sample input bit $x_i$ from SV source with bias $\varepsilon$.
   (II) Feed $x_i$ to the preparation device, which prepares system $S$ in a quantum state with energy upper bounded by $\omega$
   (III) Perform a two-outcome measurement on system $S$, which generates the output bit $b_i$.
2. **Parameter estimation:**
   (I) Evaluate the MDL score using the empirical frequencies the observed outputs $\mathbf{B} = b_1 \ldots b_n$ for given inputs $\mathbf{X} = x_1 \ldots x_n$.
   (II) If the observed MDL score is larger than the threshold $I_{\text{th}}$, then abort the protocol. Otherwise, proceed to the next step.
3. **Randomness extraction:**
   (I) Draw another bit string $\mathbf{Z} \in \{0,1\}^d$ from the SV source.
   (II) Apply the extractor to produce the final bit string $\mathbf{O} = \text{Ext}(\mathbf{B}, \mathbf{Z})$.

FIG. 1: **Randomness Amplification Protocol.** Our protocol consists of three stages. The first stage is the generation of quantum randomness through a prepare-and-measure setup, in which the prepared states have bounded energy, close to the ground state energy. The second stage is the estimation of the MDL score from the empirical data. If the MDL score is below a given threshold, the protocol proceeds to the third and final state: extraction of randomness from the measurement outcomes $\mathbf{B}$ and from an additional bit string $\mathbf{Z}$ generated by the SV source.

The protocol in Fig. 1 is designed to satisfy two properties. First, an honest implementation should pass the MDL test with high probability. Second, any implementation that passes the test should produce an output string that is close to uniform and independent of Eve's side information.

More explicitly, suppose that the honest prepare-and-measure device satisfies the energy bound $\omega$, uses inputs generated by an $\varepsilon$-biased SV source, and has expected MDL value $I_{\text{exp}}$. The statistical tolerance $\mu_I$ is chosen so that the empirical score $I_\varepsilon^\omega(\text{freq}_{BX})$ remains below the threshold $I_{\text{th}} = I_{\text{exp}} + \mu_I$ except with a small completeness error. Hence, for the honest implementation,

$$\Pr[\Omega_{\text{acc}}] = 1 - \Pr[\text{Abort}] \geq 1 - \varepsilon_{\text{comp}}, \qquad \varepsilon_{\text{comp}} \lesssim \exp\left[-c\, n\mu_I^2\right], \tag{16}$$

for some positive constant $c$ determined by the concentration bound used for the empirical MDL score. Thus increasing the number of rounds suppresses the probability that an honest execution aborts.

For security, the acceptance test restricts the possible single-round behaviours to those compatible with the relaxed score threshold. The single-round SDI entropy optimization then gives a lower bound on the raw entropy of $B^n$. The IID attack analysis converts this into a smooth min-entropy bound (more details in Supplementary Note 5). The quantum-proof two-source extractor maps $B^n$ and the additional weak seed $Z^d$ to the final output $O$ (more details in Supplementary Note 6). Therefore, for any target security parameter $\varepsilon_{\text{sec}}$ and output length $m$ satisfying the extractor condition, one can choose $n$, $d$, and the statistical tolerance so that the accepted output obeys

$$\frac{1}{2}\Pr[\Omega_{\text{acc}}]\left\|\rho_{OX^nZE^n\Lambda^n|\Omega_{\text{acc}}} - \chi_O \otimes \rho_{X^nZE^n\Lambda^n|\Omega_{\text{acc}}}\right\|_1 \leq \varepsilon_{\text{sec}}. \tag{17}$$

Here $\chi_O$ is the uniform state on $m$ bits. This means that, conditioned on the protocol not aborting, the final output is composably close to a uniform string and is independent of the weak seed, the public inputs, and Eve's quantum side information.

### 1.2 Single-shot and asymptotic entropy measures

In this subsection, we recall the entropy and information quantities used in the security analysis. All logarithms are base 2, so entropies are measured in bits. Classical random variables are denoted by capital letters such as $A$, with probability mass function $p_A(a)$. Quantum systems are represented by density operators $\rho$ acting on a Hilbert space $\mathcal{H}$. Unless otherwise stated, the definitions below are written in finite dimension.

Shannon entropy [3] quantifies the average uncertainty of a classical random variable. Let $A$ be a discrete random variable with distribution $p_A$. Its Shannon entropy is

$$H(A) = -\sum_a p_A(a) \log p_A(a). \tag{18}$$

For two discrete random variables $A$ and $C$, the conditional Shannon entropy is

$$H(A|C) = -\sum_{a,c} p_{AC}(a,c) \log p_{A|C=c}(a|c). \tag{19}$$

The mutual information between $A$ and $C$ is

$$I(A;C) = H(A) + H(C) - H(A,C) = H(A) - H(A|C). \tag{20}$$

For a classical-quantum state $\rho_{AE} = \sum_a p_A(a)|a\rangle\langle a| \otimes \rho_E^a$, where $A$ is classical and $\rho_E^a$ is Eve's conditional state, the optimal guessing probability is

$$P_{\mathrm{g}}(A|E)_\rho = \max_{\{M_a\}_a} \sum_a p_A(a)\, \mathrm{Tr}[M_a \rho_E^a]\,. \tag{21}$$

The maximization is over all POVMs $\{M_a\}_a$ on $E$. The conditional min-entropy [4] of the classical register $A$ is then

$$H_{\min}(A|E)_\rho = -\log P_{\mathrm{g}}(A|E)_\rho. \tag{22}$$

For finite-size analysis, one replaces the conditional min-entropy by its $\delta$-smooth version [5],

$$H_{\min}^{\delta}(A|E)_\rho = \sup_{\tilde{\rho}: P(\tilde{\rho}, \rho_{AE}) \le \delta} H_{\min}(A|E)_{\tilde{\rho}}. \tag{23}$$

Here $P(\cdot,\cdot)$ is the purified distance, and the optimization is over sub-normalized states $\tilde{\rho}$ with $\mathrm{Tr}(\tilde{\rho}) \le 1$. Operationally, $H_{\min}^{\delta}(A|E)$ characterizes the amount of randomness that can be extracted up to error $\delta$. In the asymptotic independent and identically distributed regime, it approaches the conditional von Neumann entropy through the fully quantum asymptotic equipartition property [6].

The von Neumann entropy [7] is the quantum analogue of Shannon entropy. For a state $\rho_A$, it is defined as

$$H(A)_\rho = -\,\mathrm{Tr}[\rho_A \log \rho_A]\,. \tag{24}$$

If $\rho_A$ is diagonal with eigenvalues $\{p_A(a)\}_a$, then $H(A)_\rho$ reduces to the Shannon entropy of the corresponding classical random variable $A$. For a bipartite state $\rho_{AC}$, the conditional von Neumann entropy is

$$H(A|C)_\rho = H(AC)_\rho - H(C)_\rho. \tag{25}$$

$H(A|C)_\rho$ can be negative when $\rho_{AC}$ is entangled.

For two quantum states $\rho$ and $\sigma$, the quantum relative entropy is defined as

$$D(\rho\|\sigma) = \mathrm{Tr}\left[\rho(\log \rho - \log \sigma)\right]. \tag{26}$$

### 1.3 Orthogonal polynomials and spectral methods

Orthogonal polynomials [8] provide the weighted polynomial basis used in our spectral surrogate construction. The roots generate the quadrature nodes, and the recurrence relations determine the corresponding quadrature weights. We define the notations for polynomials orthogonal with respect to a weight function.

Let $I = (c, d)$, where $-\infty \le c < d \le +\infty$, and let the weight function $W$ satisfy

$$W(\tau) > 0, \qquad \forall \tau \in I, \qquad W \in L^1(I). \tag{27}$$

Two functions $g$ and $h$ are orthogonal with respect to $W$ if

$$\langle g, h\rangle_W = \int_c^d g(\tau)h(\tau)W(\tau)\, d\tau = 0. \tag{28}$$

A polynomial of degree $m$ has the form

$$q_m(\tau) = c_m\tau^m + c_{m-1}\tau^{m-1} + \cdots + c_0, \qquad c_m \neq 0. \tag{29}$$

A sequence $\{q_m\}_{m=0}^{\infty}$ with $\deg(q_m) = m$ is orthogonal in $L^2_W(c, d)$ if

$$\langle q_m, q_l\rangle_W = \eta_m\delta_{ml}, \tag{30}$$

where $\eta_m = \|q_m\|^2_W \neq 0$. Let $\mathbb{P}_m = \text{span}\{1, \tau, \tau^2, \ldots, \tau^m\}$, then $\{q_0, \ldots, q_m\}$ forms a basis of $\mathbb{P}_m$, so $\mathbb{P}_m = \text{span}\{q_0, \ldots, q_m\}$. We will use the following immediate consequence.

**Lemma 2.** *$q_{m+1}$ is orthogonal to every polynomial in $\mathbb{P}_m$.*

*Proof.* For any $v \in \mathbb{P}_m$, write $v = \sum_{j=0}^{m} d_jq_j$. Then

$$\langle q_{m+1}, v\rangle_W = \sum_{j=0}^{m} d_j\langle q_{m+1}, q_j\rangle_W = 0. \tag{31}$$

□

Let

$$\tilde{q}_m(\tau) = \frac{q_m(\tau)}{c_m} = \tau^m + a^{(m)}_{m-1}\tau^{m-1} + \cdots + a^{(m)}_0 \tag{32}$$

denote the monic polynomial associated with $q_m$. For each weight function $W$, the monic orthogonal polynomials are uniquely determined and satisfy a three-term recurrence.

**Theorem 1.** *Let $W \in L^1(I)$ be positive on $I$. Then there exists a unique sequence of monic orthogonal polynomials $\{\tilde{q}_k\}_{k=0}^{\infty}$ with $\deg(\tilde{q}_k) = k$, given by*

$$\tilde{q}_0(\tau) = 1, \qquad \tilde{q}_1(\tau) = \tau - s_0, \tag{33}$$

*and*

$$\begin{aligned}\tilde{q}_{k+1}(\tau) &= (\tau - s_k)\tilde{q}_k(\tau) - t_k\tilde{q}_{k-1}(\tau),\\ &= (\tau - \frac{\langle \tau\tilde{q}_k, \tilde{q}_k\rangle_W}{\|\tilde{q}_k\|^2_W})\tilde{q}_k(\tau) - \frac{\|\tilde{q}_k\|^2_W}{\|\tilde{q}_{k-1}\|^2_W}\tilde{q}_{k-1}(\tau), \qquad k \ge 1.\end{aligned} \tag{34}$$

*Proof.* The first two monic polynomials are $\tilde{q}_0(\tau) = 1$ and $\tilde{q}_1(\tau) = \tau - s_0$. The condition $\langle \tilde{q}_0, \tilde{q}_1\rangle_W = 0$ gives

$$0 = \langle 1, \tau - s_0\rangle_W = \int_c^d \tau W(\tau)\, d\tau - s_0\int_c^d W(\tau)\, d\tau, \qquad s_0 = \frac{\langle \tau\tilde{q}_0, \tilde{q}_0\rangle_W}{\|\tilde{q}_0\|^2_W}. \tag{35}$$

Assume that $\{\tilde{q}_j\}_{j=0}^{k}$ is a mutually orthogonal family. Write

$$\tilde{q}_{k+1}(\tau) = \tau\tilde{q}_k(\tau) - s_k\tilde{q}_k(\tau) - t_k\tilde{q}_{k-1}(\tau) - \sum_{j=0}^{k-2} \mu^{(k)}_j\tilde{q}_j(\tau), \tag{36}$$

and impose $\langle \tilde{q}_{k+1}, \tilde{q}_j \rangle_W = 0$, $0 \le j \le k$. Taking inner products with $\tilde{q}_k, \tilde{q}_{k-1}$, and $\tilde{q}_j$ gives

$$\begin{aligned}
0 &= \langle \tilde{q}_{k+1}, \tilde{q}_k \rangle_W = \langle \tau \tilde{q}_k, \tilde{q}_k \rangle_W - s_k \|\tilde{q}_k\|_W^2, & s_k &= \frac{\langle \tau \tilde{q}_k, \tilde{q}_k \rangle_W}{\|\tilde{q}_k\|_W^2}, \\
0 &= \langle \tilde{q}_{k+1}, \tilde{q}_{k-1} \rangle_W = \langle \tau \tilde{q}_k, \tilde{q}_{k-1} \rangle_W - t_k \|\tilde{q}_{k-1}\|_W^2, & t_k &= \frac{\langle \tau \tilde{q}_k, \tilde{q}_{k-1} \rangle_W}{\|\tilde{q}_{k-1}\|_W^2} = \frac{\|\tilde{q}_k\|_W^2}{\|\tilde{q}_{k-1}\|_W^2}, \\
0 &= \langle \tilde{q}_{k+1}, \tilde{q}_j \rangle_W = \langle \tau \tilde{q}_k, \tilde{q}_j \rangle_W - \mu_j^{(k)} \|\tilde{q}_j\|_W^2, & \mu_j^{(k)} &= 0, \qquad 0 \le j \le k-2.
\end{aligned} \tag{37}$$

Here the identity for $t_k$ follows from $\tau \tilde{q}_{k-1} = \tilde{q}_k + \sum_{j=0}^{k-1} d_j^{(k)} \tilde{q}_j$, while $\mu_j^{(k)} = 0$ follows from $\langle \tau \tilde{q}_k, \tilde{q}_j \rangle_W = \langle \tilde{q}_k, \tau \tilde{q}_j \rangle_W = 0$, since $\tau \tilde{q}_j \in \mathbb{P}_{k-1}$. Therefore

$$\tilde{q}_{k+1}(\tau) = (\tau - s_k)\tilde{q}_k(\tau) - t_k \tilde{q}_{k-1}(\tau), \tag{38}$$

with the stated coefficients.

It remains to prove uniqueness. Let $\{r_k\}_{k=0}^{\infty}$ be another monic orthogonal sequence and set $h_{k+1} = \tilde{q}_{k+1} - r_{k+1}$. Since both polynomials are monic, $h_{k+1} \in \mathbb{P}_k$. Hence

$$\begin{aligned}
\|h_{k+1}\|_W^2 &= \langle \tilde{q}_{k+1} - r_{k+1}, h_{k+1} \rangle_W \\
&= \langle \tilde{q}_{k+1}, h_{k+1} \rangle_W - \langle r_{k+1}, h_{k+1} \rangle_W = 0,
\end{aligned} \tag{39}$$

because both $\tilde{q}_{k+1}$ and $r_{k+1}$ are orthogonal to $\mathbb{P}_k$. Thus $h_{k+1} = 0$, and the sequence is unique. □

Gauss-type quadrature [9] is the numerical step that turns the integral representation of a spectral function into a finite sum over nodes and weights. The nodes are chosen from the roots of orthogonal polynomials associated with the weight $W$, which gives high polynomial exactness and positive weights. We describe the form of this construction below.

A Gauss-type quadrature approximates a weighted integral with $g \in C^{m+1}[c,d]$ and $\xi(\tau) \in [c,d]$ by

$$\begin{aligned}
\int_c^d g(\tau) W(\tau)\, d\tau &= \sum_{k=0}^{m} g(\tau_k)\gamma_k + \mathcal{E}_m[g] \\
&= \sum_{k=0}^{m} g(\tau_k)\gamma_k + \frac{1}{(m+1)!} \int_c^d g^{(m+1)}(\xi(\tau)) \prod_{l=0}^{m} (\tau - \tau_l)\, W(\tau)\, d\tau \\
&= \sum_{k=0}^{m} g(\tau_k) \int_c^d \ell_k(\tau) W(\tau)\, d\tau + \frac{1}{(m+1)!} \int_c^d g^{(m+1)}(\xi(\tau)) \prod_{l=0}^{m} (\tau - \tau_l)\, W(\tau)\, d\tau \\
&= \sum_{k=0}^{m} g(\tau_k) \int_c^d \prod_{\substack{l=0 \\ l \ne k}}^{m} \frac{\tau - \tau_l}{\tau_k - \tau_l} W(\tau)\, d\tau + \frac{1}{(m+1)!} \int_c^d g^{(m+1)}(\xi(\tau)) \prod_{l=0}^{M} (\tau - \tau_l)\, W(\tau)\, d\tau
\end{aligned} \tag{40}$$

where $0 \le k \le m$, $\tau_k$ are distinct nodes, $\gamma_k$ are the associated weights, and $\mathcal{E}_m[g]$ is the error term. The rule is exact for $g$ if $\mathcal{E}_m[g] = 0$. The degree of precision of the rule is $m$ if $\mathcal{E}_m[p] = 0$, $\forall p \in \mathbb{P}_m$, while $\mathcal{E}_m[q] \ne 0$ for at least one $q \in \mathbb{P}_{m+1}$.

In general, for any $m+1$ distinct nodes $\{\tau_k\}_{k=0}^{m} \subset (c,d)$, the degree of precision lies between $m$ and $2m+1$. The upper bound is attained when the nodes are chosen as the zeros of the orthogonal polynomial $q_{m+1}$ associated with the weight $W(\tau)$. The resulting rule is the Gauss quadrature.

**Theorem 2** (Gauss quadrature)**.** *Let $\{\tau_k\}_{k=0}^{m}$ be the zeros of the orthogonal polynomial $q_{m+1}$. There exists a unique set of quadrature weights $\{\gamma_k\}_{k=0}^{m}$, defined by Eq. (40), such that the weights are positive and satisfy*

$$\int_c^d p(\tau) W(\tau)\, d\tau = \sum_{k=0}^{m} p(\tau_k)\, \gamma_k = \sum_{k=0}^{m} p(\tau_k)\, \frac{c_{m+1} \|q_m\|_W^2}{c_m\, q_m(\tau_k)\, q'_{m+1}(\tau_k)}, \qquad \forall p \in \mathbb{P}_{2m+1}, \tag{41}$$

*where $0 \le k \le m$ and $c_j$ denotes the leading coefficient of $q_j$.*

*Proof.* Let $\{\ell_k\}_{k=0}^m$ be the Lagrange basis associated with the nodes $\{\tau_k\}_{k=0}^m$, and let $\gamma_k = \int_c^d \ell_k(\tau)W(\tau)\,d\tau$. Take $p \in \mathbb{P}_{2m+1}$. By polynomial division, write $p = rq_{m+1} + v$, where $r, v \in \mathbb{P}_m$. Since $q_{m+1}(\tau_k) = 0$, we have $p(\tau_k) = v(\tau_k)$. Hence, using the orthogonality of $q_{m+1}$ to $\mathbb{P}_m$ and the Lagrange expansion of $v$, we obtain

$$\begin{aligned}\int_c^d p(\tau)W(\tau)\,d\tau &= \int_c^d \big(r(\tau)q_{m+1}(\tau) + v(\tau)\big)W(\tau)\,d\tau \\ &= \int_c^d v(\tau)W(\tau)\,d\tau \\ &= \int_c^d \sum_{k=0}^m v(\tau_k)\ell_k(\tau)W(\tau)\,d\tau \\ &= \sum_{k=0}^m v(\tau_k)\int_c^d \ell_k(\tau)W(\tau)\,d\tau = \sum_{k=0}^m v(\tau_k)\gamma_k = \sum_{k=0}^m p(\tau_k)\gamma_k, \end{aligned} \tag{42}$$

$$\gamma_j = \sum_{k=0}^m \ell_j^2(\tau_k)\gamma_k = \int_c^d \ell_j^2(\tau)W(\tau)\,d\tau > 0, \qquad 0 \le j \le m. \tag{43}$$

The first equation proves exactness on $\mathbb{P}_{2m+1}$. The second equation follows by taking $p = \ell_j^2 \in \mathbb{P}_{2m}$ and using $\ell_j(\tau_k) = \delta_{jk}$. Since $W > 0$ and $\ell_j \not\equiv 0$, all weights are strictly positive.

It remains to compute the weights. Since $q_{m+1}(\tau_k) = 0$, the Christoffel-Darboux identity gives

$$q_m(\tau_k)\frac{q_{m+1}(\tau)}{\tau - \tau_k} = \frac{c_{m+1}}{c_m}\sum_{i=0}^m \frac{\|q_m\|_W^2}{\|q_i\|_W^2} q_i(\tau_k)q_i(\tau). \tag{44}$$

Together with $\ell_k(\tau) = \frac{q_{m+1}(\tau)}{q'_{m+1}(\tau_k)(\tau-\tau_k)}$, integration against $W(\tau)$ gives

$$\begin{aligned}\frac{c_{m+1}}{c_m}\|q_m\|_W^2 &= \int_c^d \frac{c_{m+1}}{c_m}\sum_{i=0}^m \frac{\|q_m\|_W^2}{\|q_i\|_W^2} q_i(\tau_k)q_i(\tau)W(\tau)\,d\tau \\ &= q_m(\tau_k)\int_c^d \frac{q_{m+1}(\tau)}{\tau - \tau_k}W(\tau)\,d\tau \\ &= q_m(\tau_k)q'_{m+1}(\tau_k)\int_c^d \ell_k(\tau)W(\tau)\,d\tau \\ &= q_m(\tau_k)q'_{m+1}(\tau_k)\gamma_k. \end{aligned} \tag{45}$$

Therefore $\gamma_k = \frac{c_{m+1}\|q_m\|_W^2}{c_m q_m(\tau_k)q'_{m+1}(\tau_k)}$, which completes the proof. □

Because the spectral parameter $\lambda$ takes values in $[0,1]$, shifted Jacobi polynomials will be the relevant basis. We collect here the standard formulas that will be used later.

For $\alpha, \beta > -1$, the Jacobi polynomial $F^n_{(\alpha,\beta)}$ is written as

$$F^n_{(\alpha,\beta)}(x) = \sum_{k=0}^n f^{n,k}_{(\alpha,\beta)}(x-1)^{n-k}(x+1)^k, \qquad f^{n,k}_{(\alpha,\beta)} = \frac{1}{2^n}\binom{n+\alpha}{k}\binom{n+\beta}{n-k}. \tag{46}$$

These polynomials are orthogonal on $(-1,1)$ with respect to $(1-x)^\alpha(1+x)^\beta$:

$$\int_{-1}^1 (1-x)^\alpha(1+x)^\beta F^i_{(\alpha,\beta)}(x)F^j_{(\alpha,\beta)}(x)\,dx = \delta_{ij}\|F^j_{(\alpha,\beta)}\|^2 = \delta_{ij}\,\frac{2^{\alpha+\beta+1}}{2j+\alpha+\beta+1}\frac{\Gamma(j+\alpha+1)\Gamma(j+\beta+1)}{\Gamma(j+1)\Gamma(j+\alpha+\beta+1)}. \tag{47}$$

Setting $x = 2\lambda - 1$, define the shifted Jacobi basis and its weight on $(0,1)$ by

$$J^n_{(\alpha,\beta)}(\lambda) := F^n_{(\alpha,\beta)}(2\lambda - 1), \qquad W_{\alpha,\beta}(\lambda) := (1-\lambda)^\alpha\lambda^\beta, \qquad \lambda \in [0,1]. \tag{48}$$

Since $dx = 2\,d\lambda$ and $(1-x)^\alpha(1+x)^\beta = 2^{\alpha+\beta}(1-\lambda)^\alpha\lambda^\beta$, the orthogonality relation becomes

$$\begin{aligned}\int_0^1 W_{\alpha,\beta}(\lambda)J^i_{(\alpha,\beta)}(\lambda)J^j_{(\alpha,\beta)}(\lambda)\,d\lambda &= \frac{1}{2^{\alpha+\beta+1}}\int_{-1}^1(1-x)^\alpha(1+x)^\beta F^i_{(\alpha,\beta)}(x)F^j_{(\alpha,\beta)}(x)\,dx \\ &= \delta_{ij}\,\frac{1}{2j+\alpha+\beta+1}\frac{\Gamma(j+\alpha+1)\Gamma(j+\beta+1)}{\Gamma(j+1)\Gamma(j+\alpha+\beta+1)}.\end{aligned} \tag{49}$$

The norm is symmetric in the parameters $\|J^j_{(\alpha,\beta)}\| = \|J^j_{(\beta,\alpha)}\|$.

**Definition 1** (Weighted $L^2_W(0,1)$ spaces)**.** *For a positive weight function $W(\lambda) > 0$ on $(0,1)$, define*

$$L^2_W(0,1) = \left\{ g : (0,1) \to \mathbb{R} \;\middle|\; \int_0^1 W(\lambda)|g(\lambda)|^2\,d\lambda < \infty \right\}, \tag{50}$$

*with inner product $\langle g,h\rangle_W = \int_0^1 W(\lambda)g(\lambda)h(\lambda)\,d\lambda$ and norm $\|g\|_W = \langle g,g\rangle_W^{1/2}$.*

We will work with $L^2_{W_{\alpha,\beta}}(0,1)$, where $W_{\alpha,\beta}$ is defined in Eq. (48). By Eq. (49), the shifted Jacobi family $\{J^n_{(\alpha,\beta)}\}_{n\geq 0}$ is orthogonal in this space.

For integers $j \geq 0$ and $i \geq j$,

$$\frac{d^j}{dx^j}F^i_{(\alpha,\beta)}(x) = \frac{1}{2^j}\frac{\Gamma(i+j+\alpha+\beta+1)}{\Gamma(i+\alpha+\beta+1)}F^{i-j}_{(\alpha+j,\beta+j)}(x), \tag{51}$$

and

$$\frac{d^j}{d\lambda^j}J^i_{(\alpha,\beta)}(\lambda) = \frac{\Gamma(i+j+\alpha+\beta+1)}{\Gamma(i+\alpha+\beta+1)}J^{i-j}_{(\alpha+j,\beta+j)}(\lambda). \tag{52}$$

An identity can be written as

$$\frac{d^j}{d\lambda^j}\left\{W_{\alpha+j,\beta+j}(\lambda)\,J^{i-j}_{(\alpha+j,\beta+j)}(\lambda)\right\} = (-1)^j\frac{i!}{(i-j)!}W_{\alpha,\beta}(\lambda)\,J^i_{(\alpha,\beta)}(\lambda), \qquad i \geq j \geq 0. \tag{53}$$

## SUPPLEMENTARY NOTE 2: SPECTRAL APPROXIMATION

In this section, we present a general framework for replacing nonlinear spectral functions by finite-dimensional surrogates. The starting point is a scalar function defined on a spectral interval. We first rewrite this function through a weighted response representation, then approximate the response in a finite-dimensional polynomial space, and finally reconstruct a scalar surrogate. This general construction will later be applied to the logarithmic functions appearing in the entropy approximation.

### 2.1 General framework for spectral approximation

We first describe the general construction. Let $\Phi : (0,1) \to \mathbb{R}$ be the scalar function to be approximated. The construction has three steps. We rewrite $\Phi$ through a weighted derivative, approximate the resulting response function in a finite-dimensional polynomial space, and then reconstruct a scalar surrogate. The final surrogate is useful for the SDP if it satisfies a one-sided bound on the relevant spectral interval.

Let $\kappa : (0,1) \to (0,\infty)$ be a positive reconstruction weight. Define the response function the response function $q$ by

$$\begin{aligned} q(x) = -\kappa(x)\Phi'(x), \qquad \Phi'(x) = -\frac{q(x)}{\kappa(x)}, \\ \Phi(x) = \Phi(0) + \int_0^x \Phi'(s)\,ds = \Phi(0) - \int_0^x \frac{q(s)}{\kappa(s)}\,ds. \end{aligned} \tag{54}$$

Here $x$ is the spectral variable and $s$ is the integration variable. Thus approximating $\Phi$ is reduced to approximating $q$, followed by the reconstruction formula above. If $\Phi(0) = 0$, then the first term disappears. If an endpoint condition $\Phi(1) = 0$ is also imposed, it becomes the compatibility condition $\int_0^1 q(s)/\kappa(s)\,ds = 0$.

We next choose a positive projection weight $W$. This weight defines the inner product, the orthogonal polynomial basis, and the associated Gaussian-type quadrature rule. If $\{P_j\}_{j\geq 0}$ is the orthogonal basis associated with $W$, we approximate the response and reconstruct the surrogate as

$$\begin{aligned} q_N(x) &= \sum_{j=0}^{N} a_j P_j(x), \qquad q_N \in V_N = \operatorname{span}\{P_0, \ldots, P_N\}, \\ r_N(x) &= \Phi(0) - \int_0^x \frac{q_N(s)}{\kappa(s)}\, ds. \end{aligned} \tag{55}$$

Here $\kappa$ is used to reconstruct $\Phi$ from the response $q$ and $W$ is used to define the polynomial approximation and the quadrature nodes. If $W = W_\theta$ depends on an external parameter $\theta$, then the orthogonal polynomials, quadrature nodes, and quadrature weights also depend on $\theta$. This is the adaptive setting used below.

The admissibility requirement is $r_N(x) \leq \Phi(x)$ on the relevant spectral interval. Once this scalar inequality holds, the functional calculus gives $r_N(A) \preceq \Phi(A)$ for every self-adjoint operator $A$ with spectrum in that interval. The rest of this subsection makes these steps precise: spectral measure selection, Galerkin projection and the final variational bound.

We start with the domination principle used in the construction. We describe the lower-bound case, and the upper-bound version is obtained by reversing the inequality. Let $\Phi : \mathcal{D} \to \mathbb{R}$, with $\mathcal{D} \subseteq \mathbb{R}_+$, admit a Stieltjes-type representation, and let $\mathcal{S}_m(K)$ denote the corresponding cone of positive quadrature surrogates:

$$\begin{aligned} \Phi(x) &= \int_\Gamma K(x,\tau)\, d\mu(\tau), \\ \mathcal{S}_m(K) &= \left\{ S_m(x) = \sum_{j=1}^{m} \gamma_j K(x,\tau_j) \;\middle|\; \gamma_j \geq 0,\ \tau_j \in \Gamma \right\}. \end{aligned} \tag{56}$$

where $K(x,\tau) > 0,\ d\mu(\tau) \geq 0$. Assume that the quadrature construction gives $S_m \in \mathcal{S}_m(K)$ such that, on the relevant spectral interval $\mathcal{D}_1 = \mathcal{D} \cap [0,1]$, $R_m(x) = \Phi(x) - S_m(x) \geq 0,\ \epsilon_m = \sup_{x\in\mathcal{D}_1} |R_m(x)| = \sup_{x\in\mathcal{D}_1} |\Phi(x) - S_m(x)|$.

**Theorem 3** (Spectral domination and convergence)**.** *Let $\mathcal{H}$ be a finite-dimensional Hilbert space. Under Eqs. (56), the following statements hold.*

*(i) For every self-adjoint operator $A$ on $\mathcal{H}$ with $\sigma(A) \subseteq \mathcal{D}_1$ and positive linear functional $\mathcal{L}$,*

$$S_m(A) \preceq \Phi(A), \qquad \mathcal{L}(S_m(A)) \leq \mathcal{L}(\Phi(A)). \tag{57}$$

*(ii) The approximation error converges to zero in operator norm*

$$\lim_{m\to\infty} \sup_{\|A\|\leq 1,\ \sigma(A)\subseteq\mathcal{D}_1} \|\Phi(A) - S_m(A)\|_\infty = 0. \tag{58}$$

*Proof.* Let $A = \sum_k \lambda_k \Pi_k,\ \lambda_k \in \mathcal{D}_1$, be the spectral decomposition of $A$, where $\Pi_k$ are the spectral projectors. They are orthogonal projections onto the eigenspaces of $A$. By the continuous functional calculus, $h(A) = \sum_k h(\lambda_k)\Pi_k$ for every real-valued function $h$ on $\mathcal{D}_1$.

For the domination statement, using $R_m = \Phi - S_m$, we obtain

$$\begin{aligned} \Phi(A) - S_m(A) &= \sum_k \Phi(\lambda_k)\Pi_k - \sum_k S_m(\lambda_k)\Pi_k \\ &= \sum_k \big(\Phi(\lambda_k) - S_m(\lambda_k)\big)\Pi_k \\ &= \sum_k R_m(\lambda_k)\Pi_k \succeq 0. \end{aligned} \tag{59}$$

The last inequality follows from $R_m(\lambda_k) \geq 0$ and $\Pi_k \succeq 0$. Hence $S_m(A) \preceq \Phi(A)$. Applying a positive linear functional to the positive semidefinite operator $\Phi(A) - S_m(A)$ gives $\mathcal{L}(S_m(A)) \leq \mathcal{L}(\Phi(A))$.

For the convergence statement and $m \to \infty$, the same spectral decomposition gives

$$
\begin{aligned}
\|\Phi(A) - S_m(A)\|_\infty &= \left\| \sum_k \big(\Phi(\lambda_k) - S_m(\lambda_k)\big)\Pi_k \right\|_\infty \\
&= \left\| \sum_k R_m(\lambda_k)\Pi_k \right\|_\infty \\
&= \max_k |R_m(\lambda_k)| \\
&\le \sup_{x\in\mathcal{D}_1} |\Phi(x) - S_m(x)| = \epsilon_m \longrightarrow 0.
\end{aligned}
\tag{60}
$$

The bound is uniform over all self-adjoint $A$ with $\|A\| \le 1$ and $\sigma(A) \subseteq \mathcal{D}_1$. Therefore Eq. (58) follows. □

The convergence rate is controlled by the scalar error $\epsilon_m$. This error depends on the quadrature nodes and weights. These are fixed by the measure used to define the orthogonal polynomial basis. Thus the measure determines the geometry of the approximation. We therefore regard the measure as a design variable. In the absolutely continuous case, write $d\mu_u(\tau) = u(\tau)\,d\tau,\ u(\tau) > 0$. Once $u$ is chosen, it determines the weighted inner product, the orthogonal polynomials, the quadrature nodes and weights, and hence the surrogate $S^u_m$. The measure-selection problem can be written as

$$
S^u_m(x) = \sum_{j=1}^m \gamma^u_j K(x, \tau^u_j), \qquad \epsilon_m(u) = \sup_{x\in\mathcal{D}_1} |\Phi(x) - S^u_m(x)|\,, \qquad u^\star \in \arg\min_{u\in\mathcal{A}} \epsilon_m(u). \tag{61}
$$

Here $\mathcal{A}$ is the admissible class of positive normalized densities. A fixed-node rule corresponds to fixing $u$ in advance. An adaptive rule allows $u$ to vary with the parameter regime of the problem.

This is analogous to an optimal-control problem [10]. In the classical formulation, one chooses a control $u(\cdot)$ to minimize a cost functional while the state $y(\cdot)$ satisfies a constraint equation

$$
\begin{aligned}
&\inf_{u(\cdot)\in\mathcal{U}_{\rm ad}} \mathcal{C}[y,u] \\
&\mathcal{C}[y,u] = \int_{\tau_0}^{\tau_f} \mathcal{L}(\tau, y(\tau), u(\tau), y'(\tau))\,d\tau + \Psi(y(\tau_f)), \\
&\text{subject to } \ y'(\tau) = \mathcal{F}(\tau, y(\tau), u(\tau)), \qquad y(\tau_0) = y_0, \qquad \Psi(y(\tau_f)) = 0.
\end{aligned}
\tag{62}
$$

In our approximation problem, the role of the control is played by the density $u$. The cost is the approximation error $\epsilon_m(u)$. The corresponding state is a the finite-dimensional approximation structure generated by $u$. In this sense, adaptive quadrature is a measure-design problem.

Now we extend the formulation from strong form to weak form via test functions. Let $\mathcal{V}$ be the admissible function space and let $\mathcal{C} : \mathcal{V} \to \mathbb{R}$ be a differentiable functional. A minimizer $y^\star \in \mathcal{V}$ satisfies the first-variation condition, $\forall \eta \in \mathcal{V}_0$,

$$
\begin{aligned}
0 = \delta\mathcal{C}[y^\star, \eta] &= \frac{d}{d\varepsilon}\mathcal{C}[y^\star + \varepsilon\eta]\Big|_{\varepsilon=0}, \\
\delta\mathcal{C}[y^\star, \eta] &= \int_{\tau_0}^{\tau_f} \left( \frac{\partial\mathcal{L}}{\partial y}(\tau, y^\star, y^{\star\prime})\,\eta + \frac{\partial\mathcal{L}}{\partial y'}(\tau, y^\star, y^{\star\prime})\,\eta' \right) d\tau \\
&= \int_{\tau_0}^{\tau_f} \left[ \frac{\partial\mathcal{L}}{\partial y} - \frac{d}{d\tau}\left(\frac{\partial\mathcal{L}}{\partial y'}\right) \right] \eta\,d\tau + \left[ \frac{\partial\mathcal{L}}{\partial y'}\,\eta \right]_{\tau_0}^{\tau_f}.
\end{aligned}
\tag{63}
$$

If the variations satisfy homogeneous boundary conditions, the boundary term vanishes. The remaining identity is the weak form and the residual is tested against admissible variation $\eta$. A Galerkin approximation is then obtained by restricting the test functions and trial functions to a finite basis, such as the weighted orthogonal polynomial basis generated by $d\mu_\rho$.

To connect the preceding variational discussion with the spectral problem, we use the standard passage from a strong equation to a weak form. Let $\Omega \subset \mathbb{R}^d$, and consider the elliptic equation

$$
-\nabla \cdot (K(\mathbf{x})\nabla y) = g(\mathbf{x}), \qquad y|_{\partial\Omega} = 0.
$$

For any test function $v \in H_0^1(\Omega)$, multiplying by $v$, integrating over $\Omega$, and using integration by parts gives

$$\begin{aligned}\int_\Omega g(\mathbf{x})v(\mathbf{x})\,d\mathbf{x} &= -\int_\Omega v(\mathbf{x})\nabla\cdot\big(K(\mathbf{x})\nabla y(\mathbf{x})\big)\,d\mathbf{x}\\ &= \int_\Omega K(\mathbf{x})\nabla y(\mathbf{x})\cdot\nabla v(\mathbf{x})\,d\mathbf{x} - \int_{\partial\Omega} v(\mathbf{x})K(\mathbf{x})\nabla y(\mathbf{x})\cdot n\,dS \\ &= \int_\Omega K(\mathbf{x})\nabla y(\mathbf{x})\cdot\nabla v(\mathbf{x})\,d\mathbf{x}, \qquad \forall v\in H_0^1(\Omega).\end{aligned} \tag{64}$$

The boundary term vanishes because $v = 0$ on $\partial\Omega$. Thus the pointwise equation is replaced by a variational identity, which only requires the first weak derivative of $y$.

In the spectral approximation problem, the domain is the interval $(0, 1)$. The corresponding one-dimensional strong equation is

$$-\frac{d}{d\tau}\left(\kappa(\tau)\frac{dy}{d\tau}\right) = g(\tau), \qquad \kappa(\tau) > 0.$$

For $v \in H_0^1(0, 1)$, the same integration-by-parts argument gives

$$\begin{aligned}\int_0^1 g(\tau)v(\tau)\,d\tau &= -\int_0^1 v(\tau)\frac{d}{d\tau}\left(\kappa(\tau)\frac{dy}{d\tau}\right)d\tau\\ &= -\left[v(\tau)\kappa(\tau)\frac{dy}{d\tau}\right]_{\tau=0}^{\tau=1} + \int_0^1 \kappa(\tau)\frac{dy}{d\tau}\frac{dv}{d\tau}\,d\tau \\ &= \int_0^1 \kappa(\tau)\frac{dy}{d\tau}\frac{dv}{d\tau}\,d\tau, \qquad \forall v\in H_0^1(0,1).\end{aligned} \tag{65}$$

Hence the weak spectral problem is to find $y \in H^1(0, 1)$, satisfying the prescribed boundary conditions, such that

$$\int_0^1 \kappa(\tau)y'(\tau)v'(\tau)\,d\tau = \int_0^1 g(\tau)v(\tau)\,d\tau, \qquad \forall v\in H_0^1(0,1).$$

This is the form that can be projected onto a finite-dimensional weighted polynomial space in the Galerkin step.

Galerkin projection and spectral approximation can be constructed as follows. The weak form is now projected onto a finite-dimensional space. Let $\mathcal{V}_N := \operatorname{span}\{\phi_1, \ldots, \phi_N\} \subset H_0^1(0, 1)$. We seek $y_N \in \mathcal{V}_N$ and test only against functions in $\mathcal{V}_N$. Writing $y_N(\tau) = \sum_{k=1}^N c_k\phi_k(\tau)$, the Galerkin equations are obtained as

$$\begin{aligned}0 &= \int_0^1 \kappa(\tau)y_N'(\tau)\phi_j'(\tau)\,d\tau - \int_0^1 g(\tau)\phi_j(\tau)\,d\tau\\ &= \int_0^1 \kappa(\tau)\left(\sum_{k=1}^N c_k\phi_k'(\tau)\right)\phi_j'(\tau)\,d\tau - \int_0^1 g(\tau)\phi_j(\tau)\,d\tau\\ &= \sum_{k=1}^N c_k\int_0^1 \kappa(\tau)\phi_k'(\tau)\phi_j'(\tau)\,d\tau - \int_0^1 g(\tau)\phi_j(\tau)\,d\tau\\ &= \sum_{k=1}^N A_{jk}c_k - b_j, \qquad 1\le j\le N,\end{aligned} \tag{66}$$

where $A_{jk} = \int_0^1 \kappa(\tau)\phi_k'(\tau)\phi_j'(\tau)\,d\tau,\ b_j = \int_0^1 g(\tau)\phi_j(\tau)\,d\tau$. Thus the weak problem is reduced to the finite-dimensional linear system $A\mathbf{c} = \mathbf{b}$.

The approximation quality depends on the trial space and on the rule used to evaluate the weighted integrals. In a spectral Galerkin method, the functions $\phi_k$ are chosen from a global polynomial basis adapted to a positive weight $W$ on $(0, 1)$. The same weight then defines the associated orthogonality relation and the Gaussian-type quadrature rule. Therefore the matrix and vector entries can be written in the common form

$$\begin{aligned}A_{jk} &= \int_0^1 \kappa(\tau)\phi_k'(\tau)\phi_j'(\tau)\,d\tau = \int_0^1 H_{jk}(\tau)W(\tau)\,d\tau = \sum_{i=1}^M \omega_i H_{jk}(\tau_i) + \mathcal{E}_M[H_{jk}],\\ b_j &= \int_0^1 g(\tau)\phi_j(\tau)\,d\tau = \int_0^1 G_j(\tau)W(\tau)\,d\tau = \sum_{i=1}^M \omega_i G_j(\tau_i) + \mathcal{E}_M[G_j].\end{aligned} \tag{67}$$

Here $W$ is the quadrature weight, while $H_{jk}$ and $G_j$ denote the corresponding weighted integrands. Equivalently, one may use a nodal basis $\{\ell_i\}_{i=1}^M$ associated with the same nodes, satisfying $\ell_i(\tau_l) = \delta_{il}$. In this way the modal Galerkin representation and the nodal quadrature representation are linked by the same weighted polynomial structure.

For the rational and logarithmic surrogates used here, the quadrature rule must also respect the anchoring condition. For example, the normalization $f(0) = 0$ fixes the surrogate at the reference point and ensures the required normalization of the corresponding operator expression. This motivates a Gauss-Radau rule. One quadrature node is fixed at the endpoint $\tau_1 = 0$, while the remaining nodes $\{\tau_2, \dots, \tau_M\}$ are chosen according to the weighted orthogonal structure,

$$\int_0^1 h(\tau)W(\tau)\,d\tau = \omega_1 h(0) + \sum_{i=2}^{M} \omega_i h(\tau_i) + \mathcal{E}_M^{\mathrm{Radau}}[h]. \tag{68}$$

Thus the endpoint constraint is built into the quadrature rule itself, rather than imposed after discretization.

Another point is uniqueness. In several surrogate constructions, the response equation has a nontrivial kernel. Hence solving $\mathcal{T}q = g$ alone does not determine $q$ uniquely. If $h \in \ker(\mathcal{T})$, then $\mathcal{T}(q + h) = \mathcal{T}q + \mathcal{T}h = g + 0 = g$. The anchoring and compatibility conditions introduced earlier remove this ambiguity by fixing the component of $q$ in the kernel. The next lemma makes this point precise.

**Lemma 3** (Kernel correction constraints)**.** *Let $\mathcal{T} : \mathcal{X} \to \mathcal{Y}$ be a bounded linear operator between Hilbert spaces. Assume that $\ker(\mathcal{T})$ is finite-dimensional with basis $\{\alpha_1, \dots, \alpha_\ell\}$. Let $g \in \mathrm{Ran}(\mathcal{T})$, and let $q_p \in \mathcal{X}$ be one particular solution of $\mathcal{T}q_p = g$. Then all solutions of $\mathcal{T}q = g$ are obtained from $q_p$ by adding an element of $\ker(\mathcal{T})$.*

*Let $\Lambda_1, \dots, \Lambda_\ell : \mathcal{X} \to \mathbb{R}$ be continuous linear functionals, and define $\mathbf{M} \in \mathbb{R}^{\ell\times\ell}$ by $\mathbf{M}_{ij} = \Lambda_i(\alpha_j)$. If $\mathbf{M}$ is invertible, then for every $\mathbf{d} = (d_1, \dots, d_\ell)^T \in \mathbb{R}^\ell$ there exists a unique solution of $\mathcal{T}q = g$ satisfying the correction constraints $\Lambda_i(q) = d_i$, $1 \le i \le \ell$. It is given by*

$$q^* = q_p + \sum_{j=1}^{\ell} c_j^* \alpha_j, \qquad \mathbf{c}^* = \mathbf{M}^{-1}\left(\mathbf{d} - \mathbf{\Lambda}(q_p)\right), \tag{69}$$

*where $\mathbf{\Lambda}(q_p) = \big(\Lambda_1(q_p), \dots, \Lambda_\ell(q_p)\big)^T$.*

*Proof.* Let $q \in \mathcal{X}$. Since $q_p$ is a particular solution of $\mathcal{T}q_p = g$, we have

$$\begin{aligned} \mathcal{T}q = g &\iff \mathcal{T}q - \mathcal{T}q_p = 0 \iff \mathcal{T}(q - q_p) = 0 \iff q - q_p \in \ker(\mathcal{T}) \\ &\iff q - q_p = \sum_{j=1}^{\ell} c_j\alpha_j \iff q = q_p + \sum_{j=1}^{\ell} c_j\alpha_j. \end{aligned} \tag{70}$$

Conversely, every element of this affine space satisfies $\mathcal{T}\left(q_p + \sum_{j=1}^{\ell} c_j\alpha_j\right) = \mathcal{T}q_p + \sum_{j=1}^{\ell} c_j\mathcal{T}\alpha_j = g$, because $\alpha_j \in \ker(\mathcal{T})$. Hence this affine space is exactly the solution set of $\mathcal{T}q = g$.

It remains to impose the correction constraints. For $q^* = q_p + \sum_{j=1}^{\ell} c_j\alpha_j$, applying each $\Lambda_i$ gives

$$\begin{aligned} d_i = \Lambda_i(q^*) &= \Lambda_i\left(q_p + \sum_{j=1}^{\ell} c_j\alpha_j\right) \\ &= \Lambda_i(q_p) + \sum_{j=1}^{\ell} c_j\Lambda_i(\alpha_j) \\ &= \Lambda_i(q_p) + \sum_{j=1}^{\ell} \mathbf{M}_{ij}c_j, \qquad 1 \le i \le \ell. \end{aligned} \tag{71}$$

Equivalently, in vector form, $\mathbf{Mc} = \mathbf{d} - \mathbf{\Lambda}(q_p)$. Since $\mathbf{M}$ is invertible, this system has the unique solution $\mathbf{c}^* = \mathbf{M}^{-1}\left(\mathbf{d} - \mathbf{\Lambda}(q_p)\right)$. Substituting $\mathbf{c}^*$ into the affine representation gives the unique $q^*$ satisfying both $\mathcal{T}q^* = g$ and $\Lambda_i(q^*) = d_i$, $1 \le i \le \ell$. □

Lemma 3 shows that the kernel freedom can be removed by adding enough nondegenerate linear constraints. We now estimate the error introduced by the Galerkin truncation. The useful case is when the operator $\mathcal{T}$ maps one weighted orthogonal basis into another one diagonally. Then the Galerkin approximation is simply a truncated spectral expansion, and both the residual and the truncation error have explicit tail formulas.

**Lemma 4** (Adapted eigenfunction structure). *Let $d\mu(\tau) = \gamma(\tau)d\tau$ and $d\mu^*(\tau) = \gamma^*(\tau)d\tau$ be measures on $[0,1]$ with strictly positive densities. Let $\{P^n\}_{n\geq 0}$ and $\{P^{n,*}\}_{n\geq 0}$ be orthogonal polynomial families in $L^2(\mu)$ and $L^2(\mu^*)$, respectively. Assume that the linear operator $\mathcal{T}$ acts diagonally between the two weighted bases, namely $\mathcal{T}(\gamma P^n) = \lambda_n P^{n,*}$, with $\lambda_n > 0$ and $\lambda_n \asymp (n+1)^k$ for some $k > 0$.*

*Let $g \in L^2(\mu^*)$ have expansion $g = \sum_{n\geq 0} g_n^* P^{n,*}$, where $g_n^* = \langle g, P^{n,*}\rangle_{\mu^*}/\|P^{n,*}\|_{\mu^*}^2$. Define $q_N = \gamma\sum_{n=0}^N (g_n^*/\lambda_n)P^n$, and let $q := \gamma\sum_{n\geq 0}(g_n^*/\lambda_n)P^n$ be the corresponding formal exact solution. Then*

$$\begin{aligned} \mathcal{T}q_N = \sum_{n=0}^N g_n^* P^{n,*} = \Pi_N^* g, \qquad \|g - \mathcal{T}q_N\|_{\mu^*}^2 = \sum_{n>N} |g_n^*|^2 \|P^{n,*}\|_{\mu^*}^2, \\ \left\|\gamma^{-1}(q - q_N)\right\|_\mu^2 \leq \left(\min_{n>N}\lambda_n\right)^{-2} \sum_{n>N} |g_n^*|^2 \|P^n\|_\mu^2. \end{aligned} \tag{72}$$

*In particular, the growth $\lambda_n \asymp (n+1)^k$ gives a coefficient-level gain of order $(N+1)^{-k}$, up to the remaining tail energy.*

*Proof.* Applying $\mathcal{T}$ to the truncated expansion and using the diagonal mapping relation gives

$$\begin{aligned} \mathcal{T}q_N &= \mathcal{T}\left(\gamma \sum_{n=0}^N \frac{g_n^*}{\lambda_n} P^n\right) \\ &= \sum_{n=0}^N \frac{g_n^*}{\lambda_n}\mathcal{T}(\gamma P^n) = \sum_{n=0}^N \frac{g_n^*}{\lambda_n}\lambda_n P^{n,*} = \sum_{n=0}^N g_n^* P^{n,*} = \Pi_N^* g. \end{aligned} \tag{73}$$

Thus $\mathcal{T}q_N$ is exactly the $L^2(\mu^*)$-orthogonal projection of $g$ onto $\operatorname{span}\{P^{0,*},\dots,P^{N,*}\}$. Therefore, by orthogonality in $L^2(\mu^*)$,

$$\begin{aligned} \|g - \mathcal{T}q_N\|_{\mu^*}^2 &= \left\|\sum_{n\geq 0} g_n^* P^{n,*} - \sum_{n=0}^N g_n^* P^{n,*}\right\|_{\mu^*}^2 \\ &= \left\|\sum_{n>N} g_n^* P^{n,*}\right\|_{\mu^*}^2 = \sum_{n>N} |g_n^*|^2 \|P^{n,*}\|_{\mu^*}^2. \end{aligned} \tag{74}$$

For the approximation error, write $q = \gamma\sum_{n\geq 0}(g_n^*/\lambda_n)P^n$. Then

$$\begin{aligned} \left\|\gamma^{-1}(q - q_N)\right\|_\mu^2 &= \left\|\sum_{n>N} \frac{g_n^*}{\lambda_n} P^n\right\|_\mu^2 \\ &= \sum_{n>N} \left|\frac{g_n^*}{\lambda_n}\right|^2 \|P^n\|_\mu^2 \leq \left(\min_{n>N}\lambda_n\right)^{-2} \sum_{n>N} |g_n^*|^2 \|P^n\|_\mu^2. \end{aligned} \tag{75}$$

Finally, $\lambda_n \asymp (n+1)^k$ implies $(\min_{n>N}\lambda_n)^{-1} \asymp (N+1)^{-k}$, which gives the stated decay rate up to the remaining tail energy. □

The preceding lemmas describe how the response equation is solved after fixing the kernel freedom, and how the Galerkin truncation controls the approximation error. If the underlying weight depends on a parameter, the associated orthogonal basis and quadrature nodes also change with that parameter. This is the adaptive case. For the variational bounds used below, convergence alone is not sufficient: the reconstructed scalar surrogate must also preserve the direction of the bound. Thus we require $r_m(x) \leq \Phi(x)$ on the relevant spectral interval $\mathcal{D}$, where $\mathcal{D}$ contains the spectra of the operators under consideration. By the spectral functional calculus, this scalar inequality implies $r_m(A) \preceq \Phi(A)$. The next theorem records the corresponding variational principle.

**Theorem 4** (Variational bound from an admissible surrogate). *Let $\Phi, r_m : \mathcal{D} \to \mathbb{R}$, with $\mathcal{D} \subseteq [0,\infty)$, satisfy $r_m(x) \leq \Phi(x)$ on $\mathcal{D}$. Assume that, for each $\theta \in \Theta$, the operators $A_i(\theta)$ are bounded self-adjoint operators with $\sigma(A_i(\theta)) \subseteq \mathcal{D}$,*

*the coefficients $c_i$ are non-negative, and the functionals $\Omega_i$ are order-preserving: $X \preceq Y \;\rightarrow\; \Omega_i(X) \leq \Omega_i(Y)$. Define the exact and surrogate objectives by*

$$F_\Phi = \inf_{\theta\in\Theta}\left\{L(\theta) + \sum_{i=1}^{M} c_i \Omega_i\big(\Phi(A_i(\theta))\big)\right\}, \qquad F_{r_m} = \inf_{\theta\in\Theta}\left\{L(\theta) + \sum_{i=1}^{M} c_i \Omega_i\big(r_m(A_i(\theta))\big)\right\}. \tag{76}$$

*Then $F_{r_m} \leq F_\Phi$. Moreover, if $r_1(x) \leq r_2(x) \leq \Phi(x)$ on $\mathcal{D}$, then $F_{r_1} \leq F_{r_2} \leq F_\Phi$.*

*Proof.* Fix $\theta \in \Theta$ and $1 \leq i \leq M$. Set $h(x) = \Phi(x) - r_m(x) \geq 0$. Since $A_i(\theta)$ is self-adjoint and $\sigma(A_i(\theta)) \subseteq \mathcal{D}$, the spectral theorem gives $A_i(\theta) = \int_{\mathcal{D}} \lambda \, dE_{A_i}(\lambda)$. Hence

$$\begin{aligned} \Phi(A_i(\theta)) - r_m(A_i(\theta)) &= h(A_i(\theta)) \\ &= \int_{\mathcal{D}} h(\lambda)\, dE_{A_i}(\lambda) \succeq 0. \end{aligned} \tag{77}$$

Therefore $r_m(A_i(\theta)) \preceq \Phi(A_i(\theta))$. Using the order-preserving property of $\Omega_i$ and the non-negativity of $c_i$, we obtain

$$\begin{aligned} L(\theta) + \sum_{i=1}^{M} c_i\Omega_i\big(r_m(A_i(\theta))\big) &\leq L(\theta) + \sum_{i=1}^{M} c_i\Omega_i\big(\Phi(A_i(\theta))\big), \\ \inf_{\theta\in\Theta}\left\{L(\theta) + \sum_{i=1}^{M} c_i\Omega_i\big(r_m(A_i(\theta))\big)\right\} &\leq \inf_{\theta\in\Theta}\left\{L(\theta) + \sum_{i=1}^{M} c_i\Omega_i\big(\Phi(A_i(\theta))\big)\right\}. \end{aligned} \tag{78}$$

Thus $F_{r_m} \leq F_\Phi$.

Finally, if $r_1 \leq r_2 \leq \Phi$ on $\mathcal{D}$, then applying the same argument to $h_{12}(x) = r_2(x) - r_1(x) \geq 0$ gives

$$r_1(A_i(\theta)) \preceq r_2(A_i(\theta)) \preceq \Phi(A_i(\theta)).$$

The order-preserving property of $\Omega_i$, followed by the same infimum argument, gives $F_{r_1} \leq F_{r_2} \leq F_\Phi$. □

Theorem 4 shows that an admissible surrogate gives a valid variational bound. We now record the approximation step in the parameter-dependent case. For each fixed $\theta$, the weighted spaces, orthogonal bases, and multipliers are fixed, so Lemmas 3 and 4 apply directly.

**Lemma 5** (Parameterized Galerkin approximation)**.** *Fix $\theta \in \Theta$. Assume that the linearized equation $\mathcal{T}_\theta q = g_\theta$ has a unique solution $q_\theta$, where uniqueness is enforced by the compatibility constraints of Lemma 3. Assume also that $\mathcal{T}_\theta$ admits an adapted eigenfunction structure: for positive densities $\gamma_\theta, \gamma_\theta^*$, orthogonal families $\{P_\theta^n\}_{n\geq 0} \subset L^2(\mu_\theta)$ and $\{P_\theta^{n,*}\}_{n\geq 0} \subset L^2(\mu_\theta^*)$, and positive multipliers $\lambda_n(\theta)$, $\mathcal{T}_\theta(\gamma_\theta P_\theta^n) = \lambda_n(\theta) P_\theta^{n,*}$, $n \geq 0$. If $g_{\theta,n}^* = \frac{\langle g_\theta, P_\theta^{n,*}\rangle_{\mu_\theta^*}}{\|P_\theta^{n,*}\|^2_{\mu_\theta^*}}$, then the $N$-th Galerkin truncation is*

$$q_{\theta,N} = \gamma_\theta \sum_{n=0}^{N} \frac{g_{\theta,n}^*}{\lambda_n(\theta)} P_\theta^n, \qquad \mathcal{T}_\theta q_{\theta,N} = \sum_{n=0}^{N} g_{\theta,n}^* P_\theta^{n,*} =: \Pi_{\theta,N}^* g_\theta. \tag{79}$$

*Thus $\mathcal{T}_\theta q_{\theta,N}$ is the $\mu_\theta^*$-orthogonal projection of $g_\theta$ onto $\operatorname{span}\{P_\theta^{0,*}, \ldots, P_\theta^{N,*}\}$. In particular, the residual estimate and truncation bound of Lemma 4 hold for this fixed value of $\theta$.*

*Proof.* For fixed $\theta$, all measures, polynomial families, and multipliers are fixed. The compatibility constraints select the unique exact solution $q_\theta$, and Eq. (79) is precisely the Galerkin truncation of Lemma 4 applied to $\mathcal{T}_\theta$. The projection identity and the corresponding error bounds therefore follow directly from Lemmas 3 and 4. □

### 2.2 Variational node quadrature as a certified spectral-approximation primitive

The previous subsection describes the general framework from the construction of a certified surrogate to the choice of an admissible quadrature measure. We now make this step explicit by restricting the measure density to a two-parameter shifted Jacobi family [8], which turns the abstract measure-selection problem into a concrete variational

node quadrature rule on the spectral interval $[0,1]$. For $\alpha,\beta > -1$, set $W_{\alpha,\beta}(\tau) := (1-\tau)^{\alpha}\tau^{\beta}$. The shifted Jacobi polynomials $J^n_{(\alpha,\beta)}$ satisfy

$$\begin{aligned}\int_0^1 J^n_{(\alpha,\beta)}(\tau)J^m_{(\alpha,\beta)}(\tau)W_{\alpha,\beta}(\tau)\,d\tau &= \int_0^1 J^n_{(\alpha,\beta)}(\tau)J^m_{(\alpha,\beta)}(\tau)(1-\tau)^{\alpha}\tau^{\beta}\,d\tau \\ &= \|J^n_{(\alpha,\beta)}\|^2_{W_{\alpha,\beta}}\delta_{mn} = \gamma_n^{(\alpha,\beta)}\delta_{mn}.\end{aligned} \tag{80}$$

Here $\gamma_n^{(\alpha,\beta)} = \|J^n_{(\alpha,\beta)}\|^2_{W_{\alpha,\beta}}$, and we write $J^n_{(\alpha,\beta)}(\tau) = k_n^{(\alpha,\beta)}\tau^n + \cdots$ for its leading coefficient. Thus each choice of $(\alpha,\beta)$ fixes a weighted inner product, an orthogonal polynomial family, and a Gaussian-type quadrature rule. Varying $(\alpha,\beta)$ changes the weight, hence changes the roots of the orthogonal polynomials. This is the source of the variable-node construction.

For the Gauss–Radau rule used below, one node is fixed at the endpoint $\tau = 1$. The remaining interior nodes are obtained from a degree $N$ polynomial $Q_N$. Define

$$Q_N(\tau) = \frac{J^{N+1}_{(\alpha,\beta)}(\tau) + c_N J^N_{(\alpha,\beta)}(\tau)}{\tau - 1}, \qquad c_N = -\frac{J^{N+1}_{(\alpha,\beta)}(1)}{J^N_{(\alpha,\beta)}(1)}. \tag{81}$$

where $J^{N+1}_{(\alpha,\beta)}(1) + c_N J^N_{(\alpha,\beta)}(1) = 0$. The numerator therefore vanishes at $\tau = 1$, so $Q_N$ is a polynomial of degree $N$. Its leading coefficient is $k^{(\alpha,\beta)}_{N+1}$ from $J^{N+1}_{(\alpha,\beta)}$.

The polynomial $Q_N$ is orthogonal with respect to the effective Radau weight

$$\widetilde{W}_{\alpha,\beta}(\tau) = (1-\tau)W_{\alpha,\beta}(\tau) = (1-\tau)^{\alpha+1}\tau^{\beta}.$$

Indeed, for every $r_{N-1} \in \mathbb{P}_{N-1}$,

$$\begin{aligned}\int_0^1 Q_N(\tau)r_{N-1}(\tau)\widetilde{W}_{\alpha,\beta}(\tau)\,d\tau &= \int_0^1 \frac{J^{N+1}_{(\alpha,\beta)}(\tau) + c_N J^N_{(\alpha,\beta)}(\tau)}{\tau - 1} r_{N-1}(\tau)(1-\tau)W_{\alpha,\beta}(\tau)\,d\tau \\ &= -\int_0^1 \left(J^{N+1}_{(\alpha,\beta)}(\tau) + c_N J^N_{(\alpha,\beta)}(\tau)\right) r_{N-1}(\tau)W_{\alpha,\beta}(\tau)\,d\tau = 0.\end{aligned} \tag{82}$$

The last equality follows from the Jacobi orthogonality, since $r_{N-1}$ has degree at most $N-1$. Therefore $r_{N-1}$ is orthogonal to both $J^N_{(\alpha,\beta)}$ and $J^{N+1}_{(\alpha,\beta)}$. Hence the zeros of $Q_N$, together with the fixed endpoint $\tau = 1$, define the Gauss-Radau nodes associated with $W_{\alpha,\beta}$.

**Theorem 5** (Variational node Gauss–Radau quadrature)**.** *Let the boundary node be fixed at $\tau_0 = 1$, and let $\{\tau_j\}_{j=1}^N$ be the zeros of $Q_N$. Then there exists a unique set of positive quadrature weights $\{\gamma_j\}_{j=0}^N$ such that $\forall \phi \in \mathbb{P}_{2N}$ and $\int_0^1 \phi(\tau)W_{\alpha,\beta}(\tau)\,d\tau = \sum_{j=0}^N \gamma_j\phi(\tau_j)$,*

$$\gamma_0 = \frac{1}{Q_N(1)}\int_0^1 Q_N(\tau)W_{\alpha,\beta}(\tau)\,d\tau, \qquad \gamma_j = \frac{1}{1-\tau_j}\frac{k^{(\alpha,\beta)}_{N+1}}{k^{(\alpha,\beta)}_N}\frac{\|Q_{N-1}\|^2_{\widetilde{W}_{\alpha,\beta}}}{Q_{N-1}(\tau_j)Q'_N(\tau_j)}, \qquad 1 \le j \le N. \tag{83}$$

*Here $\widetilde{W}_{\alpha,\beta}(\tau) = (1-\tau)W_{\alpha,\beta}(\tau)$, and $k^{(\alpha,\beta)}_{N+1}$ and $k^{(\alpha,\beta)}_N$ are the leading coefficients inherited from the corresponding shifted Jacobi polynomials.*

*Proof.* Let $\pi_{N+1}(\tau) = (\tau-1)Q_N(\tau)$. The Lagrange basis associated with the nodes $\{\tau_j\}_{j=0}^N$, where $\tau_0 = 1$, is

$$h_j(\tau) = \frac{\pi_{N+1}(\tau)}{\pi'_{N+1}(\tau_j)(\tau - \tau_j)}, \qquad 0 \le j \le N,$$

and the quadrature weights are $\gamma_j = \int_0^1 h_j(\tau)W_{\alpha,\beta}(\tau)\,d\tau$.

For any $\phi \in \mathbb{P}_{2N}$, write $\phi(\tau) = a_{N-1}(\tau)(\tau-1)Q_N(\tau) + s_N(\tau)$, $a_{N-1} \in \mathbb{P}_{N-1}$, $s_N \in \mathbb{P}_N$. Since $\pi_{N+1}(\tau_j) = 0$, we

have $\phi(\tau_j) = s_N(\tau_j)$. Therefore

$$\begin{aligned}
\int_0^1 \phi(\tau) W_{\alpha,\beta}(\tau)\, d\tau &= \int_0^1 \left[a_{N-1}(\tau)(\tau-1)Q_N(\tau) + s_N(\tau)\right] W_{\alpha,\beta}(\tau)\, d\tau \\
&= -\int_0^1 a_{N-1}(\tau) Q_N(\tau) \widetilde{W}_{\alpha,\beta}(\tau)\, d\tau + \int_0^1 s_N(\tau) W_{\alpha,\beta}(\tau)\, d\tau \\
&= \int_0^1 s_N(\tau) W_{\alpha,\beta}(\tau)\, d\tau \\
&= \sum_{j=0}^{N} s_N(\tau_j) \int_0^1 h_j(\tau) W_{\alpha,\beta}(\tau)\, d\tau \\
&= \sum_{j=0}^{N} \gamma_j s_N(\tau_j) = \sum_{j=0}^{N} \gamma_j \phi(\tau_j).
\end{aligned} \tag{84}$$

The third equality follows from the orthogonality of $Q_N$ with respect to $\widetilde{W}_{\alpha,\beta}$. Thus the rule is exact on $\mathbb{P}_{2N}$. Taking $\phi = h_k^2 \in \mathbb{P}_{2N}$ gives $\gamma_k = \sum_{j=0}^{N} \gamma_j h_k^2(\tau_j) = \int_0^1 h_k^2(\tau) W_{\alpha,\beta}(\tau)\, d\tau > 0$, so all weights are strictly positive.

It remains to compute the weights. For the endpoint node, $h_0(\tau) = \frac{\pi_{N+1}(\tau)}{\pi'_{N+1}(1)(\tau-1)} = \frac{Q_N(\tau)}{Q_N(1)}$, hence

$$\gamma_0 = \int_0^1 h_0(\tau) W_{\alpha,\beta}(\tau)\, d\tau = \frac{1}{Q_N(1)} \int_0^1 Q_N(\tau) W_{\alpha,\beta}(\tau)\, d\tau.$$

For $1 \le j \le N$, since $Q_N(\tau_j) = 0$, we have $h_j(\tau) = \frac{(\tau-1)Q_N(\tau)}{(\tau_j-1)Q'_N(\tau_j)(\tau-\tau_j)}$. Therefore

$$\begin{aligned}
\gamma_j &= \int_0^1 h_j(\tau) W_{\alpha,\beta}(\tau)\, d\tau \\
&= \frac{1}{(\tau_j-1)Q'_N(\tau_j)} \int_0^1 \frac{(\tau-1)Q_N(\tau)}{\tau-\tau_j} W_{\alpha,\beta}(\tau)\, d\tau \\
&= -\frac{1}{\tau_j-1} \int_0^1 \frac{Q_N(\tau)}{Q'_N(\tau_j)(\tau-\tau_j)} \widetilde{W}_{\alpha,\beta}(\tau)\, d\tau \\
&= \frac{1}{1-\tau_j} \frac{k_{N+1}^{(\alpha,\beta)}}{k_N^{(\alpha,\beta)}} \frac{\|Q_{N-1}\|^2_{\widetilde{W}_{\alpha,\beta}}}{Q_{N-1}(\tau_j) Q'_N(\tau_j)}.
\end{aligned} \tag{85}$$

This proves and the explicit expressions for the weights. □

The preceding Radau construction can now be written explicitly for the shifted Jacobi weight. Since

$$\widetilde{W}_{\alpha,\beta}(\tau) = (1-\tau) W_{\alpha,\beta}(\tau) = (1-\tau)^{\alpha+1} \tau^{\beta},$$

the interior Radau nodes are the zeros of the shifted Jacobi polynomial $J^N_{(\alpha+1,\beta)}$. Thus changing $(\alpha,\beta)$ changes the weight and hence moves the quadrature nodes.

**Theorem 6** (Variational node Jacobi–Radau quadrature). *Let $\alpha, \beta > -1$. Fix the endpoint node $\tau_0 = 1$, and let $\{\tau_j\}_{j=1}^N$ be the zeros of $J^N_{(\alpha+1,\beta)}$. Then there exists a unique set of positive weights $\{\gamma_j\}_{j=0}^N$*

$$\begin{aligned}
\gamma_0 &= \frac{(\alpha+1)\Gamma^2(\alpha+1) N! \Gamma(N+\beta+1)}{\Gamma(N+\alpha+2)\Gamma(N+\alpha+\beta+2)}, \\
\gamma_j &= \frac{1}{1-\tau_j} \frac{G^{\alpha+1,\beta}_{N-1}}{J^{N-1}_{(\alpha+1,\beta)}(\tau_j)\partial_\tau J^N_{(\alpha+1,\beta)}(\tau_j)} = \frac{1}{(1-\tau_j)^2\tau_j} \frac{\widetilde{G}^{\alpha+1,\beta}_{N-1}}{\left[\partial_\tau J^N_{(\alpha+1,\beta)}(\tau_j)\right]^2}, \qquad 1 \le j \le N,
\end{aligned} \tag{86}$$

*where, for generic parameters $a, b > -1$,*

$$G^{a,b}_L = \frac{(2L+a+b+2)\Gamma(L+a+1)\Gamma(L+b+1)}{(L+1)!\Gamma(L+a+b+2)}, \qquad \widetilde{G}^{a,b}_L = \frac{\Gamma(L+a+2)\Gamma(L+b+2)}{(L+1)!\Gamma(L+a+b+2)}. \tag{87}$$

*Proof.* The exactness on $\mathbb{P}_{2N}$ and the positivity of the weights follow from Theorem 5, because $J^N_{(\alpha+1,\beta)}$ is orthogonal with respect to $\widetilde{W}_{\alpha,\beta} = (1-\tau)W_{\alpha,\beta}$. It remains only to write the weights in shifted Jacobi form.

For the endpoint node and the standard shifted Jacobi value, the Radau formula gives

$$\begin{aligned}\gamma_0 &= \frac{1}{J^N_{(\alpha+1,\beta)}(1)} \int_0^1 J^N_{(\alpha+1,\beta)}(\tau)(1-\tau)^\alpha \tau^\beta \, d\tau \\ &= \frac{N!\Gamma(\alpha+2)}{\Gamma(N+\alpha+2)} \frac{\Gamma(\alpha+1)\Gamma(N+\beta+1)}{\Gamma(N+\alpha+\beta+2)} = \frac{(\alpha+1)\Gamma^2(\alpha+1)N!\Gamma(N+\beta+1)}{\Gamma(N+\alpha+2)\Gamma(N+\alpha+\beta+2)}.\end{aligned} \tag{88}$$

For the interior nodes, set $a = \alpha+1$ and $b = \beta$. The general Radau weight formula gives $\gamma_j = \frac{1}{1-\tau_j} \frac{G^{a,b}_{N-1}}{J^{N-1}_{(a,b)}(\tau_j)\partial_\tau J^N_{(a,b)}(\tau_j)}$, for $1 \le j \le N$. Since $\tau_j$ is a zero of $J^N_{(a,b)}$, the shifted Jacobi recurrence yields

$$J^{N-1}_{(a,b)}(\tau_j) = \frac{2N+a+b}{(N+a)(N+b)} \tau_j(1-\tau_j)\partial_\tau J^N_{(a,b)}(\tau_j).$$

Substituting this identity into the previous formula gives

$$\begin{aligned}\gamma_j &= \frac{1}{1-\tau_j} \frac{G^{a,b}_{N-1}}{J^{N-1}_{(a,b)}(\tau_j)\partial_\tau J^N_{(a,b)}(\tau_j)} \\ &= \frac{1}{(1-\tau_j)^2\tau_j} \frac{(N+a)(N+b)}{2N+a+b} \frac{G^{a,b}_{N-1}}{\left[\partial_\tau J^N_{(a,b)}(\tau_j)\right]^2} \\ &= \frac{1}{(1-\tau_j)^2\tau_j} \frac{\widetilde{G}^{a,b}_{N-1}}{\left[\partial_\tau J^N_{(a,b)}(\tau_j)\right]^2}.\end{aligned} \tag{89}$$

Taking $a = \alpha+1$ and $b = \beta$ gives the stated formula. □

The parameters $\alpha$ and $\beta$ control the distribution of the roots. Increasing the weight near one endpoint moves the corresponding quadrature nodes toward that part of the interval. In the spectral approximation problem, this gives a controlled way to place more nodes near the region where the relevant eigenvalues concentrate. This improves the numerical resolution of the surrogate without changing the certification principle.

The entropy bound is reduced to a lower approximation of the logarithm. In the SDI energy-constrained problem, the nonlinear spectral term is still the same function $\ln x$. The change is the rational surrogate used to approximate $\ln x$.

For $x > 0$, write

$$\ln x = \int_0^1 \frac{x-1}{1+\tau(x-1)} \, d\tau = \int_0^1 g_x(\tau) \, d\tau, \tag{90}$$

where $g_x(\tau) = \frac{x-1}{1+\tau(x-1)}$. For a positive operator $A$, each quadrature node gives the resolvent-type term

$$g_A(\tau) = (A-I)\big(I+\tau(A-I)\big)^{-1},$$

which is the SDP-compatible structure.

The Jacobi-Radau construction is used only to choose a variable set of nodes and coefficients. For fixed $m, \alpha, \beta$, define the shifted Jacobi-node surrogate

$$r^{\,m}_{\alpha,\beta}(x) = \sum_{i=1}^m \gamma_i^{(\alpha,\beta)} \frac{x-1}{1+\tau_i^{(\alpha,\beta)}(x-1)}, \qquad e^{\,m}_{\alpha,\beta}(x) = \ln x - r^{\,m}_{\alpha,\beta}(x). \tag{91}$$

Here $\tau_i^{(\alpha,\beta)}$ are the Jacobi–Radau nodes and $\gamma_i^{(\alpha,\beta)}$ are the associated weights. The admissibility question is now simply whether error $e^{\,m}_{\alpha,\beta}(x) \ge 0$ on the spectral interval used in the SDP.

**Lemma 6** (variational node admissibility)**.** *Let $\mathcal{D} \subset (0,1)$ be the spectral interval used in the SDP, and let $g_x(\tau) = \frac{x-1}{1+\tau(x-1)}$. For fixed $m, \alpha, \beta$, let $\{\tau_i\}_{i=1}^m$, with $\tau_m = 1$, be the Jacobi–Radau nodes, and let $r^m_{\alpha,\beta}$ be the corresponding rational surrogate $r^m_{\alpha,\beta}(x) = \sum_{i=1}^m \gamma_i^{(\alpha,\beta)} g_x(\tau_i)$. Assume that, for every $x \in \mathcal{D}$, the error admits the signed Radau remainder representation*

$$\ln x - r^m_{\alpha,\beta}(x) = \frac{g_x^{(2m-1)}(\xi_x)}{(2m-1)!} \int_0^1 \phi_{m-1}(\tau)^2(\tau-1)\, d\mu_{\alpha,\beta}(\tau), \qquad \xi_x \in (0,1), \tag{92}$$

*where $\phi_{m-1}(\tau) = \prod_{i=1}^{m-1}(\tau - \tau_i)$, $d\mu_{\alpha,\beta}(\tau) = (1-\tau)^\alpha \tau^\beta\, d\tau$. Then*

$$r_{m,\alpha,\beta}(x) \le \ln x, \qquad x \in \mathcal{D}.$$

*Hence $r_{m,\alpha,\beta}$ is an admissible logarithmic surrogate on $\mathcal{D}$.*

*Proof.* The proof is a sign check for the Radau remainder. For $x \in (0,1)$, we have $1 + \tau(x-1) > 0$ and $0 < \tau < 1$. Moreover,

$$\begin{aligned} g_x^{(2m-1)}(\tau) &= \frac{d^{2m-1}}{d\tau^{2m-1}} \left( \frac{x-1}{1+\tau(x-1)} \right) \\ &= -(2m-1)! \frac{(x-1)^{2m}}{\big(1+\tau(x-1)\big)^{2m}} < 0. \end{aligned} \tag{93}$$

The integral factor in Eq. (92) also has a fixed negative sign:

$$\int_0^1 \phi_{m-1}(\tau)^2(\tau-1) d\mu_{\alpha,\beta}(\tau) = -\int_0^1 \phi_{m-1}(\tau)^2 (1-\tau)(1-\tau)^\alpha \tau^\beta\, d\tau < 0. \tag{94}$$

Indeed, $\phi_{m-1}^2 \ge 0$, $(1-\tau) > 0$, and $(1-\tau)^\alpha\tau^\beta > 0$ on $(0,1)$. Therefore the two factors in Eq. (92) are both non-positive, and in fact strictly negative except in degenerate cases. Hence

$$\ln x - r^m_{\alpha,\beta}(x) = \frac{g_x^{(2m-1)}(\xi_x)}{(2m-1)!} \int_0^1 \phi_{m-1}(\tau)^2(\tau-1) d\mu_{\alpha,\beta}(\tau) \ge 0. \tag{95}$$

Thus $r^m_{\alpha,\beta}(x) \le \ln x$ for all $x \in \mathcal{D}$. □

With the admissibility of the Jacobi-node surrogate certified, we now explain why the variational node construction can improve the entropy bound in the SDI energy-constrained regime. The difference is the rational lower approximant used for $\ln x$. Hence the comparison is made at the level of admissible logarithmic surrogates on the spectral region reached by the energy-constrained feasible set.

Let $\mathcal{D}_{\varepsilon,\omega} \subset (0,1]$ denote the active spectral interval for the SDI problem with SV bias $\varepsilon$ and energy bound $\omega$. For fixed quadrature order $m$, define the admissible parameter set

$$\mathcal{A}_m(\mathcal{D}_{\varepsilon,\omega}) := \left\{ (\alpha,\beta) \;\middle|\; \alpha,\beta > -1,\ r^m_{\alpha,\beta}(x) \le \ln x \text{ for all } x \in \mathcal{D}_{\varepsilon,\omega} \right\}.$$

The Legendre choice corresponds to $(\alpha,\beta) = (0,0)$.

**Theorem 7** (Variational improvement)**.** *Let $H^{m,(\alpha,\beta)}_{\varepsilon,\omega}(B|X,E,\Lambda)$ be the certified lower bound obtained by replacing $\ln x$ with the admissible surrogate $r^m_{\alpha,\beta}(x)$ in the variational objective. If $(0,0) \in \mathcal{A}_m(\mathcal{D}_{\varepsilon,\omega})$, then*

$$\sup_{(\alpha,\beta)\in\mathcal{A}_m(\mathcal{D}_{\varepsilon,\omega})} H^{m,(\alpha,\beta)}_{\varepsilon,\omega}(B|X,E,\Lambda) \ge H^{m,(0,0)}_{\varepsilon,\omega}(B|X,E,\Lambda). \tag{96}$$

*Moreover, if there exists $(\alpha_\star,\beta_\star) \in \mathcal{A}_m(\mathcal{D}_{\varepsilon,\omega})$ such that $r^m_{0,0}(x) \le r^m_{\alpha_\star,\beta_\star}(x) \le \ln x$, $x \in \mathcal{D}_{\varepsilon,\omega}$, then*

$$H^{m,(0,0)}_{\varepsilon,\omega}(B|X,E,\Lambda) \le H^{m,(\alpha_\star,\beta_\star)}_{\varepsilon,\omega}(B|X,E,\Lambda). \tag{97}$$

*Proof.* The first statement follows because the admissible parameter set contains the Legendre choice. Therefore optimizing over the larger admissible class cannot give a smaller certified value:

$$\sup_{(\alpha,\beta)\in\mathcal{A}_m(\mathcal{D}_{\epsilon,\omega})} H^{m,(\alpha,\beta)}_{\varepsilon,\omega}(B|X,E,\Lambda) \ge H^{m,(0,0)}_{\varepsilon,\omega}(B|X,E,\Lambda), \tag{98}$$

since $(0,0) \in \mathcal{A}_m(\mathcal{D}_{\epsilon,\omega})$.

For the second part, we have $r^m_{0,0}(x) \le r^m_{\alpha_\star,\beta_\star}(x) \le \ln x$ on $\mathcal{D}_{\varepsilon,\omega}$. By functional calculus, for every positive operator $A$ with $\sigma(A) \subseteq \mathcal{D}_{\varepsilon,\omega}$,

$$r^m_{0,0}(A) \preceq r^m_{\alpha_\star,\beta_\star}(A) \preceq \ln A. \tag{99}$$

The entropy objective is order-preserving in the logarithmic surrogate, as in Theorem 4. Hence replacing $r^m_{0,0}$ by the larger admissible surrogate $r^m_{\alpha_\star,\beta_\star}$ gives

$$H^{m,(0,0)}_{\varepsilon,\omega}(B|X,E,\Lambda) \le H^{m,(\alpha_\star,\beta_\star)}_{\varepsilon,\omega}(B|X,E,\Lambda) \le H_{\varepsilon,\omega}(B|X,E,\Lambda). \tag{100}$$

If the surrogate ordering is strict on a spectral region that contributes with nonzero weight in the optimizing solution, then the first inequality is strict. □

## SUPPLEMENTARY NOTE 3: VON NEUMANN ENTROPY BOUNDS IN SDI SCENARIOS

In this section, we derive the single-round conditional von Neumann entropy bound used in the finite-size analysis. The goal is to lower bound $H(B|X,E,\Lambda)$ for one use of the SDI prepare-and-measure device under an energy constraint. This quantity is different from the single-round min-entropy. It is the entropy that can be converted into smooth min-entropy in the many round model against attacks under IID assumption.

We consider the binary-input, binary-output SDI randomness certification scenario introduced in Supplementary Note 1. In one round, an imperfect $\varepsilon$-biased SV source generates an input $x \in \{0,1\}$. The preparation device sends a quantum system $S$ to the measurement device, which applies an unknown POVM $\{M_b\}_{b\in\{0,1\}}$ and outputs $b \in \{0,1\}$. Eve hold classical side information $\lambda \in \Lambda$ about the weak source, as well as a quantum system $E$ correlated with the preparation. The SV source condition is

$$\frac{1}{2} - \varepsilon \le \mu_{X_i|\Lambda,X_1,\ldots,X_{i-1}}(x_i|\lambda,x_1,\ldots,x_{i-1}) \le \frac{1}{2} + \varepsilon, \quad x \in \{0,1\}. \tag{101}$$

For the single-round optimization, this reduces to the constraint $\frac{1}{2} - \varepsilon \le \mu(x|\lambda) \le \frac{1}{2} + \varepsilon,\ x \in \{0,1\}$.

A general single-round realization is described by

$$\rho_{X\Lambda SE} = \sum_{x,\lambda} p(\lambda)\mu(x|\lambda)|x,\lambda\rangle\langle x,\lambda|_{X\Lambda} \otimes \rho^{SE}_{x,\lambda}, \qquad p(b|x,\lambda) = \mathrm{Tr}\left[(M_b \otimes \mathbb{1}_E)\rho^{SE}_{x,\lambda}\right]. \tag{102}$$

The observed conditional behaviour $p_{B|X}$ is obtained after averaging over the source side information $\lambda$. The only trusted physical assumption is the energy bound on the transmitted system. Writing $\rho^S_{x,\lambda} = \mathrm{Tr}_E\, \rho^{SE}_{x,\lambda}$, we impose $\mathrm{Tr}\big[\hat{H}\rho^S_{x,\lambda}\big] \le \omega_x$.

To certify randomness, we use the measurement-dependent local score. Under the standard normalization used in the SDI protocol, the corresponding MDL score is

$$I^\omega_\varepsilon(p_{B,X}) = v^\omega_\varepsilon\big[p(0,0) + p(1,1)\big] - \frac{1}{v^\omega_\varepsilon}\big[p(1,0) + p(0,1)\big]. \tag{103}$$

where $v^\omega_\varepsilon = \left(\frac{1}{4} - \varepsilon^2\right)\omega_0\omega_1$. This MDL-score is the constraint used in the single-round entropy optimization. For a fixed observed value of $I^\omega_\varepsilon$, we minimize $H^{(\varepsilon,\omega)}(B|X,E,\Lambda)$ over all SDI realizations satisfying the SV-source constraints, the energy constraints, and the POVM constraints.

To identify the relevant quantum region, we consider the smallest value of the score achievable inside the energy-constrained set. Following the optimization framework of the guessing probability, we define

$$I^\omega_{\varepsilon,\min} = \min_{p_{B|X}\in\mathcal{Q}^{\omega'},\ \omega'} I^\omega_\varepsilon(p_{B|X}), \tag{104}$$

$$\omega_{\varepsilon,\min} = \arg\min_{\omega'} \min_{p_{B|X}\in\mathcal{Q}^{\omega'}} I^\omega_\varepsilon(p_{B|X}). \tag{105}$$

Here $\mathcal{Q}^{\omega'}$ denotes the set of SDI behaviours compatible with the energy budget $\omega'$. The first quantity gives the optimal score value obtained in the energy scan, while $\omega_{\varepsilon,\min}$ denotes the corresponding energy point.

We next describe the single-round min-entropy certification used as a baseline. The conditional min-entropy is written as $H_{\min}(B|X,E,\Lambda) = -\log_2 P_{\mathrm{g}}$, where $P_{\mathrm{g}}$ is Eve's optimal guessing probability. We model Eve's classical

information about the output strategy by a variable $s \in \{0,1\}^2$. Each value of $s$ labels one guessing branch for the two possible inputs. Let $\tilde{\mathcal{Q}}^{\tilde{\omega}}$ denote the set of subnormalized behaviours admitting positive semidefinite operators $\tilde{\rho}_x \succeq 0$ and a POVM $\{M_b\}$ such that $\tilde{p}(b|x) = \mathrm{Tr}[M_b \tilde{\rho}_x]$, $\mathrm{Tr}[\hat{H}\,\tilde{\rho}_x] \leq \tilde{\omega}_x$.

Given an observed MDL score $I$, Eve's guessing probability is defined by

$$P_{\mathrm{g}}^{(\varepsilon,\omega,I)}(B|X,E,\lambda) \;\leq\; \max_{\{\tilde{p}^s,\tilde{\omega}^s\}_{s\in\{0,1\}^2}} \sum_{s\in\{0,1\}^2} \sum_{x\in\{0,1\}} \mu_X(x|s) \max_{b\in\{0,1\}} \tilde{p}^s(b|x) \tag{106}$$

subject to

$$\begin{aligned} &\sum_{s\in\{0,1\}^2} I_{\varepsilon}^{\omega}(\tilde{p}^s_{B|X}) \leq I, \qquad \tilde{p}^s_{B|X} \in \tilde{\mathcal{Q}}^{\tilde{\omega}^s}, \\ \sum_{s\in\{0,1\}^2} \tilde{\omega}^s_x \leq \omega_x, \qquad &\sum_{s\in\{0,1\}^2} \sum_{b\in\{0,1\}} \tilde{p}^s(b|x) = 1, \quad \forall x \in \{0,1\}, \quad \forall s \in \{0,1\}^2. \end{aligned} \tag{107}$$

The subnormalized branches $\tilde{p}^s$ encode Eve's possible guessing strategies, while the constraints ensure that their sum is compatible with the observed score and with the total energy budget. The biased input distribution $\mu_X(x|s)$ is restricted by the SV condition, $\frac{1}{2} - \varepsilon \leq \mu_X(x|s) \leq \frac{1}{2} + \varepsilon$. Solving the SDP in Eqs. (106)–(107) gives an upper bound on $P_{\mathrm{g}}$, and therefore a lower bound on the single-round min-entropy.

This min-entropy calculation is a valid single-round certification method. However, it is often conservative for many-round security. In the IID-attack finite-size analysis, the relevant rate is instead governed by the conditional von Neumann entropy $H_{\varepsilon,\omega}(B|X,E,\Lambda)$. We therefore turn next to a lower bound on $H_{\varepsilon,\omega}(B|X,E,\Lambda)$. The main difficulty is that the von Neumann entropy is nonlinear in the underlying state. To make the optimization computable, we rewrite it in relative-entropy form and approximate the logarithm by SDP-compatible integral and quadrature methods.

The basic identity is that conditional entropy can be written as a negative relative entropy. For a classical-quantum state

$$\rho_{BE} = \sum_b p_b |b\rangle\langle b|_B \otimes \rho_E^b, \qquad \rho_E = \sum_b p_b \rho_E^b, \tag{108}$$

one has $H(B|E)_\rho = -D(\rho_{BE} \,\|\, \mathbb{1}_B \otimes \rho_E)$ . Thus an upper bound on the relative entropy gives a lower bound on the conditional entropy.

In the SDI protocol, the input is generated by the weak source and can be correlated with Eve's classical information $\lambda$. Therefore the single-round entropy rate is the average

$$H_{\varepsilon,\omega}(B|X,E,\Lambda)_\rho = \sum_\lambda p(\lambda) \sum_x \mu(x|\lambda) H_{\varepsilon,\omega}(B|E,\Lambda)_{\rho^{bE}_{x,\lambda}}. \tag{109}$$

It is therefore enough to lower bound the conditional entropy for each single-round branch $(x,\lambda)$, and then minimize the average over all realizations compatible with the observed score and the physical constraints.

The single-round realizations are constrained by the observed behaviour, the SV-source condition, the energy bound and the measurement relations. Together with the score constraint $I_\varepsilon^\omega(p_\rho(b,x)) \leq I_{\mathrm{th}}$, these conditions define the single-round feasible set used in the entropy optimization.

The target single-round von Neumann entropy minimization is

$$\inf_{\rho\in\mathcal{Q}^\omega:\, I_\varepsilon^\omega(p_\rho(b,x))\leq I_{\mathrm{th}}} H_{\varepsilon,\omega}(B|X,E,\Lambda)_\rho. \tag{110}$$

The difficulty is that the objective is nonlinear because of the logarithm in the relative entropy. We handle this by replacing the logarithm with a one-sided rational surrogate constructed from the variational node quadrature in Supplementary Note 2. The surrogate is chosen with the correct operator inequality direction, so the resulting SDP gives a certified lower bound on $H_{\varepsilon,\omega}(B|X,E,\Lambda)$.

At the optimization level, the SDI constraints are relaxed to a finite semidefinite program. Schematically, one introduces a moment matrix $\Gamma$ whose entries are expectations of products of the relevant preparation, measurement and energy operators. Any admissible quantum realization gives $\Gamma \succeq 0$, together with linear constraints encoding normalization, observed probabilities, the POVM relations and the energy bound. In schematic form,

$$\sum_b M_b = \mathbb{1}, \qquad \langle M_b\rangle_{x,\lambda} = p(b|x,\lambda), \qquad \langle \hat{H}\rangle_{x,\lambda} \leq \omega_x. \tag{111}$$

After substituting the rational entropy surrogate, the original nonlinear entropy minimization becomes a finite-dimensional SDP over this relaxed moment set. Increasing the relaxation level and the quadrature order tightens the certified lower bound.

Putting these ingredients together, we obtain a certified lower bound on the single-round conditional von Neumann entropy. The optimization problem minimizes the Gauss–Radau entropy surrogate over all SDI realizations compatible with the observed score, the SV-source constraint, and the energy bound:

$$
\begin{aligned}
H_{\varepsilon,\omega}(B|X,E,\Lambda) \geq \min_{\{\rho_x^S\},\,\{M_b\}} & \sum_{x\in\{0,1\}} \mu_X(x|\lambda) H^m_{\varepsilon,\omega}(B|X,E,\Lambda) \\
\text{s.t.} \quad & \mu_X(x|\lambda) \in \mathsf{SV}_\varepsilon, \qquad p_{B|X}(b|x) \in \mathcal{Q}^\omega, \qquad \forall x, \\
& \mathrm{Tr}[\hat{H}\rho_x^S] \leq \omega_x, \qquad \forall x, \\
& I_\varepsilon^\omega\big(\{\mathrm{Tr}[M_b\rho_x^S]\}_{b,x}\big) \leq I^\omega_{\varepsilon,\min}.
\end{aligned}
\tag{112}
$$

Here $\mathcal{Q}^\omega$ denotes the chosen SDP relaxation $\mathcal{Q}^\omega_{(\ell)}$, determined by the prepare-and-measure device model and the energy-overlap constraints. The notation $\mathsf{SV}_\varepsilon$ denotes the linear constraints imposed by the $\varepsilon$-biased SV source.

The quantity $H^{m,(\alpha,\beta)}_{\varepsilon,\omega}(B|X,E,\Lambda)$ is the order-$m$ variational quadrature approximation of the conditional entropy. We now describe how this term is implemented as an SDP. Fix an $m$-point variational quadrature rule on $[0,1]$, with nodes $\{\tau_i\}_{i=1}^m$ and weights $\{\gamma_i\}_{i=1}^m$, where $\tau_m = 1$ and $\gamma_i > 0$. The logarithmic term in the relative-entropy representation is replaced by the corresponding rational surrogate. As in the standard SDP construction, we introduce auxiliary operators $\{Z_{b,i}\}$, indexed by the outcome $b$ and the quadrature node $i$, to linearize the resolvent terms.

The resulting noncommutative polynomial optimization takes the form

$$
\begin{aligned}
H^{m,(\alpha,\beta)}_{\varepsilon,\omega}(B|X,E,\Lambda) \;\geq\; & \min_{\{\rho_x, M_b, \hat{H}, Z_{b,i}\}} \sum_x \mu_X(x|\lambda)\left(c_m + \sum_{i=1}^{m-1} \frac{\gamma_i}{\tau_i \ln 2} \mathrm{Tr}\left[\tau_i Z_{b,i} Z^*_{b,i} + \rho_x\left(Z_{b,i} + Z^*_{b,i} + (1-\tau_i) Z_{b,i} Z^*_{b,i}\right)\right]\right) \\
& \text{subject to:} \quad \rho_x \succeq 0, \quad \mathrm{Tr}[\rho_x] = 1, \quad \rho_x - \rho_x^2 \succeq 0\,; \quad M_b \succeq 0, \quad \sum_b M_b = \mathbb{1}, \\
& \qquad \mathrm{Tr}[\rho_x M_b] = p(b|x) \quad \forall x, b\,; \quad \mathrm{Tr}[\hat{H}\rho_x] \leq \omega \quad \forall x, \\
& \qquad I_\varepsilon^\omega\big(\{\mathrm{Tr}[M_b\rho_x^S]\}_{b,x}\big) \leq I^\omega_{\varepsilon,\min}
\end{aligned}
\tag{113}
$$

where $c_m = \sum_{i=0}^{m-1} \frac{\gamma_i}{\tau_i \ln 2}$. The auxiliary constraint on $Z_{b,i}$ encodes the resolvent relation appearing in the rational approximation of the logarithm. Since these auxiliary operators are functions of the corresponding preparation operator, they commute with the preparations, $[Z_{b,i}, \rho_x] = 0$.

The constraint $\rho_x - \rho_x^2 \succeq 0$ allows mixed preparations. This is important in the SDI setting because the transmitted system may be correlated with an adversarial purifying system. Thus the optimization does not restrict the preparation to be pure.

The noncommutative polynomial problem in Eq. (113) is then relaxed to a hierarchy of SDPs using tracial moment matrices. A single global moment matrix containing all auxiliary operators $Z_{b,i}$ is typically too large. We therefore use a block decomposition: instead of imposing positivity on one global matrix, we impose positivity separately on blocks $\Gamma_i$, $i = 1, \ldots, m$, where the $i$-th block contains only the operators associated with the fixed quadrature node $\tau_i$.

For the relaxation level $k = 2 + \rho M\rho + HM\rho + \rho\rho\rho$, define the generating set of monomials for the $i$-th block by

$$
\begin{aligned}
\{m_{k=(2+\rho M\rho + HM\rho+\rho\rho\rho)}\}_i = \{\mathbb{1}, \{Z_{b,i}\}_b\} \times \Big\{ & \mathbb{1}, \{\rho_x\}_x, \{M_b\}_b, H, \{\rho_x\rho_{x'}\}_{x,x'}, \{\rho_x M_b\}_{x,b}, \\
& \{M_b\rho_x\}_{b,x}, \{M_b M_{b'}\}_{b,b'}, \{H\rho_x\}_x, \{\rho_x H\}_x, \\
& \{HM_b\}_b, \{M_b H\}_b, H^2, \{\rho_x M_b \rho_{x'}\}_{x,b,x'}, \\
& \{HM_b\rho_x\}_{b,x}, \{\rho_x M_b H\}_{x,b}, \{\rho_x\rho_{x'}\rho_{x''}\}_{x,x',x''} \Big\}.
\end{aligned}
\tag{114}
$$

This defines a moment matrix $\Gamma_i$ of size $|\mathcal{S}_i| \times |\mathcal{S}_i|$. The block matrices keep the relaxation tractable while retaining the local constraints needed for each quadrature node.

To impose the mixed-state constraint $\rho_x - \rho_x^2 \succeq 0$, we introduce localizing matrices $\Upsilon_{x,i}$, with entries $(\Upsilon_{x,i})_{u,v} = L_i\big(u^\dagger(\rho_x - \rho_x^2)v\big)$ , where $u, v$ range over a lower-level generating set. Positivity of these localizing matrices enforces the mixed-preparation constraint at the level of the moment relaxation. This formulation gives a certified single-round lower bound on the conditional von Neumann entropy in the SDI setting.

## SUPPLEMENTARY NOTE 4: RELAXATION OF THE NO-ENTANGLEMENT ASSUMPTION

Most SDI randomness-certification models assume that the preparation device $P$ and the measurement device $M$ do not share quantum entanglement before the protocol starts. Under this assumption, the devices may still be classically correlated through a shared random variable $\lambda$. If the internal memories of the preparation and measurement devices are described by Hilbert spaces $\mathcal{H}_{\mathcal{M}_A}$ and $\mathcal{H}_{\mathcal{M}_B}$, respectively, then the joint initial memory state is restricted to the separable form

$$\sigma_{PM} = \sigma^{\text{sep}}_{\mathcal{M}_A\mathcal{M}_B} = \sum_\lambda p_\lambda \, \rho^\lambda_{\mathcal{M}_A} \otimes \rho^\lambda_{\mathcal{M}_B}. \tag{115}$$

This is the strict no-entanglement model between the two devices.

In a composable cryptographic setting, this restriction is strong. A malicious manufacturer or Eve can distribute a joint quantum state to the two devices before the protocol begins. We therefore also consider a relaxed model in which the devices can share an arbitrary state

$$\sigma_{\mathcal{M}_A\mathcal{M}_B} \in \mathcal{D}(\mathcal{H}_{\mathcal{M}_A} \otimes \mathcal{H}_{\mathcal{M}_B}),$$

which may be entangled. Under an energy constraint, this relaxation enlarges the set of achievable prepare-and-measure correlations.

For each input $x \in X$, the preparation device applies a local channel $\Lambda_x : \mathcal{D}(\mathcal{H}_{\mathcal{M}_A}) \longrightarrow \mathcal{D}(\mathcal{H}_S)$ to its share of the memory and outputs the transmitted system $S$. The joint state of the transmitted system and Bob's memory is then

$$\rho^{S\mathcal{M}_B}_x = (\Lambda_x \otimes \mathbb{1}_{\mathcal{M}_B})(\sigma_{\mathcal{M}_A\mathcal{M}_B}). \tag{116}$$

The measurement device is now allowed to measure the joint system $S \otimes \mathcal{M}_B$. Thus, for a POVM $\{M_b\}_{b\in B}$ on $S \otimes \mathcal{M}_B$, the observed behaviour is

$$p(b|x) = \text{Tr}\left[M_b \, \rho^{S\mathcal{M}_B}_x\right]. \tag{117}$$

The trusted physical constraint is still imposed only on the transmitted system $S$. Writing

$$\rho^S_x = \text{Tr}_{\mathcal{M}_B}\left[\rho^{S\mathcal{M}_B}_x\right],$$

we require $\text{Tr}\left[\hat{H}_S \rho^S_x\right] \leq \omega_x,\ x \in X$. The memory system $\mathcal{M}_B$ is not counted as part of the transmitted system and is not subject to the energy budget. It can, however, be used jointly with $S$ in Bob's measurement. This is the operational difference between the strict no-entanglement model and the relaxed model.

The energy constraint makes this relaxation nontrivial. Without a trusted constraint on the transmitted system, the distinction between preshared memory and communicated quantum information becomes less sharp, since correlations could be absorbed into an enlarged effective system. With the energy bound in place, the amount of quantum information sent through $S$ remains physically limited, while preshared entanglement between $\mathcal{M}_A$ and $\mathcal{M}_B$ can still change the achievable correlation set. Thus the relaxed model captures a stronger adversarial scenario without removing the SDI physical constraint.

**Lemma 7** (Enlargement of the quantum set without no-entanglement assumption[11]). *For the energy-constrained prepare-and-measure scenario, allowing preshared entanglement between the preparation and measurement devices can strictly enlarge the achievable correlation set.*

$$\mathcal{Q}^{\text{sep}}_\omega \subsetneq \mathcal{Q}^{\text{ent}}_\omega. \tag{118}$$

*Consequently, there exist behaviours $p(b|x)$ satisfying the transmitted energy bound $\omega$ which are achievable in the relaxed no-entanglemen model but cannot be reproduced by any separable-memory strategy under the same energy constraint.*

Under the energy constraint, this changes the set of possible behaviours and can change the amount of certifiable randomness. The inclusion $\mathcal{Q}^{\text{sep}}_\omega \subsetneq \mathcal{Q}^{\text{ent}}_\omega$ has a direct consequence for composable randomness certification. If the trust model allows preshared entanglement between the devices, then an entropy lower bound obtained by optimizing only over $\mathcal{Q}^{\text{sep}}_\omega$ is not sound for the relaxed model. A behaviour can look randomness-certifying when it is tested only against the separable-memory set, while the same behaviour may have an relaxed no-entanglement realization in which Eve has much more information about the generated outcome. In some parameter regimes, this can reduce the certified randomness to zero even when the separable model still gives a positive bound.

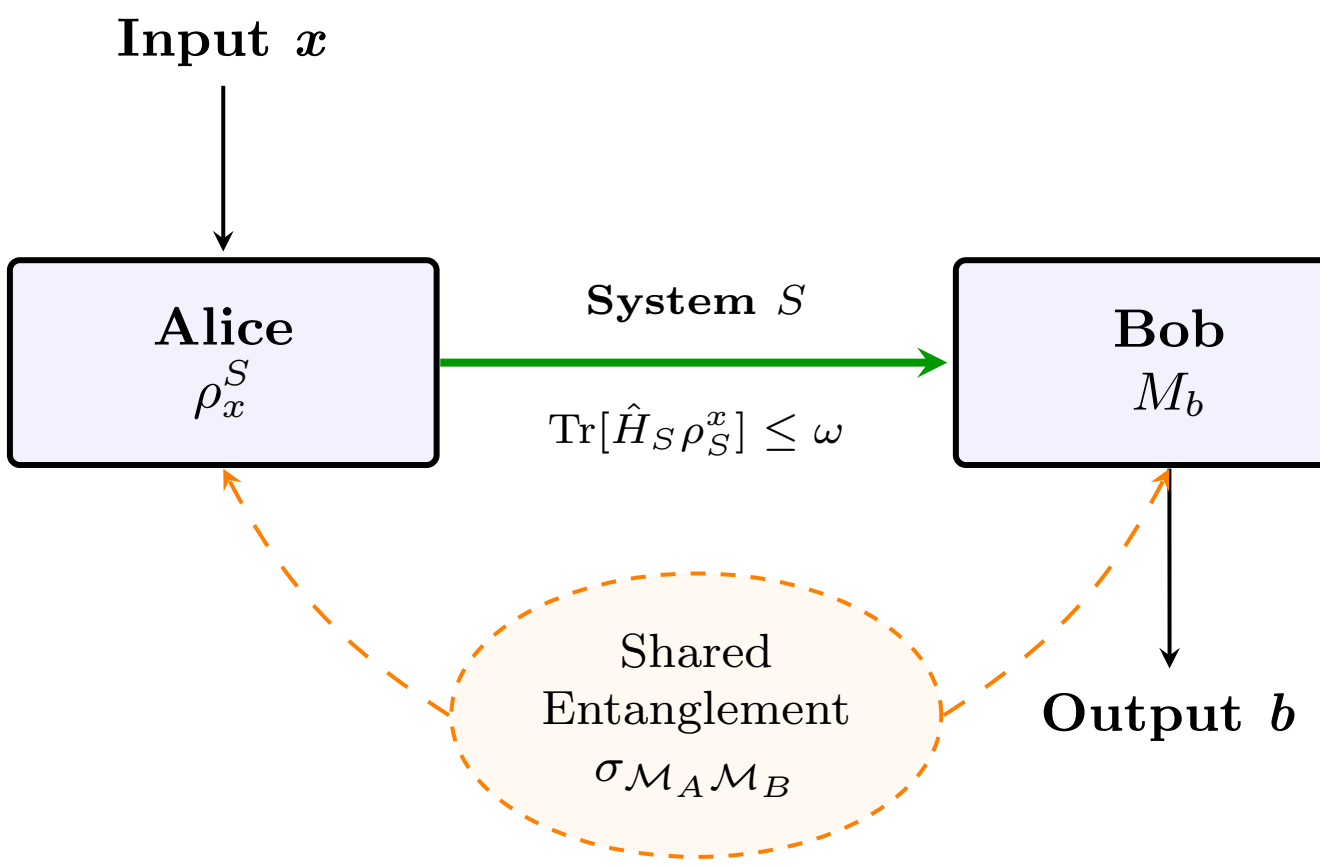


FIG. 2: Energy-constrained SDI prepare-and-measure scenario with pre-shared entanglement. Alice and Bob utilize the shared memory state $\sigma_{\mathcal{M}_A \mathcal{M}_B}$ to implement local operations dependent on input $x$ . Alice transmits the quantum system $S$ (thick green arrow), subject to the energy constraint $\omega$, which Bob measures to produce outcome $b$.

Therefore the feasible set in the entropy minimization must match the physical assumption. In the strict no-entanglement model, the outer minimization is taken over $\mathcal{Q}_\omega^{\mathrm{sep}}$. In the relaxed model, it must instead be taken over $\mathcal{Q}_\omega^{\mathrm{ent}}$. Since the relaxed set is larger, the certified entropy bound cannot increase:

$$\inf_{p \in \mathcal{Q}_\omega^{\mathrm{ent}}} H^{(\varepsilon,\omega)}(B|X,E,\Lambda)_p \leq \inf_{p \in \mathcal{Q}_\omega^{\mathrm{sep}}} H^{(\varepsilon,\omega)}(B|X,E,\Lambda)_p. \tag{119}$$

The relaxed-model bound can vanish in regimes where the separable-model bound is positive.

To formulate the relaxed optimization, it is useful to eliminate the explicit initial bipartite state and local preparation channels. Instead, we work with the effective post-preparation assemblage $\{\rho_x^{S\mathcal{M}}\}_{x \in X}$, where $S$ is the transmitted system and $\mathcal{M}$ is Bob's local memory. Since Alice's input-dependent operation acts only on her local share, it cannot change Bob's reduced memory state. Hence any relaxed strategy satisfies the shared-marginal condition

$$\mathrm{Tr}_S[\rho_x^{S\mathcal{M}}] = \sigma^{\mathcal{M}}, \qquad \forall x \in X. \tag{120}$$

Conversely, any family of positive operators $\{\rho_x^{S\mathcal{M}}\}_x$ with the same marginal $\sigma^{\mathcal{M}}$ can be realized by a suitable initial entangled state and local channels on Alice's side, after allowing a sufficiently large preparation memory. Thus the relaxed model can be parametrized directly by the assemblage $\{\rho_x^{S\mathcal{M}}\}_x$.

The physical constraints are then imposed on the assemblage

$$\begin{gathered} \rho_x^{S\mathcal{M}} \succeq 0, \qquad \mathrm{Tr}[\rho_x^{S\mathcal{M}}] = 1, \\ \mathrm{Tr}_S[\rho_x^{S\mathcal{M}}] = \sigma^{\mathcal{M}}, \qquad \mathrm{Tr}\left[(\hat{H}_S \otimes \mathbb{1}_{\mathcal{M}})\rho_x^{S\mathcal{M}}\right] \leq \omega_x, \qquad \forall x. \end{gathered} \tag{121}$$

Bob is allowed to measure the joint system $S \otimes \mathcal{M}$, so the observed statistics are

$$p(b|x) = \mathrm{Tr}\left[M_b\, \rho_x^{S\mathcal{M}}\right], \qquad M_b \succeq 0, \qquad \sum_b M_b = \mathbb{1}_{S\mathcal{M}}. \tag{122}$$

It is treated as a noncommutative polynomial optimization and relaxed by a moment-matrix SDP hierarchy. In this situation, the shared-marginal constraint and the transmitted energy constraint remain linear constraints, while the measurement relations are enforced at the level of the moment relaxation.

Combining the assemblage reduction with the variational surrogate $H_{\varepsilon,\omega}^{m,(\alpha,\beta)}(B|X,E,\Lambda)$ from Supplementary

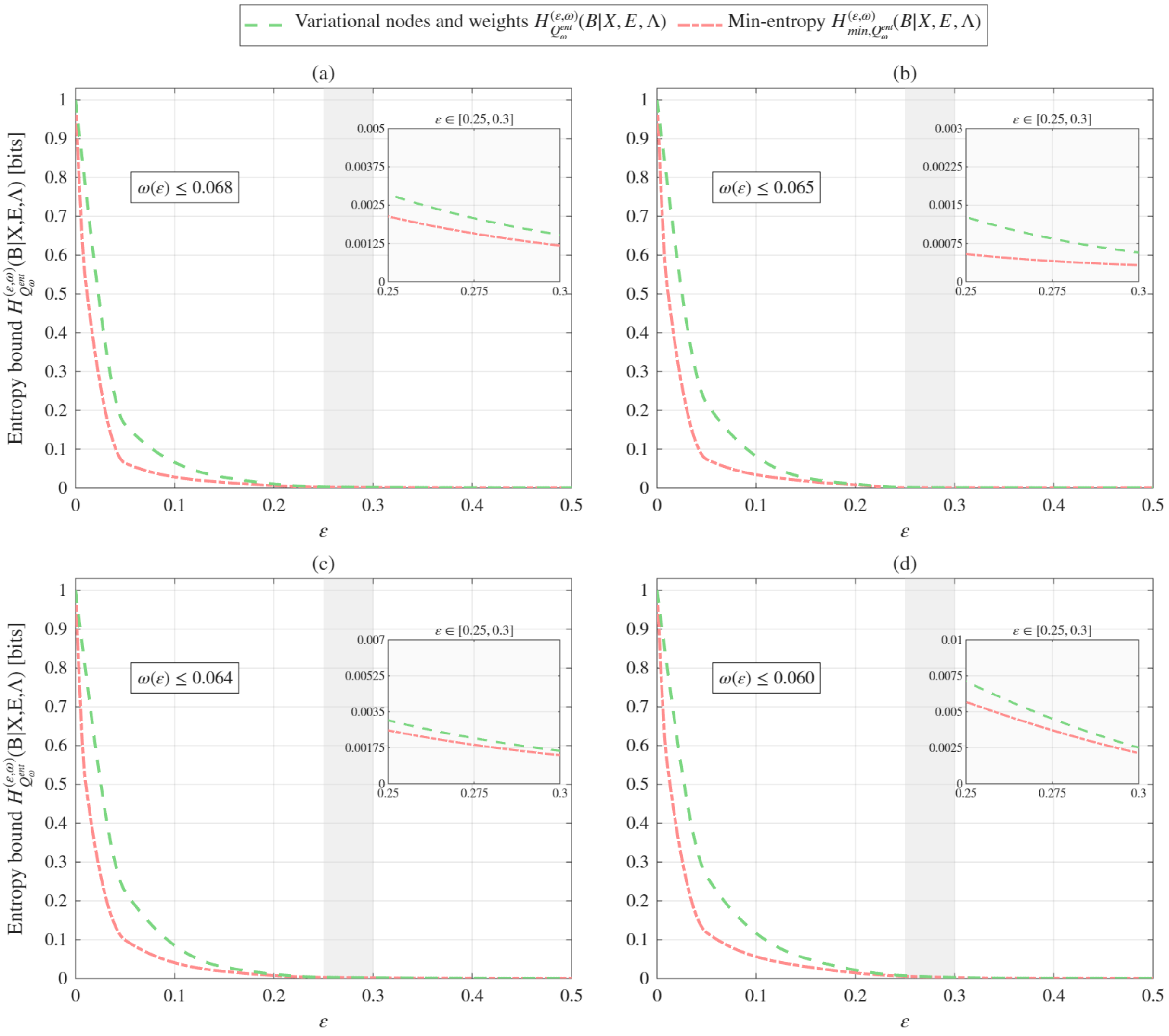


FIG. 3: **Entropy certification in the presence of preparation-measurement entanglement.** The pink dot-dashed curve depicts the certifiable min-entropy in this extended setting, showing a significant degradation due to the presence of entanglement. In contrast, the lower bound on the conditional von Neumann entropy evaluated with our method (green dashed curve) maintains relatively high values even in the presence of entanglement. Subplots (a)–(d) correspond to different values of the energy bound.

Note 3, we obtain the entropy minimization for the relaxed entanglement-assisted model:

$$\begin{aligned} H^{(\varepsilon,\omega)}_{\mathcal{Q}^{\rm ent}_\omega}(B|X,E,\Lambda) \geq \min_{\{\rho_x^{SM}\},\{M_b^{SM}\}} & \sum_{x\in\{0,1\}} \mu_X(x|\lambda) H^{m,(\alpha,\beta)}_{\varepsilon,\omega}(B|X,E,\Lambda) \\ \text{s.t.}\quad & \mu_X(x|\lambda) \in \mathrm{SV}_\varepsilon, \quad \rho_x^{S\mathcal{M}} \succeq 0, \quad \mathrm{Tr}[\rho_x^{S\mathcal{M}}] = 1, \qquad \forall x, \\ & M_b^{S\mathcal{M}} \succeq 0, \qquad \sum_b M_b^{S\mathcal{M}} = \mathbb{1}_{S\mathcal{M}}, \\ & \rho_x^S = \mathrm{Tr}_M[\rho_x^{S\mathcal{M}}], \qquad \mathrm{Tr}[\hat{H}_S \rho_x^S] \leq \omega_x, \qquad \forall x, \\ & \mathrm{Tr}_S[\rho_0^{S\mathcal{M}}] = \mathrm{Tr}_S[\rho_1^{S\mathcal{M}}], \\ & I^\omega_\varepsilon \left( \{\mathrm{Tr}[M_b^{S\mathcal{M}} \rho_x^{S\mathcal{M}}]\}_{b,x} \right) \leq I^\omega_{\varepsilon,\min}. \end{aligned} \tag{123}$$

Here $M_b^{S\mathcal{M}}$ is a POVM on the joint system $S \otimes \mathcal{M}$, and $\rho_x^S = \mathrm{Tr}_M[\rho_x^{S\mathcal{M}}]$ is the transmitted reduced state. The condition

$$\mathrm{Tr}_S[\rho_0^{S\mathcal{M}}] = \mathrm{Tr}_S[\rho_1^{S\mathcal{M}}]$$

is the shared-marginal constraint. It expresses the fact that Alice's input-dependent operation acts only locally and therefore cannot change Bob's pre-existing memory state.

Consequently, the relaxed entropy bound is valid for the entanglement-assisted trust model. Since $\mathcal{Q}_\omega^{\mathrm{sep}} \subseteq \mathcal{Q}_\omega^{\mathrm{ent}}$, the relaxed-model entropy lower bound can only be smaller than, or equal to, the bound obtained under the strict no-entanglement assumption.

The results of the minimization are shown in Fig. 3. In the figure, one can see that the standard min-entropy bound takes a significantly lower value when evaluated over the set of relaxed no-entanglement strategies. On the other hand, our variational Gauss-Radau method provides better entropy estimates, which guarantee a higher rate of entropy amplification in the end of the protocol. Our method yields non-zero entropy rates even in physical regimes where the standard min-entropy bound approaches zero.

## SUPPLEMENTARY NOTE 5: FINITE-SIZE SECURITY UNDER THE IID ASSUMPTION

In this section, we prove a finite-size security statement for the SDI energy-constrained randomness amplification protocol under independent and identically distributed (IID) attack model. The goal is to convert the single-round conditional von Neumann entropy bound into a smooth min-entropy bound for the raw output string $B^n$. In the IID model, the single-round prepare-and-measure realization is used independently in every round. The inter-round memory and adaptive non-IID behaviour inside the devices are excluded in this restricted model.

For one use of the prepare-and-measure device, the weak source provides an input $x \in X$, Eve can hold classical side information $\lambda$. The preparation device sends a system $S$, and Eve also holds a quantum register $E$. Before the measurement, the single-round classical-quantum state is

$$\rho_{X\Lambda SE} = \sum_{x,\lambda} p(\lambda)\mu(x|\lambda)|x,\lambda\rangle\langle x,\lambda|_{X\Lambda} \otimes \rho_{x,\lambda}^{SE}. \tag{124}$$

The measurement device applies an unknown POVM $\{M_b\}_{b\in B}$ on the transmitted system $S$. After the measurement, the single-round state is

$$\begin{gathered}\rho_{BXE\Lambda} = \sum_{b,x,\lambda} p(\lambda)\mu(x|\lambda)|b,x,\lambda\rangle\langle b,x,\lambda|_{BX\Lambda} \otimes \rho_{b,x,\lambda}^{E},\\ \rho_{b,x,\lambda}^{E} = \mathrm{Tr}_S\left[(M_b^{1/2}\otimes \mathbb{1}_E)\rho_{x,\lambda}^{SE}(M_b^{1/2}\otimes \mathbb{1}_E)\right], \qquad p(b|x,\lambda) = \mathrm{Tr}\,\rho_{b,x,\lambda}^{E} = \mathrm{Tr}\left[(M_b\otimes \mathbb{1}_E)\rho_{x,\lambda}^{SE}\right].\end{gathered} \tag{125}$$

The operators $\rho_{b,x,\lambda}^{E}$ are sub-normalized. Their traces give the conditional output probabilities for the branch $(x,\lambda)$.

An IID attack model is the repetition of this single-round realization. Before measurement, this means

$$\rho_{X^n\Lambda^nS^nE^n} = \rho_{X\Lambda SE}^{\otimes n}, \qquad M_{b^n} = M_{b_1}\otimes\cdots\otimes M_{b_n}. \tag{126}$$

Here $M_{b^n}$ describes the round-wise measurement performed by the device during the protocol. This does not restrict Eve's final action on her side information. Equivalently, after the device measurements, the cq state has the product form $\rho_{B^nX^nE^n\Lambda^n} = \rho_{BXE\Lambda}^{\otimes n}$.

Let $\hat{p}_{BX}(b,x) = \frac{1}{n}\left|\{i\in\{1,\dots,n\} : (B_i,X_i) = (b,x)\}\right|$ be the empirical behaviour observed in an accepted run. The raw string is $B^n$, and the same $n$ rounds are used to implement statistics test.

The finite-size analysis does not identify the exact single-round behaviour $p_\rho(b,x)$. It only gives a confidence region around the observed score. Let $\hat{I}_\varepsilon^\omega = I_\varepsilon^\omega(\hat{p}_{BX})$. We choose a statistical margin $\mu_I$ such that, for every IID strategy,

$$\Pr\left[I_\varepsilon^\omega(p_\rho) > \hat{I}_\varepsilon^\omega + \mu_I\right] \le \varepsilon_{\mathrm{stat}}, \qquad \Omega_{\mathrm{stat}} = \left\{I_\varepsilon^\omega(p_\rho) \le \hat{I}_\varepsilon^\omega + \mu_I\right\}. \tag{127}$$

Hence $\Pr[\Omega_{\mathrm{stat}}] \ge 1-\varepsilon_{\mathrm{stat}}$. On this good statistical event, the true single-round realization is compatible with the relaxed score constraint $I_\varepsilon^\omega(p_\rho) \le \hat{I}_\varepsilon^\omega + \mu_I$.

We write the single-round feasible set $\mathcal{F}_{\omega,\varepsilon}(\hat{I}_\varepsilon^\omega,\mu_I)$ as the set of single-round prepare-and-measure realizations $\rho = \Big(p(\lambda), \mu(x|\lambda), \{\rho_{x,\lambda}^{SE}\}_x, \{M_b\}_b\Big)$ satisfying the SV-source constraint, the transmitted energy constraint, the POVM

constraint, and the finite-statistical score constraint. In the scalar-score formulation, this means

$$\begin{aligned}\mathcal{F}_{\omega,\varepsilon}(\hat{I}^{\omega}_{\varepsilon},\mu_I) = \Big\{\rho : I^{\omega}_{\varepsilon}(p_\rho) \le \hat{I}^{\omega}_{\varepsilon} + \mu_I, \quad \mathrm{Tr}\Big[\hat{H}\,\rho_S^{x,\lambda}\Big] \le \omega_x, \quad \rho_S^{x,\lambda} = \mathrm{Tr}_E[\rho_{SE}^{x,\lambda}], \quad \forall x,\lambda, \\ \frac{1}{2} - \varepsilon \le \mu(x|\lambda) \le \frac{1}{2} + \varepsilon, \quad \forall x,\lambda, \quad M_b \succeq 0, \quad \sum_b M_b = \mathbb{1}\Big\}.\end{aligned} \tag{128}$$

Here $p_\rho$ denotes the single-round behaviour generated by the realization $\rho$.

The corresponding entropy rate is the worst-case single-round conditional von Neumann entropy over this feasible set

$$h = \inf_{\rho\in\mathcal{F}_{\omega,\varepsilon}(\hat{I}^{\omega}_{\varepsilon},\mu_I)} H^{(\varepsilon,\omega)}(B|X,E,\Lambda)_\rho = H(\hat{I}^{\omega}_{\varepsilon} + \mu_I). \tag{129}$$

Here the function $H(I)$ is the single-round entropy curve $H(I) = \inf_{\rho\in\mathcal{Q}^{\omega}:\, I^{\omega}_{\varepsilon}(p_\rho)\le I} H^{(\varepsilon,\omega)}(B|X,E,\Lambda)_\rho$. $\mathcal{Q}^{\omega}$ denotes the single-round quantum set compatible with the SDI energy budget and the SV-source input constraints. In practice, $H(I)$ is replaced by the computable SDP lower bound obtained from the single-round entropy method developed in the previous sections.

For every IID realization compatible with the good statistical event, the single-round state belongs to the feasible set defining $h$. The post-measurement $n$-round state is then $\rho_{B^nX^nE^n\Lambda^n} = \rho_{BXE\Lambda}^{\otimes n}$. Therefore the conditional von Neumann entropy is additive

$$\begin{aligned} H(B^n|X^n,E^n,\Lambda^n)_{\rho^{\otimes n}} &= \sum_{i=1}^{n} H(B_i|X_i,E_i,\Lambda_i)_\rho \\ &= n\,H(B|X,E,\Lambda)_\rho \\ &\ge n\,h = n\,H(\hat{I}^{\omega}_{\varepsilon} + \mu_I). \end{aligned} \tag{130}$$

Thus the multi-round von Neumann entropy bound follows directly from the IID product structure.

To use this entropy bound for extraction, we convert it into a smooth min-entropy bound by the fully quantum asymptotic equipartition property (AEP). In the considered scenario, the raw output register is $B$, and the side information is $(X,E,\Lambda)$.

**Theorem 8** (Fully quantum AEP [6]). *Let $\rho_{BXE\Lambda}$ be a single-round classical-quantum state, and assume that the $n$-round state is IID, $\rho_{B^nX^nE^n\Lambda^n} = \rho_{BXE\Lambda}^{\otimes n}$. Then*

$$\lim_{\varepsilon_{\mathrm{sm}}\to 0}\lim_{n\to\infty} \frac{1}{n} H_{\min}^{\varepsilon_{\mathrm{sm}}}(B^n|X^n,E^n,\Lambda^n)_{\rho^{\otimes n}} = H(B|X,E,\Lambda)_\rho. \tag{131}$$

*For finite $n$, we use the non-asymptotic form*

$$\begin{aligned} H_{\min}^{\varepsilon_{\mathrm{sm}}}(B^n|X^n,E^n,\Lambda^n)_{\rho^{\otimes n}} &\ge n\,H(B|X,E,\Lambda)_\rho - \Delta_{\mathrm{AEP}}(n,\varepsilon_{\mathrm{sm}},|B|), \\ \Delta_{\mathrm{AEP}}(n,\varepsilon_{\mathrm{sm}},|B|) &= \sqrt{n}\,\kappa_B\sqrt{\log_2\frac{2}{\varepsilon_{\mathrm{sm}}^2}}, \qquad \kappa_B = 2\log_2(1+2|B|). \end{aligned} \tag{132}$$

*Here $|B|$ is the alphabet size of the raw classical output register. The finite-size correction depends on $|B|$, $n$, and $\varepsilon_{\mathrm{sm}}$, but not on the dimension of Eve's quantum register.*

Applying Theorem 8 to IID realization compatible with the passed statistical event, and using the definition of $h$, gives

$$\begin{aligned} H_{\min}^{\varepsilon_{\mathrm{sm}}}(B^n|X^n,E^n,\Lambda^n)_{\rho^{\otimes n}} &\ge n\,h - \Delta_{\mathrm{AEP}}(n,\varepsilon_{\mathrm{sm}},|B|) \\ &= n\,H(\hat{I}^{\omega}_{\varepsilon} + \mu_I) - \Delta_{\mathrm{AEP}}(n,\varepsilon_{\mathrm{sm}},|B|). \end{aligned} \tag{133}$$

Thus the certified raw smooth min-entropy is

$$\begin{aligned} k_B &= n\,h - \Delta_{\mathrm{AEP}}(n,\varepsilon_{\mathrm{sm}},|B|) \\ &= n\,H(\hat{I}^{\omega}_{\varepsilon} + \mu_I) - \Delta_{\mathrm{AEP}}(n,\varepsilon_{\mathrm{sm}},|B|). \end{aligned} \tag{134}$$

We now turn the entropy bound into a finite-size security statement. The security requirement is secrecy: conditioned on acceptance, the final output should be close to uniform and independent of Eve's quantum side information, the weak seed, and the public transcript.

Let $\Omega_{\rm acc}$ be the event that the protocol accepts, and let $\Omega_{\rm stat}$ be the good statistical event defined in Eq. (127). By construction, $\Pr[\Omega^c_{\rm stat}] \le \varepsilon_{\rm stat}$.

**Theorem 9** (Finite-size security)**.** *Consider the SDI energy-constrained randomness amplification protocol under IID attacks. Let $k_B$ be defined by Eq. (134). Let* Ext *be the extractor used in the final step, with weak seed $Z$. Assume that the extractor has the following guarantee, whenever the raw string satisfies*

$$H^{\varepsilon_{\rm sm}}_{\min}(B^n|X^nE^n\Lambda^n) \ge k_B,$$

*the output $O = {\rm Ext}(B^n, Z)$ is $\varepsilon_{\rm ext}$-close to uniform against $X^nZE^n\Lambda^n$. Then the protocol is $\varepsilon_{\rm sec} = \varepsilon_{\rm stat} + \varepsilon_{\rm sm} + \varepsilon_{\rm ext}$ secure against attacks, in the sense that*

$$\frac{1}{2}\Pr[\Omega_{\rm acc}]\left\|\rho_{OX^nZE^n\Lambda^n|\Omega_{\rm acc}} - \chi_O \otimes \rho_{X^nZE^n\Lambda^n|\Omega_{\rm acc}}\right\|_1 \le \varepsilon_{\rm sec}. \tag{135}$$

*Here $\chi_O$ is the uniform state on the final output alphabet.*

*Proof.* We decompose the accepted branch into the good statistical event and its complement:

$$\Omega_{\rm acc} = (\Omega_{\rm acc} \cap \Omega_{\rm stat}) \cup (\Omega_{\rm acc} \cap \Omega^c_{\rm stat}).$$

By the choice of the statistical margin,

$$\Pr[\Omega_{\rm acc} \cap \Omega^c_{\rm stat}] \le \Pr[\Omega^c_{\rm stat}] \le \varepsilon_{\rm stat}. \tag{136}$$

Thus the bad statistical branch contributes at most $\varepsilon_{\rm stat}$ to the final secrecy error.

It remains to prove secrecy on the good statistical branch. On $\Omega_{\rm stat}$, the underlying single-round realization is contained in the feasible set defining $h$. For the $n$-round state $\rho^{\otimes n}_{BXE\Lambda}$, Eq. (130) and Theorem 8 give

$$\begin{aligned} H^{\varepsilon_{\rm sm}}_{\min}(B^n|X^nE^n\Lambda^n)_{\rho^{\otimes n}} &\ge n\,h - \Delta_{\rm AEP}(n, \varepsilon_{\rm sm}, |B|) \\ &= n\,H(\hat{I}^{\omega}_{\varepsilon} + \mu_I) - \Delta_{\rm AEP}(n, \varepsilon_{\rm sm}, |B|) = k_B. \end{aligned} \tag{137}$$

By the extractor assumption, the final output on the branch satisfies

$$\frac{1}{2}\Pr[\Omega_{\rm acc} \cap \Omega_{\rm stat}]\left\|\rho_{OX^nZE^n\Lambda^n|\Omega_{\rm acc}\cap\Omega_{\rm stat}} - \chi_O \otimes \rho_{X^nZE^n\Lambda^n|\Omega_{\rm acc}\cap\Omega_{\rm stat}}\right\|_1 \le \varepsilon_{\rm sm} + \varepsilon_{\rm ext}. \tag{138}$$

where the $\varepsilon_{\rm ext}$ is the extractor parameter (see Supplementary Note 6). Combining the good branch with the bad statistical branch, and using the triangle inequality for trace distance, gives

$$\frac{1}{2}\Pr[\Omega_{\rm acc}]\left\|\rho_{OX^nZE^n\Lambda^n|\Omega_{\rm acc}} - \chi_O \otimes \rho_{X^nZE^n\Lambda^n|\Omega_{\rm acc}}\right\|_1 \le \varepsilon_{\rm stat} + \varepsilon_{\rm sm} + \varepsilon_{\rm ext}. \tag{139}$$

This proves the finite-size security statement. □

The theorem uses the extractor guarantee stated in the protocol section. The seed $Z$ is produced from a weak source and is not assumed to be perfectly uniform and independent. The extractor used in the final step must therefore be compatible with the weak-source assumptions specified in the next section.

## SUPPLEMENTARY NOTE 6: QUANTUM-PROOF STRONG TWO-SOURCE EXTRACTION

The IID attack model analysis in Supplementary Note 5 gives a certified raw entropy bound $H^{\varepsilon^B_{\rm sm}}_{\min}(B^n|X^n, E^n, \Lambda^n) \ge k_B$. The device output $B^n$ contains private randomness, but it is not yet a nearly uniform string. The final step is therefore an extraction step. Since this is a randomness amplification protocol, the seed used for extraction is itself produced by a weak source and is weakly uniform. We therefore use a quantum-proof strong two-source extractor, as illustrated in Fig. 4. The first source is the raw device output $B^n$, and the second source is a weak seed block $Z^{n_2}$. The extractor outputs $O = {\rm Ext}(B^{n_1=n}, Z^{n_2})$, which should be close to uniform even when the seed and Eve's side information are revealed.

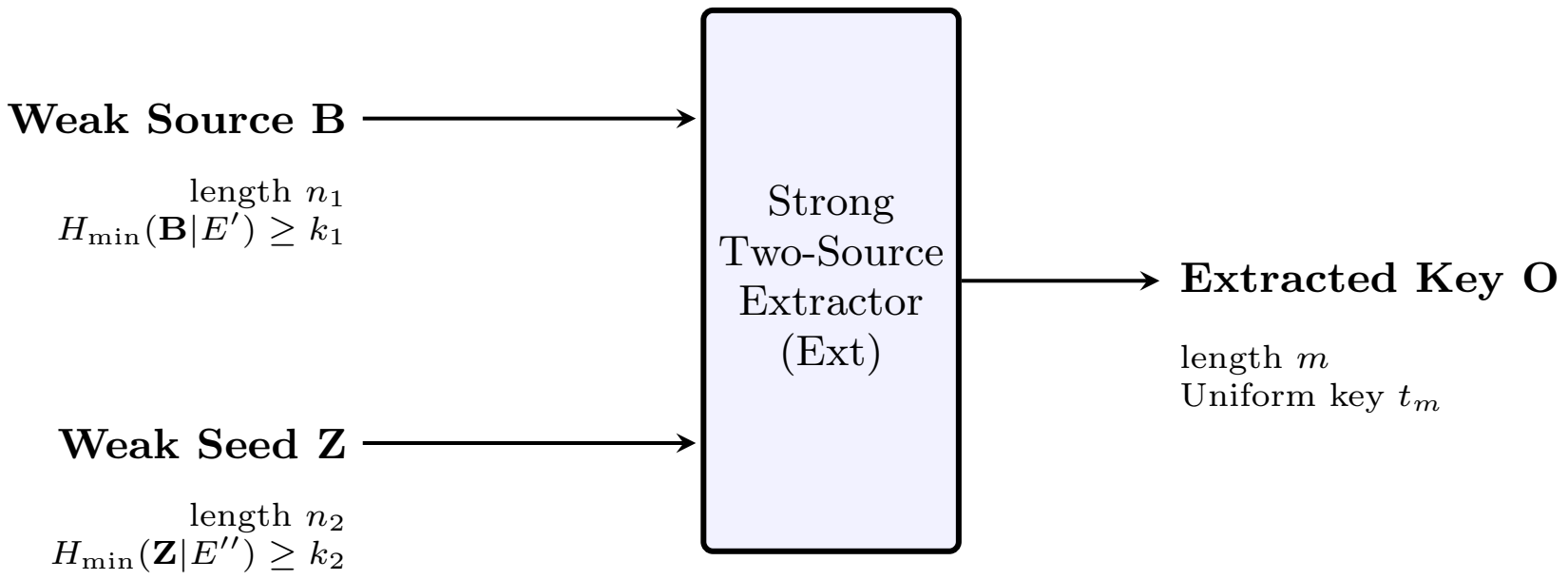


FIG. 4: Illustration of a strong quantum-proof two-source randomness extractor. Two independent weak input sources, the primary raw string **B** and the weak seed **Z**, are fed into the extractor. By utilizing a two-universal hashing primitive, such as the computationally efficient Toeplitz construction, the function distils a universally composable uniform key **O** of length $m$. The conceptual trace distance bound below highlights that security holds even if the weak seed **Z** is revealed.

**Definition 2** (Quantum-proof strong two-source extractor)**.** *A function* Ext : $\{0,1\}^{n_1} \times \{0,1\}^{n_2} \longrightarrow \{0,1\}^m$ *is a quantum-proof strong* $(n_1, k_1, n_2, k_2, m, \varepsilon_{\text{ext}})$*-two-source extractor if the following holds. For any two sources* $B^{n_1}$ *and* $Z^{n_2}$ *satisfying the required product-source or Markov-source condition, and satisfying* $H_{\min}(B^{n_1}|E_B) \geq k_1,\ H_{\min}(Z^{n_2}|E_Z) \geq k_2$, *the output* $O = \text{Ext}(B^{n_1}, Z^{n_2})$ *obeys*

$$\frac{1}{2}\left\|\rho_{OZ^{n_2}E} - \chi_O \otimes \rho_{Z^{n_2}E}\right\|_1 \leq \varepsilon_{\text{ext}}. \tag{140}$$

*Here* $\chi_O$ *is the uniform state on* $m$ *bits, and* $E$ *denotes the full side information available to the adversary. The extractor is strong in the second source, meaning that security still holds after the weak seed* $Z^{n_2}$ *is revealed.*

In our protocol, the first source is the raw output $B^n$, and the side information is $E_B = (X^n, E^n, \Lambda^n)$. Thus we set $n_1 = n,\ k_1 = k_B$. The second source is the weak seed block $Z^{n_2}$, whose min-entropy threshold will be denoted by $k_2$. The product-source or Markov-source condition is the independence assumption required by the extractor theorem. It is separate from the SDI entropy certification of $B^n$.

The entropy bound is smoothed, while the extractor definition above is stated for exact min-entropy. The next lemma records the standard conversion from exact to smooth entropy. This is the form used in the final security parameter.

**Lemma 8** (Smoothed two-source extraction)**.** *Assume that* Ext *is a quantum-proof strong two-source extractor as in Definition 2. Suppose that the two sources satisfy* $H_{\min}^{\varepsilon_{\text{sm}}^B}(B^{n_1}|E_B) \geq k_1,\ H_{\min}^{\varepsilon_{\text{sm}}^Z}(Z^{n_2}|E_Z) \geq k_2$. *Assume that the nearby states used in the smoothing can be chosen within the Markov-source model required by the extractor. Then the extracted output satisfies*

$$\frac{1}{2}\left\|\rho_{OZ^{n_2}E} - \chi_O \otimes \rho_{Z^{n_2}E}\right\|_1 \leq \varepsilon_{\text{ext}} + 2\varepsilon_{\text{sm}}^B + 2\varepsilon_{\text{sm}}^Z. \tag{141}$$

*Proof.* Let $D(\rho,\sigma) = \frac{1}{2}\|\rho - \sigma\|_1$ denote the trace distance. By the definition of smooth min-entropy, the first condition implies that there is a state $\widetilde{\rho}^B$ within distance $\varepsilon_{\text{sm}}^B$ of the actual state such that $H_{\min}(B^{n_1}|E_B)_{\widetilde{\rho}^B} \geq k_1$. The second condition gives a state $\widetilde{\rho}^Z$ within distance $\varepsilon_{\text{sm}}^Z$ such that $H_{\min}(Z^{n_2}|E_Z)_{\widetilde{\rho}^Z} \geq k_2$. Under the stated compatibility assumption for the smoothing, we take a joint nearby state $\widetilde{\rho}$ satisfying the exact entropy conditions and $D(\rho,\widetilde{\rho}) \leq \varepsilon_{\text{sm}}^B + \varepsilon_{\text{sm}}^Z$.

Let $\mathcal{E}$ be the CPTP map that applies the extractor and keeps $Z^{n_2}E$. Let $\mathcal{U}$ be the ideal map that replaces the output register by the uniform state $\chi_O$ while keeping $Z^{n_2}E$. By the triangle inequality,

$$D(\mathcal{E}(\rho),\mathcal{U}(\rho)) \leq D(\mathcal{E}(\rho),\mathcal{E}(\widetilde{\rho})) + D(\mathcal{E}(\widetilde{\rho}),\mathcal{U}(\widetilde{\rho})) + D(\mathcal{U}(\widetilde{\rho}),\mathcal{U}(\rho)). \tag{142}$$

The first and third terms are bounded by $D(\rho,\widetilde{\rho})$ because trace distance is contractive under CPTP maps. The middle term is bounded by $\varepsilon_{\text{ext}}$ by the exact extractor guarantee in Definition 2. Hence

$$D(\mathcal{E}(\rho),\mathcal{U}(\rho)) \leq 2D(\rho,\widetilde{\rho}) + \varepsilon_{\text{ext}} \leq \varepsilon_{\text{ext}} + 2\varepsilon_{\text{sm}}^B + 2\varepsilon_{\text{sm}}^Z. \tag{143}$$

□

We now specify the entropy available from the weak seed used in the extractor. Let the seed block $Z^{n_2}$ be generated by an $\varepsilon_{\mathrm{SV}}$-biased Santha–Vazirani source. That is, for every $i$, every history $z_{<i}$, every value $z_i \in \{0,1\}$, and every classical side information value $\lambda$, we assume $\frac{1}{2}-\varepsilon \leq \Pr(Z_i = z_i \mid Z_{<i} = z_{<i}, \Lambda_Z = \lambda) \leq \frac{1}{2}+\varepsilon$. Hence, for any length-$n_2$ string $z^{n_2}$,

$$\Pr(Z^{n_2} = z^{n_2} \mid \Lambda_Z = \lambda) = \prod_{i=1}^{n_2} \Pr(Z_i = z_i \mid Z_{<i} = z_{<i}, \Lambda_Z = \lambda) \leq \left(\frac{1}{2} + \varepsilon_{\mathrm{SV}}\right)^{n_2}. \tag{144}$$

Therefore the conditional min-entropy of the seed block satisfies

$$H_{\min}(Z^{n_2} \mid \Lambda_Z) = -\log_2 \max_{z^{n_2},\lambda} \Pr(Z^{n_2} = z^{n_2} \mid \Lambda_Z = \lambda) \geq -n_2 \log_2 \left(\frac{1}{2} + \varepsilon_{\mathrm{SV}}\right) = n_2 h_{\mathrm{SV}}, \tag{145}$$

where $h_{\mathrm{SV}} = \log_2 \frac{1}{1/2+\varepsilon_{\mathrm{SV}}}$.

For the final extraction step, we use a Trevisan-type extractor rather than a Toeplitz extractor as the main weak-seed construction. The reason is that the seed in our protocol is itself weak. Usually, a Toeplitz construction requires a seed whose length is linear in the raw input length, whereas the Trevisan construction uses a seed of only logarithmic or polylogarithmic length.

**Theorem 10** (Trevisan extraction [12])**.** *Consider the Trevisan extractor with input length $n_1$, seed length $d$, and output length $m$. In our protocol we set $n_1 = n$, $k_1 = k_B$, $n_2 = d$. Under the quantum-proof Markov-source model, the output length $m$ is admissible whenever*

$$m \leq \frac{1}{10}\left(k_B + 4(k_2 - d) - 4\log_2 m + 8\log_2 \varepsilon_{\mathrm{ext}} + 9\log_2 \frac{4}{3} - 6\right). \tag{146}$$

The term $4(k_2 - d))$ is the entropy penalty caused by th weakness of the seed. This penalty is small for Trevisan extraction because $d$ is only logarithmic or polylogarithmic in the raw input length. This is the main reason why Trevisan remains useful for more biased SV sources.

For the concrete implementation used in our numerical estimates, the Trevisan seed length is chosen from a weak-design construction. One can take $d = at^2$, where $t$ is a prime satisfying $t \geq l$, with

$$l = 2\left\lceil \log_2 n_1 + 2\log_2\left(\frac{2m}{\varepsilon_{\mathrm{ext}}}\right)\right\rceil, \qquad a = \max\left\{1, \left\lceil \frac{\log_2(m-2\mathrm{e}) - \log_2(t-2\mathrm{e})}{\log_2(2\mathrm{e}) - \log_2(2\mathrm{e}-1)}\right\rceil\right\}. \tag{147}$$

Since $d$ depends on $m$ through the weak design, the certified output length is obtained by a one-dimensional search: for each candidate $m$, compute $d(m)$, set $k_2 = d(m)h_{\mathrm{SV}}$, and check the Trevisan condition above. The largest integer satisfying this condition is the certified output length.

## SUPPLEMENTARY NOTE 7: THE INTEGRATED E/O PLATFORM

We build an experimental platform based on the implemented QRA chip, as shown in the main text. A high-performance narrow-linewidth (2 kHz) continuous-wave laser at 1550 nm is used, with an internal optical isolator. After passing through a variable optical attenuator (VOA) into the chip, the continuous-wave is split again by an on-chip programmable splitter, with most power serving as the LO and a small portion supporting the quantum state preparation. To physically enforce the energy constraint required by the SDI framework, the QRA chip integrates a VOA based on Mach-Zehnder interferometric structure. It fixes the mean photon number $\bar{n} = |\alpha|^2$. It is calibrated with high precision via carrier-injection phase modulation in silicon rib waveguides. A 90:10 beam splitter taps a small portion of the prepared quantum state. It sends the light to a real-time power monitor. This provides an upper bound of $\mathrm{Tr}(\hat{n}\rho_x)$ for monitoring and calibration. An isolator (ISO) is then used to strictly block back-reflected light from entering the QRA chip. With proper polarization(PC) control, the polarization of the energy-constrained quantum state is aligned with that of the LO. This ensures perfect spatial mode matching. With proper polarization(PC) control, the polarization of the energy-constrained quantum state is aligned with that of the LO. It enables stable interference and maximal extraction of measurement information. Strict optical isolation and crosstalk suppression are essential to prevent photon leakage from the LO from violating the semi-DI assumption. A multi-level strategy achieves more than 40 dB on-chip isolation. This includes absorption and guard structures at the process level. It includes high-directivity splitter networks with over 30 dB isolation at the waveguide level, and balanced receivers with over 50 dB common-mode rejection, combined with heterodyne spectral separation at

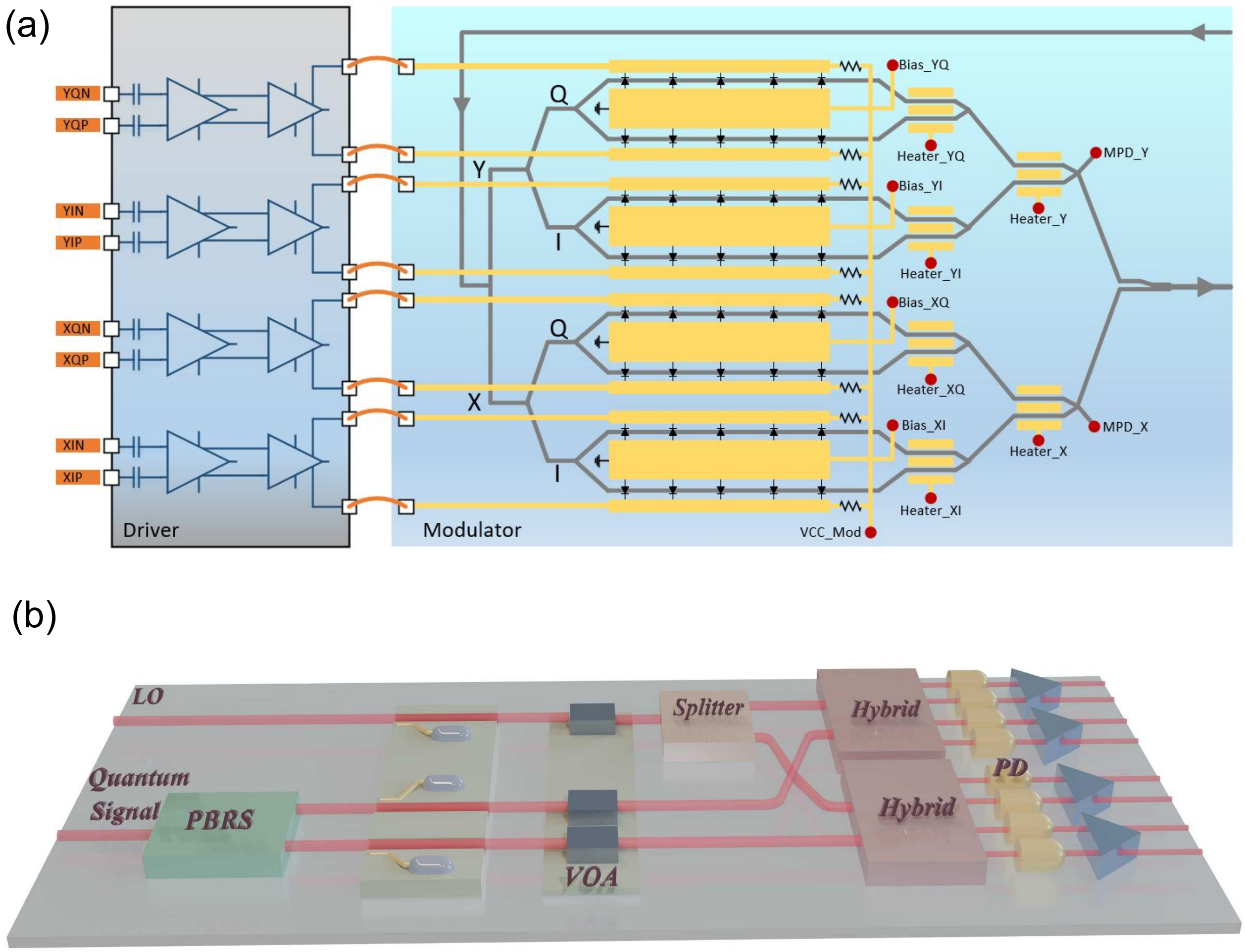


FIG. 5: **The quantum state preparation and measurement structure of the integrated I/O platform.** *(a) An amplitude-phase modulator based on the Mach-Zehnder interferometer structure on the chip.* It shows a cascaded Mach–Zehnder interferometer (MZI)-based amplitude–phase modulator operating on two orthogonal polarization directions, $X$ and $Y$, also known as an IQ modulator. On the left, the RF signal is amplified by a driver and delivered to the $I$ and $Q$ control paths of the $X$- and $Y$-polarized branches. The initial RF signal is generated and determined by the SV source. The Heater_$X(I/Q)$/Heater_$Y$ $(I/Q)$ units on each branch are thermo-optic phase shifters used to set the initial phase. A tap coupler and monitoring photodiodes (MPD_$X$ and MPD_$Y$) are then employed for optical power monitoring. Finally, the modulated signals from all branches are combined by an optical combiner, producing the transmitted output signal. (b) *The on-chip heterodyne detection architecture.* The local oscillator (LO) is split into two paths by a 50:50 optical coupler and then delivered to two 90° optical hybrids corresponding to the two orthogonal polarization directions, where it interferes with the quantum-state signal for heterodyne detection. After being coupled onto the chip, the quantum-state signal is divided into two components by a 50:50 polarization beam splitter and routed to the two 90° optical hybrids associated with the orthogonal polarization directions for mixing with the LO and heterodyne measurement. The mixed optical signals are subsequently detected by eight photodiodes (PDs), generating photocurrents that are amplified by transimpedance amplifiers (TIAs). The resulting RF signals are then transmitted to an oscilloscope for observation and data acquisition.

the system level. Thermal crosstalk is suppressed below the phase-noise floor. The power for a $\pi$ phase shift is limited to 6 mW by optimized heater design. An automatic bias control system uses low-frequency pilot tones and digital lock-in amplification. It locks the Mach–Zehnder modulators shown in Fig. 5(a) at the NULL and quadrature points. Meanwhile, localized RF confinement and electro-magnetic interference suppression enable a 45 GHz 3 dB heterodyne detector electro-optic bandwidth shown in Fig. 5(b). The continuous heterodyne measurement current is

recorded by a high-speed oscilloscope. It is digitized and processed for high-precision phase compensation. And then it is binarized with a fixed threshold to generate the output bit $b \in \{0, 1\}$. Finally, the measurement results from each round are fed into the entropy estimation computational framework and Trevisan extractors. This produces the final randomness amplification.

Realizing such a optical chip is challenging as it must integrate manipulation and measurement together while achieving ultra-low loss, suppressed free-carrier absorption, and high phase-modulation efficiency to ensure the energy bound works. We address these bottlenecks through coordinated device design and System-in-Package (SiP) co-design, as detailed below.

To achieve high-fidelity coherent-state control at low energy levels, we first address thermo-optic instability, as thermally induced phase noise blurs state distinguishability. The high thermo-optic coefficient of silicon causes phase sensitivity.The buried oxide in the waveguide has low thermal conductivity. Both cause even sub-millikelvin thermo-drift to induce phase errors, appearing as state jitter in phase-space region with energy constraints. We use selective-area epitaxy with thermal-sink co-design. The thermal time constant is reduced to sub-microsecond. Shot-noise-limited measurement is maintained at sub-photon levels. This differs from post-fabrication thermal tuning[13] and active thermal stabilization[14]. Here, thermal management is built into the material and structure. Thermal drift is suppressed at the source to ensure high-fidelity coherent-state control. Besides, to preserve quantum state purity under strict energy constraints, both low back-reflection and high coupling efficiency are required, as any reflection or loss degrades purity and randomness. Bare-fiber edge couplers achieve $> 80\%$ efficiency with standard telecom fibers via inverse-taper converters, suppressing back-reflection below $-45$dB and limiting fiber–chip loss to $< 0.97$dB per facet.

Carrier-induced nonlinearities are another obstacle. Two-photon absorption generates free carriers. They induce refractive-index blue shift and excess loss via the plasma-dispersion effect. At low photon numbers, this degrades manipulaiton fidelity. This directly distorts photon statistics, which leads to bias in min-entropy estimation and compromises security. Through atomic-level engineering of arsenic and boron doping profiles, we compress the free-carrier lifetime to the sub-nanosecond regime. Free-carrier absorption is reduced below 0.1dB/cm. Meanwhile, quantum state manipulation efficiency in phase-space exceeds $5\text{rad}/(\text{V} \cdot \text{mm})$ . Unlike prior silicon-rich nitride waveguide approaches [15], our method is compatible with CMOS photonic platforms. It addresses carrier-induced nonlinearities in silicon [16, 17]. Nonlinearity is suppressed while maintaining high state manipulation efficiency.

Sidewall scattering in high-index-contrast waveguides is the main source of on-chip propagation loss. The 220nm silicon device layer on a $2\mu$m buried oxide provides strong optical confinement. But nanoscale sidewall roughness induces Rayleigh scattering. At low mean photon numbers, this couples manipulated coherent states to vacuum noise modes. It causes state position shift in phase-space and state-distribution broadening. And ultimately reduces state manipulation accuracy. We therefore employ deep-ultraviolet lithography with sidewall oxidation smoothing. Propagation loss is suppressed below 0.5dB/cm. It is comparable to ultra-low-loss quantum photonic circuits with integrated single quantum emitters [18] and unconventional silicon photonic platforms for quantum technologies [19, 20]. This shows that strong confinement and low propagation loss can be achieved simultaneously in high-index-contrast systems.

At ultralow photon numbers, the core challenge of heterodyne detection is resolving vacuum fluctuations in the presence of on-chip electronic noise. Under energy bounds, the manipulated quantum states approach vacuum. Electronic noise can mask the quantum states. This makes measurement results dominated by classical noise(mainly electronic noise). We address this by interfering the weak signal with a strong on-chip local oscillator (LO). This linearly amplifies the measured field quadratures of the manipulated quantum state. We achieve shot-noise-limited detection to obtain shot-noise clearance. Quantum noise is lifted above the electronic noise floor. This ensures results are dominated by quantum signals. Through high-speed substrate interconnects and co-packaging of transimpedance amplifiers with photodiodes, we suppress RF crosstalk and parasitic capacitance below $-40$dB. This enables fully integrated shot-noise-limited heterodyne detection at sub-photon levels. This differs from prior homodyne systems for quantum measurement[21, 22] and photon-counting platforms[23]. These methods focus on coherent-state or entanglement detection. They do not address optimized noise suppression at low photon numbers.

Finally, to ensure that the previously established state manipulaiton and measurement are not compromised by side channels, and to physically enforce the security assumptions of the SDI model, we isolate all potential side-channel leakage at the device level. We adopt a 600-ball BGA-encapsulated SiP approach. The monolithic silicon photonic core, quad-channel modulator drivers, and linear transimpedance amplifiers are enclosed in a low-parasitic electromagnetic cavity. Measured channel-to-channel isolation is better than $-40$ dB. Compared to prior works, including the roadmap [24] that remains at the visionary stage, our solution co-integrates the silicon photonic core with electronic components in a low-parasitic electromagnetic cavity. This achieves superior channel isolation. This significantly enhances the security, integration, and interference robustness of the randomness amplification chip.

To enable high-fidelity sub-photon on-chip detection, silicon platform limitations must be overcome to ensure high-bandwidth, low-noise detection. We employ germanium-on-silicon (Ge-on-Si) photodetectors to implemente

high-performance coherent receivers on silicon, which offer responsivities $> 0.8A/W$ at 1550 nm and dark currents$< 10nA$, grown epitaxially on the SOI platform to mitigate lattice mismatch-induced defects via selective area growth and rapid thermal annealing. This enables four pairs of high-speed balanced photodiodes with 3 dB bandwidths exceeding 52 GHz, supporting simultaneous resolution of all quadratures in a dual-polarization coherent receiver configuration comprising two 90° optical hybrids. The overall system detection efficiency reaches nearly 40%, limited primarily by waveguide losses and coupler imperfections, but optimized through inverse taper designs at the edge couplers and grating-assisted polarization rotators. These optimizations draw from recent mechanism-level insights into defect passivation and strain engineering in Ge-on-Si heterostructures, reducing dislocation densities to $< 10^6\,\mathrm{cm}^{-2}$ and enabling quantum-limited noise performance for weak-coherent state detection.

To prevent phase-space distribution distortion and quadrature cross-talk in energy-constrained quantum states, which degrade system security, photonic and electronic designs enabling distortion-free, stable, and accurate high-speed quantum state manipulation are essential. The detailed structure enabling phase-space manipulation is illustrated in Fig. 5(a). Taking one polarization degree of freedom (e.g., X polarization) as an example, the external laser is split by waveguides and Y-branch structures into two parallel Mach–Zehnder interferometer (MZI) modulators. Each MZI contains a pair of push–pull electrodes for high-frequency modulation, while thermo-optic phase shifters control the relative phase between the interferometer arms. The MZIs are biased at the minimum optical power (NULL) operating point to suppress residual amplitude modulation, achieving extinction ratios $> 22$ dB and arm skews $< 3$ ps through careful doping profile engineering in the $p$–$n$ junctions, which enhances carrier depletion efficiency and reduces velocity mismatch in the traveling-wave electrode design. A thermo-optic phase shifter at the recombination stage sets the relative phase between the two MZI outputs to $\pi/2$, thereby realizing standard IQ modulation. An on-chip waveguide-integrated photodiode monitors a tapped fraction of the optical signal to track the operating point of the IQ modulator in real time. A polarization rotator–combiner then rotates one optical path by 90° and recombines it with the orthogonal polarization before waveguide output. This structure enables coherent manipulation of quantum states across multiple polarization degrees of freedom within a compact footprint.

For the specific experimental setups, as illustrated in the optical setup, a 1550 nm narrow-linewidth laser with kHz-level linewidth is used as the optical source. The laser is an NKT Photonics Koheras BASIK X15 system. The optical field is first attenuated by a variable optical attenuator to reach proper optical power required for the preparation of energy-constrained quantum states. The attenuated optical field is then injected into the QRA chip, where two nearly indistinguishable energy-bounded weak coherent states are prepared. The phase and amplitude modulation required for the state preparation are driven by a amplified RF signal from an arbitrary waveform generator of the Tektronix AWG 5000 series. In our experiment, the target energy parameters are chosen such that $w_0, w_1 \in [0.01, 0.02]$, corresponding to mean photon numbers of the same order under single-quadrature SNU normalization. At a modulation rate of 1 GHz and wavelength 1550 nm, this energy range corresponds to an average optical power of approximately $-88.92$ dBm to $-85.91$ dBm before the monitoring beam splitter. The prepared weak coherent states are subsequently sent to a 90:10 beam splitter, where 10% of the optical power is tapped out for real-time monitoring of the energy constraint. Due to the extremely low optical power level, direct measurement with a conventional optical power meter is not feasible. Therefore, the signal power is calibrated at a higher power level before attenuation, and the final power is inferred from the accurately characterized attenuation of the VOA. The optical components used in the fiber paths, including beam splitters, optical isolators, and polarization controllers, are commercially sourced from Thorlabs. After the monitoring stage, the transmitted optical states pass through an optical isolator ISO to suppress backward scattering and ensure stable propagation of the weak coherent states. The states are then adjusted by a polarization controller PC to achieve spatial-mode and polarization-mode alignment with the local oscillator. This alignment is necessary to obtain high-visibility interference in the subsequent measurement stage. The prepared states are then injected into the on-chip measurement module of the QRA system where they interfere with a strong local oscillator LO. In the present setup, the LO power is fixed at 13.68 dBm, corresponding to about 23.3 mW, which provides a strong phase reference field and ensures shot-noise-limited heterodyne readout, enabling reliable extraction of the final quantum measurement outcomes. The heterodyne signals are detected by a high-speed photodetection and acquisition system based on the WaveMaster 813Zi-B platform, enabling real-time recording and analysis of the measurement outcomes.

We performed systematic measurements for different values of the SV source bias parameter $\varepsilon$ to experimentally evaluate the violation condition. In our experiment, $\varepsilon$ determines the asymmetry of the input distribution and directly influences the theoretical bound of the inequality. According to the theoretical analysis, the experimentally accessible violation region is chosen by setting $\varepsilon$ in the interval [0.05, 0.20]. It is noting that when $\varepsilon$ exceeds 0.20, the theoretical bound becomes increasingly restrictive and the violation region shrinks. Therefore, we focus on the experimentally relevant regime $\varepsilon \in [0.05, 0.20]$. For each value of $\varepsilon$, weak coherent states were prepared such that

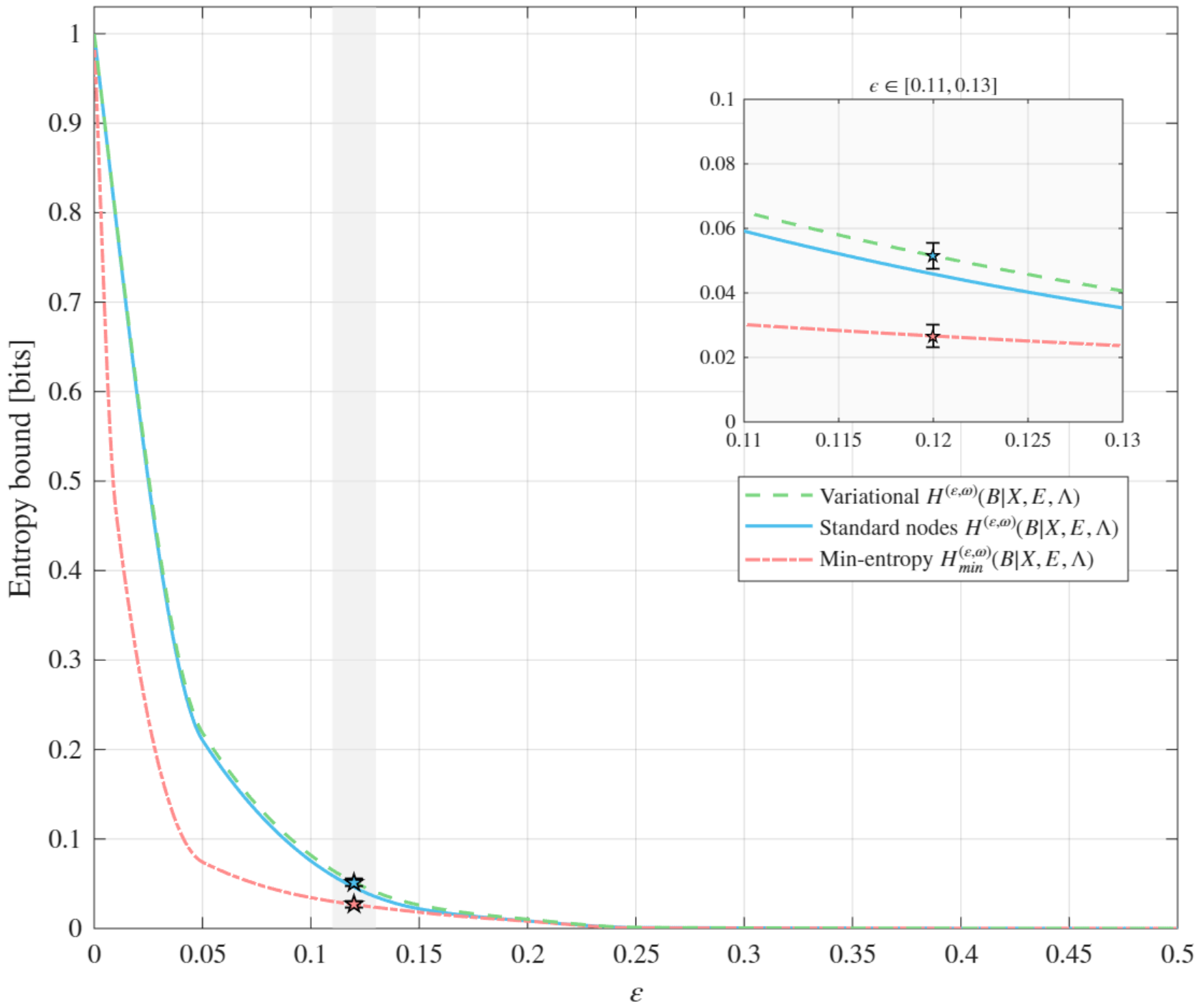


FIG. 6: Certified entropies in the relaxed entanglement model. The curves show lower bounds on the conditional von Neumann entropy as a function of the SV-source bias, evaluated under the experimental parameters of our setup when entanglement between the preparation and measurement devices is allowed. The bounds are obtained from the variational Gauss–Radau method (green dashed line), standard-node Gauss–Radau method (blue solid line), and the conditional min-entropy bound (salmon dash-dot line). The inset provides a magnified view of the low-entropy region $\varepsilon \in [0.11, 0.13]$. The star markers indicate the certified values at the representative bias $\varepsilon = 0.12$.

the corresponding energy parameters satisfy $w_0, w_1 \approx 0.01$–$0.02$. This energy range is chosen because it corresponds to the theoretically predicted violation region and can be reliably realized in the weak coherent state regime of our optical system. In practical terms, this range corresponds to mean photon numbers on the order of $10^{-2}$ per pulse. For each parameter configuration, $10^9$ optical pulses were collected in order to obtain statistically stable probability estimates. In addition, the experiment was repeated 50 groups under identical preparation conditions to verify the reproducibility of the state preparation and measurement processes. From these repeated experimental realizations, we selected eight representative group runs whose measured parameters fall within the stable preparation region, which is typical realizations. As shown in Table II, it reports these representative experimental realizations of the prepared parameters $(w_0, w_1)$ together with the corresponding observed probabilities $p(b = x)$. The consistency of the reported values across different runs demonstrates the stability of the experimental preparation process, while the large number of collected pulses ensures that statistical fluctuations remain negligible.

Table III shows the experimentally certified entropy under different modeling assumptions using the measured data at $\varepsilon = 0.12$. For this value of $\varepsilon$, the certified single-round min-entropy is approximately 0.03 bits as shown in Fig.6, while the corresponding variational nodes von Neumann entropy reaches 0.05 bits, giving an improvement of 66.7%. When the restrictive no-entanglement assumption between the preparation and measurement devices is relaxed and the proposed variational nodes framework is applied, the experimentally supported von Neumann entropy further increases to approximately 0.18 bits per round, compared with a corresponding min-entropy bound of 0.08 bits under the same conditions, yielding a relative enhancement of 125%. All entropy values are computed directly from experimentally acquired data, demonstrating the practical advantage of variational von Neumann entropy certification in energy-constrained semi-device-independent settings. Based on the experimentally certified entropy at $\varepsilon = 0.12$, the on-chip implementation operating at a 1 GHz modulation rate yields a peak real-time certified entropy

TABLE II: Typical experimental realizations of the energy parameters $w_0$ and $w_1$ for different bias values $\varepsilon$. The reported values of $p(b = x)$ correspond to experimentally observed typical probabilities under the theoretical violation condition.

| $\varepsilon$ | Run | $w_0$ | $w_1$ | $p(b = x)$ |
|---|---|---|---|---|
| 0.05 | 1 | 0.0102 | 0.0114 | 0.412 |
| 0.05 | 2 | 0.0111 | 0.0123 | 0.405 |
| 0.05 | 3 | 0.0120 | 0.0132 | 0.401 |
| 0.05 | 4 | 0.0128 | 0.0141 | 0.398 |
| 0.05 | 5 | 0.0137 | 0.0150 | 0.394 |
| 0.05 | 6 | 0.0145 | 0.0159 | 0.390 |
| 0.05 | 7 | 0.0154 | 0.0167 | 0.386 |
| 0.05 | 8 | 0.0162 | 0.0175 | 0.382 |
| 0.10 | 1 | 0.0105 | 0.0117 | 0.365 |
| 0.10 | 2 | 0.0114 | 0.0126 | 0.360 |
| 0.10 | 3 | 0.0123 | 0.0135 | 0.355 |
| 0.10 | 4 | 0.0132 | 0.0144 | 0.349 |
| 0.10 | 5 | 0.0140 | 0.0152 | 0.345 |
| 0.10 | 6 | 0.0149 | 0.0161 | 0.340 |
| 0.10 | 7 | 0.0157 | 0.0169 | 0.336 |
| 0.10 | 8 | 0.0166 | 0.0178 | 0.331 |
| 0.15 | 1 | 0.0108 | 0.0120 | 0.332 |
| 0.15 | 2 | 0.0117 | 0.0129 | 0.327 |
| 0.15 | 3 | 0.0126 | 0.0138 | 0.321 |
| 0.15 | 4 | 0.0135 | 0.0147 | 0.316 |
| 0.15 | 5 | 0.0144 | 0.0155 | 0.311 |
| 0.15 | 6 | 0.0152 | 0.0163 | 0.306 |
| 0.15 | 7 | 0.0161 | 0.0172 | 0.301 |
| 0.15 | 8 | 0.0169 | 0.0180 | 0.296 |
| 0.20 | 1 | 0.0111 | 0.0124 | 0.295 |
| 0.20 | 2 | 0.0120 | 0.0133 | 0.289 |
| 0.20 | 3 | 0.0129 | 0.0142 | 0.284 |
| 0.20 | 4 | 0.0138 | 0.0151 | 0.279 |
| 0.20 | 5 | 0.0146 | 0.0159 | 0.273 |
| 0.20 | 6 | 0.0155 | 0.0168 | 0.269 |
| 0.20 | 7 | 0.0163 | 0.0176 | 0.264 |
| 0.20 | 8 | 0.0172 | 0.0185 | 0.259 |

generation rate of approximately $1.8 \times 10^8$ bits/s before strong randomness extraction and finite-size post-processing. Extrapolating to the hardware limits of the platform, with a 52 GHz detection bandwidth and a 47 GHz modulation bandwidth, and conservatively taking the effective symbol rate to be limited by the smaller modulation bandwidth, the potential peak entropy generation rates under full-bandwidth operation become $0.03 \times 47 \times 10^9 \approx 1.4$ Gbps for conservative min-entropy certification, $0.05 \times 47 \times 10^9 = 2.35$ Gbps for standard node von Neumann entropy certification, and $0.18 \times 47 \times 10^9 \approx 8.5$ Gbps for the proposed variational von Neumann entropy certification, before accounting for extractor overhead and finite-size effects. In the implemented randomness amplification experiment, the certified extractable randomness after finite-size analysis is $m = 1\,797\,309\,120$ bits for a raw data block of $n = 10^{11}$ bits. The Trevisan extractor completed the extraction in approximately 20 s. Together with the experimental acquisition time of approximately 70 s, the complete procedure requires about 90 s, resulting in an end-to-end certified randomness generation rate of approximately 20 Mbps, where the details are shown in Suppplementary Note 8.

### 7.1 Electro-optical bandwidth and linearity characterization

To validate the high-speed capability of the monolithic QRA optical chip, we characterized the electro–optical and opto–electrical frequency responses of the key components. In a QRA system, the preparation and measurement of energy-constrained weak coherent states must be performed with high temporal fidelity. Any bandwidth limitation or nonlinear distortion in the electro–optic modulation or photodetection processes can introduce systematic errors in the prepared quantum states or in the extracted measurement statistics. Such distortions may directly affect the

TABLE III: Comparison of experimentally certified entropy at $\varepsilon = 0.12$ under different modeling assumptions.

| **Metric** | **With relaxation ($\mathcal{Q}^{\text{ent}}$)** | **Without relaxation ($\mathcal{Q}^{\text{sep}}$)** |
|---|---|---|
| Min-entropy $H_{\min}(B\|X, E, \Lambda)$ (bits/round) | 0.03 | 0.08 |
| Standard nodes entropy $H(B\|X, E, \Lambda)$ (bits/round) | 0.04 | 0.17 |
| Variational nodes entropy $H(B\|X, E, \Lambda)$ (bits/round) | 0.05 | 0.18 |
| Entropy gain over min-entropy | 66.7% | 125% |
| Potential rate from min-entropy | 1.4 Gbps | 3.76 Gbps |
| Potential rate from standard nodes von Neumann entropy | 1.88 Gbps | 7.99 Gbps |
| Potential rate from variational nodes von Neumann entropy | 2.35 Gbps | 8.5 Gbps |
| Certified rate at 1 GHz | 50 Mbps | 180 Mbps |
| IID-based extracted output length | $1.8 \times 10^9$ bits | |
| IID-based extraction throughput | 20 Mbps | |
| Soundness error | $\varepsilon_{\text{sec}} = 10^{-12}$ | |

accuracy of the estimated probability distributions used in randomness certification. Therefore, precise calibration and characterization of the electro–optical transfer function are essential to ensure that the generated optical states and the detected signals faithfully represent the intended quantum processes.

To evaluate the intrinsic device response, we performed small-signal frequency response measurements using a broadband lightwave component analyzer with a measurement bandwidth extending beyond tens of gigahertz. Prior to the measurement, the electrical and optical test paths were carefully calibrated. The RF cables, connectors, bias tees, and probes were calibrated using the instrument reference procedure so that the measurement reference plane was effectively shifted to the chip interface. On the optical side, the insertion loss of the fiber coupling path and passive optical components was independently characterized to maintain stable optical power during the frequency sweep. These calibration procedures remove systematic attenuation from external interconnections and ensure that the extracted S21 responses reflect the intrinsic electro–optic and opto–electronic performance of the integrated devices. Fig. 7(a) shows the measured electro–optical small-signal response of the silicon photonic modulator under different reverse-bias voltages. During the test, a continuous-wave optical carrier was injected into the chip while the modulator was driven by the RF output of the analyzer. The S21 response was normalized to the low-frequency response level to highlight the relative bandwidth roll-off. The measurement results demonstrate that the modulation bandwidth increases systematically with reverse bias. This behavior is attributed to the reduction of the junction capacitance in the carrier-depletion regime, which decreases the RC time constant of the modulator and extends its high-frequency response. Such bandwidth improvement is particularly important for QRA operation because the weak coherent states used for randomness generation are modulated at gigahertz rates, and insufficient modulation bandwidth would distort the amplitude and phase profiles of the prepared quantum states.

The opto–electrical response of the integrated photodetector was characterized using a similar procedure, as shown in Fig. 7(b). An intensity-modulated optical signal was injected into the detector and the electrical output was measured across the frequency range. The detector bias voltage was varied to evaluate its influence on the response bandwidth. The measured curves show that the photodetector bandwidth increases as the reverse bias increases. This trend is consistent with the enhanced carrier drift velocity and reduced carrier transit time under stronger electric fields. The improved bandwidth directly benefits the heterodyne detection stage of the QRA system, where accurate temporal reconstruction of the optical field quadratures is required to reliably extract the measurement outcomes. To ensure measurement reliability, each frequency sweep was repeated multiple times and the resulting traces were averaged to suppress random fluctuations. The 3-dB bandwidth was determined as the frequency at which the normalized response drops by 3 dB from the low-frequency level. The consistency between repeated measurements confirms that the observed bandwidth trends originate from the intrinsic device physics rather than measurement noise or transient fluctuations.

Overall, the calibrated EO and OE frequency responses demonstrate that both the modulator and the photodetector provide sufficient bandwidth to support gigahertz-level operation. This high-speed performance ensures that the prepared weak coherent states and the corresponding measurement signals remain effective to the intended quantum protocol, thereby enabling reliable implementation of the on-chip QRA system.

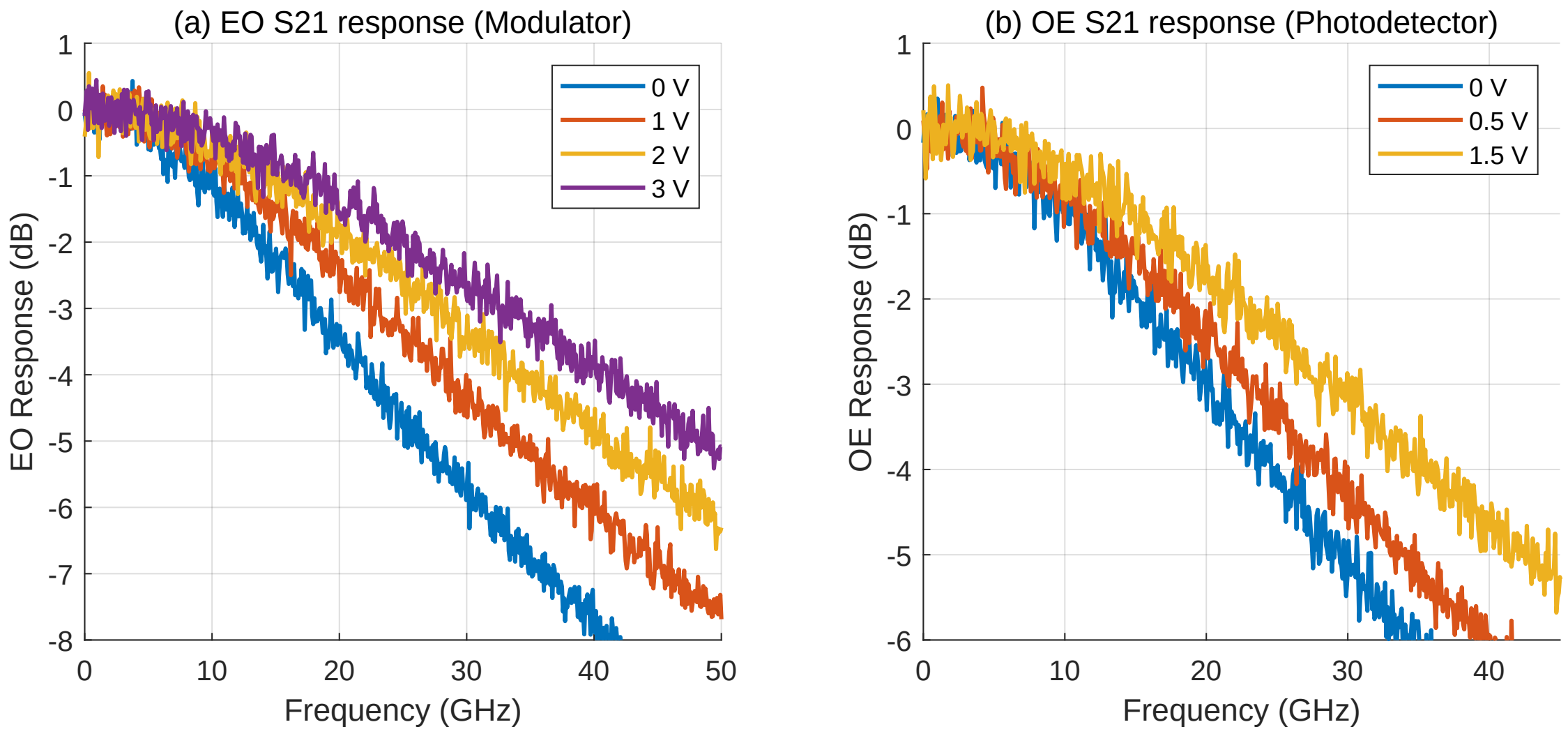


FIG. 7: Small-signal frequency response (S21) of the key integrated components. (a) Electro–optical(EO) response of the silicon photonic modulator measured under different reverse-bias voltages. (b) Opto–electrical(OE) response of the integrated photodetector measured under different reverse-bias voltages. The responses are normalized to the low-frequency level. Increasing reverse bias improves the high-frequency response of both devices.

### 7.2 System-level response

Building upon the component-level characterization, we further evaluated the end-to-end electro–optical (EO) and opto–electrical (OE) responses of the fully integrated transceiver. In the context of quantum randomness amplification (QRA), accurate preparation and measurement of weak coherent quantum states require that the entire transmitter–receiver chain preserves the temporal and spectral properties of the optical signals. Any bandwidth limitation in the modulation or detection stages may distort the amplitude or phase profile of the prepared states, which can directly affect the probability distributions used for randomness certification. Therefore, system-level frequency-response characterization is essential to verify that the integrated platform faithfully implements the intended quantum operations.

The system-level EO response was first characterized by injecting a continuous-wave optical carrier into the transmitter while driving the electrical input with a broadband RF signal. The resulting optical modulation was measured after propagation through the full modulation chain. As shown in Fig. 8(a), the EO frequency response exhibits an average 3-dB bandwidth of approximately 47 GHz. This bandwidth significantly exceeds the gigahertz modulation rate used in the QRA experiment, ensuring that the generated weak coherent states accurately follow the intended modulation waveform. As a result, the prepared quantum states maintain the required energy constraints and phase relationships imposed by the QRA protocol.

Correspondingly, the opto–electrical response of the coherent receiver front-end was characterized by injecting an intensity-modulated optical signal into the receiver and recording the electrical output spectrum. The measured OE response, shown in Fig. 8(b), demonstrates a 3-dB bandwidth of approximately 52 GHz. Such broadband detection capability guarantees that the heterodyne measurement process can accurately capture the temporal fluctuations of the incoming quantum states. This is particularly important for QRA because the certified randomness relies on the precise statistical distribution of the measurement outcomes.

Together, the EO and OE system-level measurements confirm that the integrated transceiver supports broadband operation well beyond the modulation rate required by the QRA experiment. This ensures that the prepared weak coherent states are generated and measured without bandwidth-induced distortion, thereby preserving the statistical integrity of the quantum measurement outcomes used for randomness certification.

### 7.3 Linearity characterization

High linearity in the electro–optical modulation chain is critical for preserving the accuracy of the prepared quantum states in the QRA protocol. In particular, the weak coherent states used in the experiment rely on precise amplitude

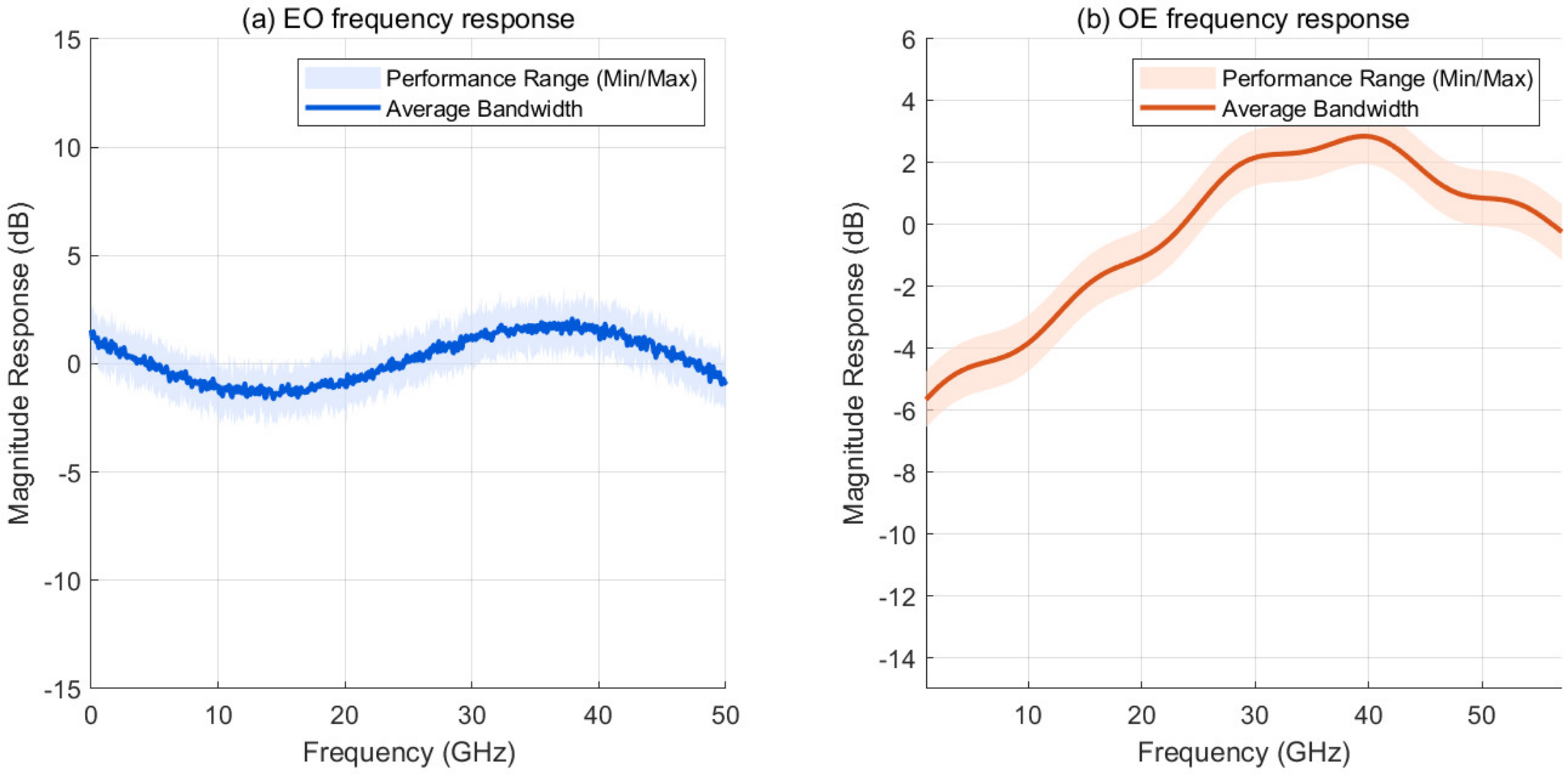


FIG. 8: Frequency response of the integrated E/O platform. (a) Electro–optical frequency response of the transmitter chain. (b) Opto–electrical frequency response of the coherent receiver front-end. The solid curves indicate the average system response while the shaded regions represent the performance variation observed across repeated measurements.

and phase modulation to satisfy the required energy constraints. Nonlinear distortion in the driver–modulator chain could alter the intended state parameters, thereby introducing systematic deviations in the probability distributions extracted from the measurement outcomes.

To evaluate the linear operating range of the transmitter, we characterized the linearity of the driver–modulator chain by measuring the output 1-dB compression point. During the measurement, the RF drive amplitude was gradually increased while monitoring the resulting optical modulation response. The measured results indicate an output 1-dB compression point of approximately 3.9 Vppd. This relatively large linear dynamic range ensures that the modulation amplitudes used in the QRA experiment remain well within the linear regime of the electro–optical conversion process. Consequently, the prepared weak coherent states accurately follow the intended modulation waveform, preventing distortion of the quantum state parameters. Such linear operation is essential for guaranteeing that the experimentally observed probability distributions originate from the designed quantum protocol rather than from classical nonlinear artifacts.

| Symbol | Quantity | Value |
|---|---|---|
| $\varepsilon$ | Experimental SV soure bias | 0.12 |
| $I_\varepsilon^\omega - \mathcal{B}_c$ | MDL violation | $-420.08$ |
| $\varepsilon_{\mathrm{sec}}$ | Security parameter | $10^{-12}$ |
| $\varepsilon_{\mathrm{sm}}$ | Smoothing parameter | $5 \cdot 10^{-13}$ |
| $\varepsilon_{\mathrm{sm}}^B, \varepsilon_{\mathrm{sm}}^Z$ | Smoothing parameters for the primary source and weak seed | $10^{-13}$ |
| $\varepsilon_{\mathrm{ext}}$ | Extractor error | $10^{-13}$ |
| $\varepsilon_{\mathrm{stat}}$ | Statistical check probability | $10^{-12}$ |
| $n$ | Number of protocol rounds | $10^{11}$ |
| $d$ | Seed length of extractor | $1.092 \cdot 10^7$ |
| Ext | The two-process extractor | Trevisan extractor (see Theorem. 10) |
| $m$ | Output length | $m = 1,797,309,120$ |

TABLE IV: Choices of parameter setting for entropy accumulation and extraction.

## SUPPLEMENTARY NOTE 8: ENTROPY ACCUMULATION AND RANDOMNESS EXTRACTION

The randomness extraction is implemented using a Trevisan extractor. In this experiment setting, as shown in Table IV, the raw data length is $n = 10^{11}$ bits and the final extracted random string length is $m = 1\,797\,309\,120$ bits, corresponding to an extraction efficiency of approximately $m/n \approx 0.018$ bits per raw bit. The Trevisan extractor is implemented using GPU acceleration on an NVIDIA GeForce RTX 4090 graphics card. For the present parameter setting, the extraction of the entire data block was completed in approximately 20s. The experimental data acquisition time is approximately 70 s, yielding a total runtime of about 90 s for the complete randomness amplification procedure. The corresponding end-to-end certified randomness generation rate is

$$R_{\text{overall}} = \frac{m}{70+20} = \frac{1\,797\,309\,120}{90} \approx 2.0 \times 10^7 \,\text{bits/s}.$$

corresponding to approximately 20Mbps. The extraction algorithm can also be future deployed on the associated electronic control board of the chip system, enabling practical integration of the randomness amplification module with the on-chip quantum hardware.


\twocolumngrid
\clearpage

\end{document}